\documentclass[preprint,showpacs,preprintnumbers,amsmath,amssymb]{revtex4}
\usepackage{amsmath}

\usepackage{graphicx}
\usepackage{pdfpages}
\usepackage{dcolumn}
\usepackage{bm}
\usepackage{amsmath} 
\usepackage{amsfonts}
\usepackage{amssymb}
\usepackage{url}
\usepackage{microtype}
\usepackage{graphicx}%

\newtheorem{theorem}{Theorem}[section]
\newtheorem{lemma}[theorem]{Lemma}
\newtheorem{proposition}[theorem]{Proposition}

\newenvironment{proof}[1][Proof]{\begin{trivlist}
\item[\hskip \labelsep {\bfseries #1}]}{\end{trivlist}}

\newcommand{\qed}{\nobreak \ifvmode \relax \else
      \ifdim\lastskip<1.5em \hskip-\lastskip
      \hskip1.5em plus0em minus0.5em \fi \nobreak
      \vrule height0.75em width0.5em depth0.25em\fi}
\newcommand{\norm}[1]{\left\lVert#1\right\rVert}

\begin{document}

\preprint{}
\title{Master Lovas-Andai and Equivalent Formulas Verifying the $\frac{8}{33}$ Two-Qubit Hilbert-Schmidt Separability Probability and Companion Rational-Valued Conjectures}
\author{Paul B. Slater}
 \email{slater@kitp.ucsb.edu}
\affiliation{%
Kavli Institute for Theoretical Physics, University of California, Santa Barbara, CA 93106-4030\\
}
\date{\today}
            
\begin{abstract}
We begin by investigating relationships between two forms of Hilbert-Schmidt two-rebit and two-qubit ``separability functions''--those recently advanced by Lovas and Andai ({\it J. Phys. A} {\bf{50}} [2017] 295303), and those earlier presented by Slater ({\it J. Phys. A} {\bf{40}} [2007] 14279). In the Lovas-Andai framework, the independent variable $\varepsilon  \in [0,1]$ is the ratio $\sigma(V)$ of the singular values of the $2 \times 2$ matrix  $V=D_2^{1/2} D_1^{-1/2}$ formed from the two $2 \times 2$ diagonal blocks ($D_1, D_2$) of a 
 $4 \times 4$ density matrix $D= \norm{\rho_{ij}}$. In the Slater setting, the independent variable $\mu$ is the diagonal-entry ratio $\sqrt{\frac{\rho_ {11} \rho_ {44}}{\rho_ {22} \rho_ {33}}}$--with, of central importance, $\mu=\varepsilon$ or $\mu=\frac{1}{\varepsilon}$ when both $D_1$ and  $D_2$ are themselves diagonal.
 Lovas and Andai  established that their two-rebit ``separability function'' $\tilde{\chi}_1 (\varepsilon )$ ($\approx \varepsilon$) yields the 
previously conjectured Hilbert-Schmidt separability probability of $\frac{29}{64}$. We are able, in the Slater framework (using cylindrical algebraic decompositions [CAD] to enforce positivity constraints), to reproduce this result. Further, we newly find its  two-qubit, two-quater[nionic]-bit and ``two-octo[nionic]-bit'' counterparts, $\tilde{\chi_2}(\varepsilon)  =\frac{1}{3} \varepsilon ^2 \left(4-\varepsilon ^2\right)$, 
$\tilde{\chi_4}(\varepsilon) =\frac{1}{35} \varepsilon ^4 \left(15 \varepsilon ^4-64 \varepsilon ^2+84\right)$ and $\tilde{\chi_8} (\varepsilon )= \frac{1}{1287}\varepsilon ^8 \left(1155 \varepsilon ^8-7680 \varepsilon ^6+20160 \varepsilon
   ^4-25088 \varepsilon ^2+12740\right)$. These immediately lead to predictions of  Hilbert-Schmidt separability/PPT-probabilities of $\frac{8}{33}$, $\frac{26}{323}$  and $\frac{44482}{4091349}$, in full agreement with those of the ``concise formula'' ({\it J. Phys. A} {\bf{46}} [2013] 445302), and, additionally, of a ``specialized induced measure'' formula. Then, we find a Lovas-Andai ``master formula'', $\tilde{\chi_d}(\varepsilon)= \frac{\varepsilon ^d \Gamma (d+1)^3 \,
   _3\tilde{F}_2\left(-\frac{d}{2},\frac{d}{2},d;\frac{d}{2}+1,\frac{3
   d}{2}+1;\varepsilon ^2\right)}{\Gamma \left(\frac{d}{2}+1\right)^2}$, encompassing both even and odd values of $d$. Remarkably, we are able to obtain the 
   $\tilde{\chi_d}(\varepsilon)$ formulas, $d=1,2,4$, applicable to full (9-, 15-, 27-) dimensional sets of density matrices, by analyzing (6-, 9, 15-) dimensional sets, with not only diagonal $D_1$ and $D_2$, but also an additional pair of nullified entries. Nullification of a further  pair still, leads to X-matrices, for which a distinctly different, simple Dyson-index phenomenon is noted. C. Koutschan, then, using his HolonomicFunctions program, develops an order-4 recurrence satisfied by the predictions of the several formulas,  establishing their equivalence. A two-qubit separability probability of $1-\frac{256}{27 \pi ^2}$ is obtained based on the {\it operator 
   monotone} function $\sqrt{x}$, with the use of $\tilde{\chi_2}(\varepsilon)$.
\end{abstract}


\maketitle
\tableofcontents
\section{Introduction and initial analyses}
To begin our investigations, focusing on interesting recent work of Lovas and Andai \cite{lovasandai}, we examined a certain possibility--motivated by a number of previous studies (e.g. \cite{slater833,slaterJGP2,JMP2008,ratios2}) and the apparent strong relevance there of the Dyson-index vantage upon random matrix theory \cite{dumitriu2002matrix}. More specifically, we ask whether the sought Lovas-Andai ``separability  function'' $\tilde{\chi}_2(\varepsilon)$  for the standard (complex) two-qubit systems might be simply proportional (or even  equal) to the square of their successfully constructed two-rebit separability function   \cite[eq. (9)]{lovasandai}, 
\begin{equation} \label{BasicFormula}
\tilde{\chi}_1 (\varepsilon ) = 1-\frac{4}{\pi^2}\int\limits_\varepsilon^1 
\left(
s+\frac{1}{s}-
\frac{1}{2}\left(s-\frac{1}{s}\right)^2\log \left(\frac{1+s}{1-s}\right)
\right)\frac{1}{s}
\mbox{d}  s 
\end{equation}
\begin{displaymath}
 = \frac{4}{\pi^2}\int\limits_0^\varepsilon
\left(
s+\frac{1}{s}-
\frac{1}{2}\left(s-\frac{1}{s}\right)^2\log \left(\frac{1+s}{1-s}\right)
\right)\frac{1}{s}
\mbox{d} s .
\end{displaymath}	
Let us note that 
$\tilde{\chi}_1 (\varepsilon )$ has a closed form,
\begin{equation} \label{poly}
\frac{2 \left(\varepsilon ^2 \left(4 \text{Li}_2(\varepsilon )-\text{Li}_2\left(\varepsilon
   ^2\right)\right)+\varepsilon ^4 \left(-\tanh ^{-1}(\varepsilon )\right)+\varepsilon ^3-\varepsilon
   +\tanh ^{-1}(\varepsilon )\right)}{\pi ^2 \varepsilon ^2},    
\end{equation}
where the polylogarithmic function is defined by the infinite sum
	\begin{equation*}
		\text{Li}_s (z) =
		\sum\limits_{k=1}^\infty 
		\frac{z^k}{k^s},
	\end{equation*}
for arbitrary complex $s$ and for all complex arguments $z$ with $|z|<1$.
Let us further observe that in the proof of (\ref{BasicFormula}), the authors were able to formulate the problem of finding $\tilde{\chi}_1(\varepsilon)$ 
rather concisely in terms of a ``defect function" \cite[App. A]{lovasandai},
\begin{equation} \label{Defect}
\Delta(\delta) =\frac{2 \pi^2}{3} - \tilde{\chi}_1 (e^{-\delta})=\frac{16}{3} \int_{0}^{\delta} \cosh(t) -\sinh(t)^2 t \log(\frac{e^t +1}{e^t-1}) \mbox{d}t.
\end{equation}

We will be able in sec.~\ref{Reproduction} to obtain the formula (\ref{poly})  for $\tilde{\chi}_1 (\varepsilon )$
by alternative (cylindrical algebraic decomposition \cite{strzebonski2016cylindrical}) means. Further, in sec.~\ref{Verification}, we will apply the same basic methodology to obtain (the much simpler) polynomial formula (\ref{VerifiedFormula}) for $\tilde{\chi}_2 (\varepsilon )$. Then, we will be able (sec.~\ref{GeneralConstruction}) to develop a general procedure for finding $\tilde{\chi}_d (\varepsilon )$ for integer $d>0$. The (rational) separability/PPT-probabilities predicted using these functions will--as extensive symbolic and numerical testing reveals--be identically the same as those yielded by the ``concise'' formula reported in \cite[eqs. (1)-(3)]{slaterJModPhys},
\begin{equation} \label{Hou1}
\mathcal{P}_{sep/PPT}(\alpha) =\Sigma_{i=0}^\infty f(\alpha+i),
\end{equation}
where
\begin{equation} \label{Hou2}
f(\alpha) = \mathcal{P}_{sep/PPT}(\alpha)-\mathcal{P}_{sep/PPT}(\alpha +1) = \frac{ q(\alpha) 2^{-4 \alpha -6} \Gamma{(3 \alpha +\frac{5}{2})} \Gamma{(5 \alpha +2})}{3 \Gamma{(\alpha +1)} \Gamma{(2 \alpha +3)} 
\Gamma{(5 \alpha +\frac{13}{2})}},
\end{equation}
and
\begin{equation} \label{Hou3}
q(\alpha) = 185000 \alpha ^5+779750 \alpha ^4+1289125 \alpha ^3+1042015 \alpha ^2+410694 \alpha +63000 = 
\end{equation}
\begin{displaymath}
\alpha  \bigg(5 \alpha  \Big(25 \alpha  \big(2 \alpha  (740 \alpha
   +3119)+10313\big)+208403\Big)+410694\bigg)+63000.
\end{displaymath}
(Here, $\alpha = \frac{d}{2}$. This set of relationships was developed with the [high-precision] use of a  probability distribution
reconstruction method \cite{Provost}, 
using formulas for the moments of the determinants of density matrices and of their partial transposes. This was followed by an application by  Qing-Hu Hou of
``Zeilberger's algorithm'' (``creative telescoping") \cite{paule1995mathematica} to the large hypergeometric-based expression so-obtained, displayed in Fig. 3 of  \cite{slaterJModPhys}. See also (\ref{InducedMeasureCase}), for a quite distinct, but--as will eventually be shown here--equivalent formula [Fig~\ref{fig:Order4Recurrence}].)

As part of their analysis, Lovas and Andai assert  \cite[p. 13]{lovasandai} that 
\begin{equation} \label{sepR}
\mathcal{P}_{sep}(\mathbb{R}) =    
\frac{\int\limits_{-1}^1\int\limits_{-1}^x  \tilde{\chi}_1 \left(
\left.\sqrt{\frac{1-x}{1+x}}\right/ \sqrt{\frac{1-y}{1+y}}	
	\right)(1-x^2)(1-y^2) (x-y) \mbox{d} y\mbox{d} x}{\int\limits_{-1}^1\int\limits_{-1}^x  (1-x^2)(1-y^2)(x-y)  \mbox{d} y \mbox{d} x},
\end{equation}
with the denominator evaluating to $\frac{16}{35}$.
Here, $\mathcal{P}_{sep}(\mathbb{R})$ is the Hilbert-Schmidt separability probability for the 
nine-dimensional convex set of two-rebit states \cite{carl}. With the indicated use of $\tilde{\chi}_1 (\varepsilon )$ this probability evaluates to $\frac{29}{64}$ 
(the numerator of (\ref{sepR}) equalling $\frac{16}{35} -\frac{1}{4}=\frac{29}{140}$, with $\frac{29}{64} =\frac{\frac{29}{140}}{\frac{16}{35}}$),
a result that had been strongly anticipated by prior analyses \cite{FeiJoynt2,slaterJModPhys,MomentBased}.

If the (Dyson-index) 
proportionality relationship
\begin{equation}
\tilde{\chi}_2 (\varepsilon ) \propto \tilde{\chi}_1^2 (\varepsilon )
\end{equation}
held, we would have 
\begin{equation} \label{sepC}
\mathcal{P}_{sep}(\mathbb{C}) \propto   
\frac{\int\limits_{-1}^1\int\limits_{-1}^x  \tilde{\chi}_1^2 \left(
\left.\sqrt{\frac{1-x}{1+x}}\right/ \sqrt{\frac{1-y}{1+y}}	
	\right)(1-x^2)^2(1-y^2)^2 (x-y)^2 \mbox{d} y\mbox{d} x}{\int\limits_{-1}^1\int\limits_{-1}^x  (1-x^2)^2(1-y^2)^2(x-y)^2  \mbox{d} y \mbox{d} x}.
\end{equation}
Here,  $\mathcal{P}_{sep}(\mathbb{C})$ is--in the Lovas-Andai framework--the Hilbert-Schmidt separability probability for the 
fifteen-dimensional convex set of the (standard/complex) two-qubit states \cite{Gamel}.
They expressed hope that  they too would be able to demonstrate that $\mathcal{P}_{sep}(\mathbb{C}) =\frac{8}{33}$, as also has been strongly
indicated is, in fact, the case \cite{FeiJoynt2,slaterJModPhys,MomentBased,shang2015monte,zhou2012topology}.
("It is interesting whether this observation has some deep background or is an accidental fact only" \cite{khvedelidze2015geometric}.)
We generalized (from $\alpha=\frac{1}{2}$) the denominator of the ratio (\ref{sepC})   to 
\begin{equation} \label{General}
\int\limits_{-1}^1\int\limits_{-1}^x  (1-x^2)^{2 \alpha} (1-y^2)^{2 \alpha} (x-y)^{2 \alpha}  \mbox{d} y \mbox{d} x=
\frac{\pi  2^{6 \alpha +1} 3^{-3 \alpha } \alpha  \Gamma (3 \alpha ) \Gamma (2 \alpha
   +1)^2}{\Gamma \left(\alpha +\frac{5}{6}\right) \Gamma \left(\alpha +\frac{7}{6}\right)
   \Gamma (5 \alpha +2)}.
\end{equation}

Our Dyson-index-based ansatz, then, is that 
\begin{equation} \label{sepX}  
\frac{\int\limits_{-1}^1\int\limits_{-1}^x  \tilde{\chi}_{2 \alpha} \left(
\left.\sqrt{\frac{1-x}{1+x}}\right/ \sqrt{\frac{1-y}{1+y}}	
	\right)(1-x^2)^{2 \alpha}(1-y^2)^{2 \alpha} (x-y)^{2 \alpha} \mbox{d} y\mbox{d} x}{\int\limits_{-1}^1\int\limits_{-1}^x  (1-x^2)^{2 \alpha}(1-y^2)^{2 \alpha}(x-y)^{2 \alpha}  \mbox{d} y \mbox{d} x}
\end{equation}
gives the generalized ($\alpha$-th) Hilbert-Schmidt separability probability. For $\alpha = \frac{1}{2}$, we recover the two-rebit formula 
(\ref{sepR}), while for  $\alpha=1$, under the ansatz, we would conjecturally obtain the two-qubit value of $\frac{8}{33}$, while for $\alpha=2$, the two-quater[nionic]bit value of $\frac{26}{323}$ would be gotten, and similarly, for $\alpha=4$, the (presumably) two-octo[nionic]bit value of $\frac{44482}{4091349}$ \cite{slateroctonionic}. (The volume forms listed in \cite[Table 1]{lovasandai} for the sets of self-adjoint matrices $\mathcal{M}^{sa}_{2,\mathbb{R}}$, $\mathcal{M}^{sa}_{2,\mathbb{C}}$,
are $\frac{|x-y|}{\sqrt{2}}$ in the $\alpha=\frac{1}{2}$ case, and $\frac{(x-y)^2 \sin{\phi}}{2}$ in the $\alpha =1$ case, respectively. Our calculations of the term 
$\mbox{det}(1-Y^2)^d$ appearing in the several Lovas-Andai volume formulas \cite[pp. 10, 12]{lovasandai}, such as this one for the volume of separable states,
\begin{align}\label{eq:cndVol}
\begin{split}
\text{Vol}\left(\mathcal{D}_{\{4,\mathbb{K}\}}^s(D)\right)	&= 
\frac{\det (D)^{4d-\frac{d^2}{2}}}{2^{6d}} \\
&\times
\int\limits_{\mathcal{E}_{2,\mathbb{K}}}
\det (I-Y^2)^d\times
\chi_d\circ\sigma\left(\sqrt{\frac{I-Y}{I+Y}}\right)
\mbox{d} \lambda_{d+2}(Y),
\end{split}
\end{align}
[the function $\sigma(V)=\varepsilon$ being the ratio of the two singular values of the $2 \times 2$ matrix $V$] appear to be  consistent with the use of the $(1-x^2)^{2 \alpha} (1-y^2)^{2 \alpha}$ terms in the  ansatz (\ref{sepX}).)

The values $\alpha=\frac{1}{2}, 1, 2, 4$ themselves correspond to the real, complex, quaternionic and octonionic division algebras. We can, further, look at the other nonnegative (non-division algebra) integral values of $\alpha$. So, for $\alpha=3$, we have the formal prediction \cite{MomentBased,LatestCollaboration} of $\frac{2999}{103385}$.

In this context, let us first note that for the denominator of  (\ref{sepC}), corresponding to $\alpha=1$, we obtain $\frac{256}{1575}$ (a result we later importantly employ (\ref{Denominator})). Using high-precision numerical integration (\url{https://mathematica.stackexchange.com/q/133556/29989})
for the 
corresponding numerator of (\ref{sepC}), we obtained 0.0358226206958479506059638010848. The resultant ratio (dividing by $\frac{256}{1575}$) is 0.220393076546720789860910104330, within $90\%$ of 
0.242424. However, somewhat disappointingly, it was not readily apparent as to what exact values these figures might correspond.

The analogous numerator-denominator ratio in the $\alpha = 2$ (two-quaterbit) instance was 0.0534499, while the predicted separablity probability is $\frac{26}{323} \approx 0.0804954$. It can then be seen that the required constant of proportionality ($\frac{0.0534499}{0.0804954} =0.664013$) in the $\alpha =2$ case is not particularly close to the square of that in the $\alpha =1$ instance ($0.909106^2 =0.826473$). Similarly, in the $\alpha =4$ case, the numerator-denominator ratio is 0.00319505, while the predicted value would be $\frac{44482}{4091349} =0.0108722$ (with the ratio of these two values being 0.293873). So, our
ansatz (\ref{sepX}) would not seem to extend to the sequence of constants of proportionality themselves conforming to the Dyson-index pattern. But the analyses so far could only address this specific issue concerning constants of proportionality.
\section{Expanded analyses}
We, then, broadened the scope of the inquiry with the use of this particular formula of Lovas and Andai for the Hilbert-Schmidt volume of separable states \cite[p. 11]{lovasandai},
\begin{equation*}
\text{Vol}(\mathcal{D}^s_{\{4,\mathbb{K}\}})=	\int\limits_{\begin{array}{c}
		D_1,D_2>0 \\
		\mbox{Tr} (D_1+D_2) = 1
		\end{array}
	}  
	\det (D_1 D_2)^d f(D_2 D_1^{-1})
	\mbox{d} \lambda_{2d+3}(D_1,D_2), 
	\end{equation*}
where 
\begin{equation} \label{relationship}
f(D_2 D_1^{-1}) = \chi_d\circ
	\exp \left({-\cosh^{-1}\left(
		\frac{1}{2}
		\sqrt{\frac{\det (D_1)}{\det (D_2)}}
		\mbox{Tr} \left(D_2 D_1^{-1}\right)
		\right)}\right).
\end{equation}
Here $D_1$ denotes the upper diagonal $2 \times 2$ block, and $D_2$, the lower diagonal $2 \times 2$  block of the $4 \times 4$ density 
matrix \cite[p. 3]{lovasandai},
\begin{equation*}
	D = \left(
	\begin{array}{cc}
	D_1    & C \\
	C^\ast & D_2
	\end{array}
	\right).
\end{equation*}
The Lovas-Andai parameter $d$ is defined as 1 in the two-rebit case and 2 in the standard two-qubit case (that is, in our notation,
$\alpha=\frac{d}{2}$). Further, the relevant division algebra 
$\mathbb{K}$ is $\mathbb{R}$, $\mathbb{C}$ or $\mathbb{Q}$,
according to $d=1,2,4$. The exponential term in (\ref{relationship}) corresponds to the ``singular value ratio'', 
\begin{equation}
\sigma(V) =  \exp \left({-\cosh^{-1}\left(
		\frac{||V||^2_{HS}}{2 \det (V)}
		\right)}\right)= \exp \left({-\cosh^{-1}\left(
		\frac{1}{2}
		\sqrt{\frac{\det (D_1)}{\det (D_2)}}
		\mbox{Tr} \left(D_2 D_1^{-1}\right)
		\right)}\right),  
\end{equation}
of the  matrix $V=D_2^{1/2} D_1^{-1/2}$, where the Hilbert-Schmidt norm is indicated. (In \cite[sec. IV]{yin2015empirical} the ratio of singular values of $2 \times 2$ ``empirical polarization matrices'' is investigated.) 
\subsection{Generation of random density matrices}
\subsubsection{Two-rebit case}
Firstly, taking $d=1$, we generated 687 million random (with respect to Hilbert-Schmidt measure) $4 \times 4$ density matrices situated in the 9-dimensional convex set of two-rebit states \cite[App. B]{generating} \cite{carl,batle2}. Of these, 311,313,185 were separable (giving a sample probability of 0.453149, close to the value of $\frac{29}{64} \approx 0.453125$, now formally established by Lovas and Andai). Additionally, we  binned the two sets (separable and all) of density matrices into 200 subintervals of $[0,1]$, based on their corresponding values 
of $\sigma(V)$ (Fig.~\ref{fig:rebitpair}). Fig.~\ref{fig:rebitsep} is a plot of the estimated separability probabilities (remarkably close to linear with slope 1--as previously observed \cite[Fig. 1]{lovasandai}), while Fig.~\ref{fig:rebitdiff} shows the result of subtracting from this curve  the very well-fitting (as we, of course, expected from the Lovas-Andai proof) function
$\tilde{\chi}_1 (\varepsilon )$, as given by ((\ref{BasicFormula}),(\ref{poly})). (If one replaces $\tilde{\chi}_1 (\varepsilon)$ by simply its close
approximant $\varepsilon$, then the corresponding integrations would yield a  ``separability probability", not of $\frac{29}{64} \approx 0.453125$, but of 
$\frac{16}{9}-\frac{35 \pi ^2}{256} \approx 0.428418$. If we similarly employ $\varepsilon^2$ in the two-qubit case, rather than the
[previously undetermined] $\tilde{\chi}_2 (\varepsilon)$, the corresponding integrations yield $\frac{13}{66} \approx 0.19697$, and not the presumed correct
result of $\frac{8}{33} \approx 0.242424$.) Fig.~\ref{fig:qubitdiffZ} will serve as the two-qubit analogue of Fig.~\ref{fig:rebitdiff}, further validating the  formula (\ref{VerifiedFormula}) for  $\tilde{\chi}_2 (\varepsilon)$ to be obtained.

\begin{figure}
    \centering
    \includegraphics{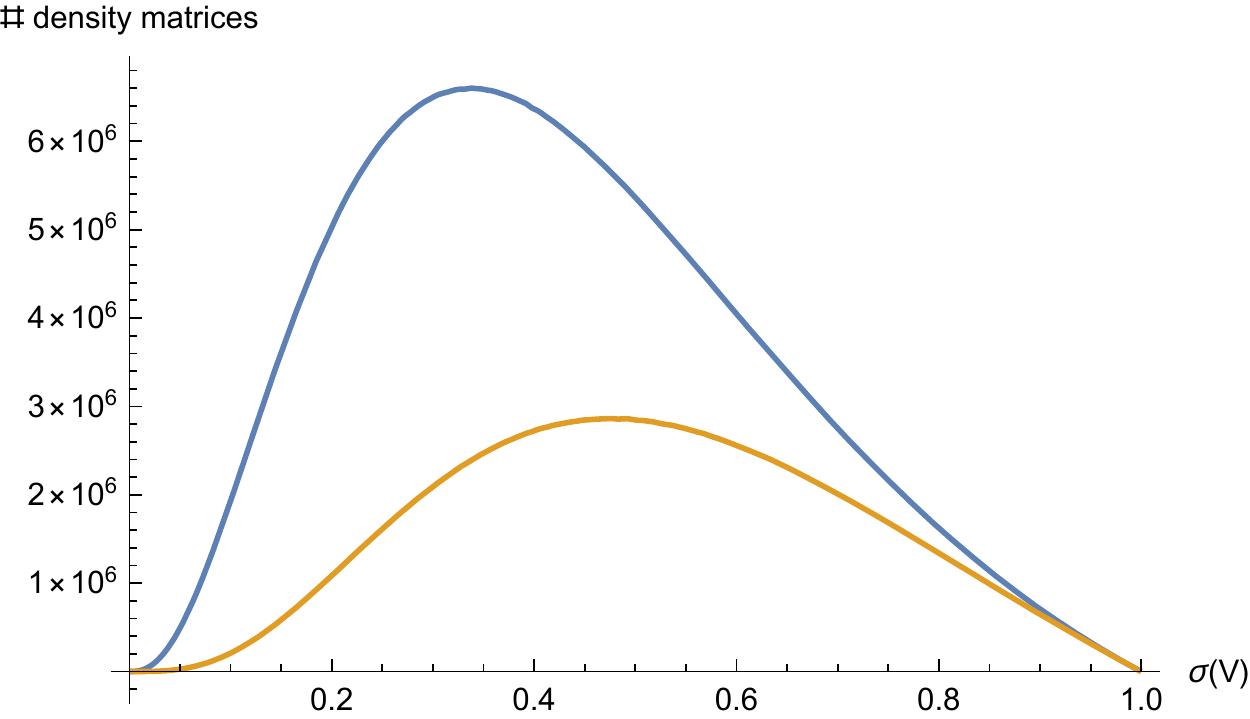}
    \caption{Recorded counts by binned values of the singular value ratio $\sigma(V)$  of 687 million two-rebit density matrices randomly generated (with respect to Hilbert-Schmidt measure), along with the accompanying (lesser) counts of separable density matrices}
    \label{fig:rebitpair}
\end{figure}
\begin{figure}
    \centering
    \includegraphics{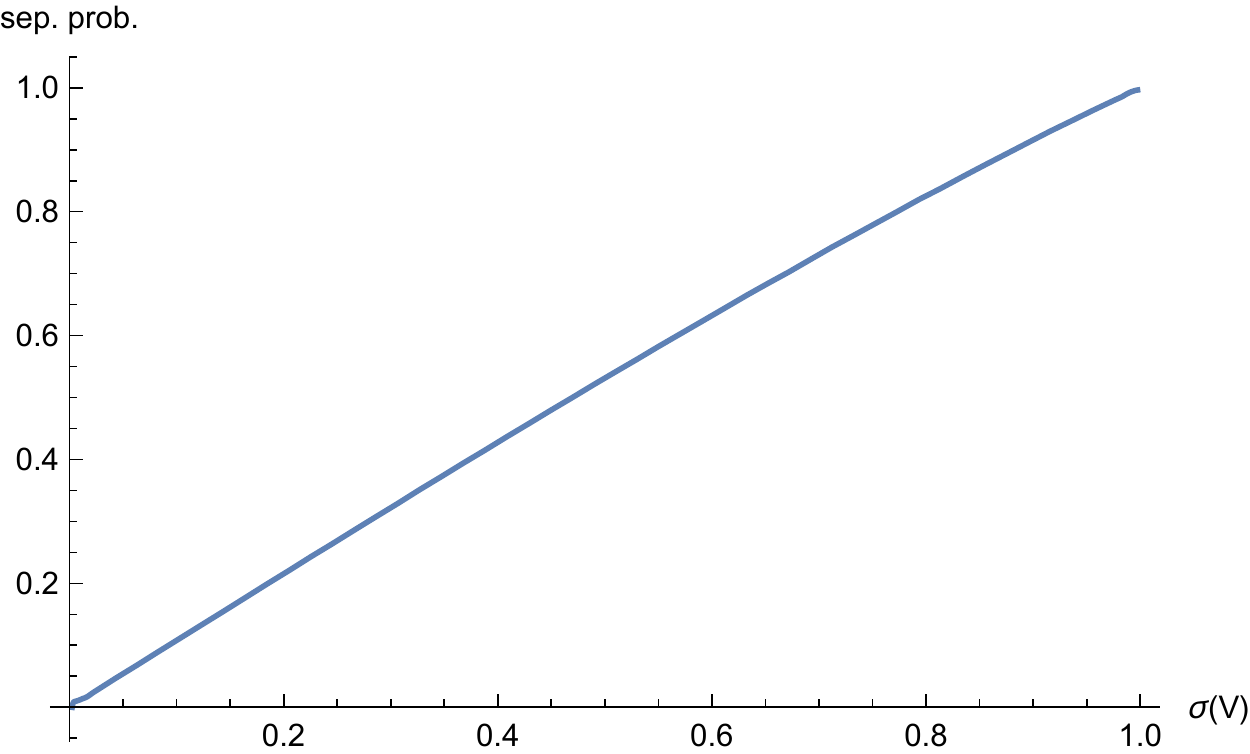}
    \caption{Estimated two-rebit separability probabilities (close to linear with slope 1)}
    \label{fig:rebitsep}
\end{figure}
\begin{figure}
    \centering
    \includegraphics{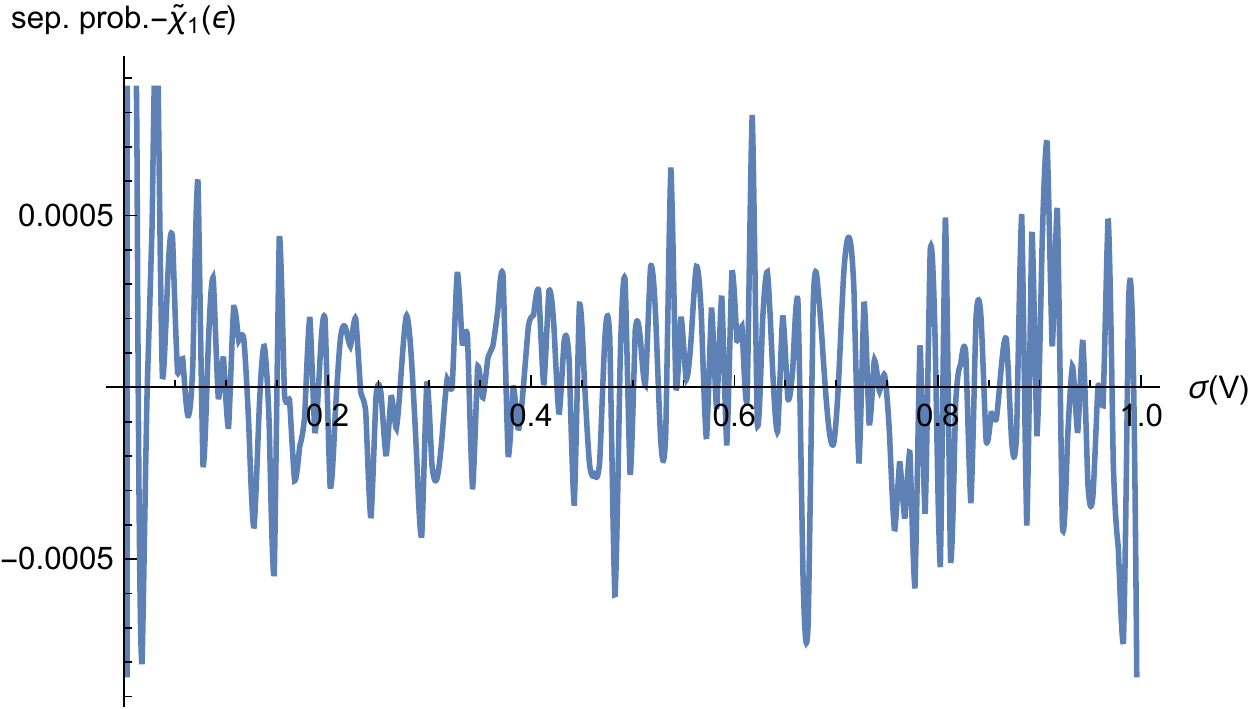}
    \caption{Result of subtracting $\tilde{\chi}_1 (\varepsilon )$ from the estimated two-rebit separability probability curve (Fig.~\ref{fig:rebitsep}). Fig.~\ref{fig:qubitdiffZ} will be the two-qubit analogue.}
    \label{fig:rebitdiff}
\end{figure}
\subsubsection{Two-qubit case}
We, next, to test a Dyson-index ansatz, taking $d=2$, generated 6,680 million random (with respect to Hilbert-Schmidt measure) $4 \times 4$ density matrices situated in the 15-dimensional convex set of 
(standard) two-qubit states \cite[eq. (15)]{generating}. Of these, 1,619,325,156 were separable (giving a sample probability of 0.242414, close to the conjectured, well-supported [but not yet formally proven] value of $\frac{8}{33} \approx 0.242424$) (cf. \cite{singh2014relative}). We, again,  binned the two sets (separable and all) of density matrices into 200 subintervals of $[0,1]$, based on their corresponding values 
of $\sigma(V)$ (Fig.~\ref{fig:qubitpair}). Fig.~\ref{fig:qubitsep} is a plot (now, clearly non-linear [cf. Fig.~\ref{fig:rebitsep}]) of the estimated separability probabilities, along with the quite 
closely fitting, but mainly slightly subordinate $\tilde{\chi}_1^2 (\varepsilon )$ curve. Fig.~\ref{fig:qubitdiff} shows the result/residuals (of relatively small magnitude) of subtracting
$\tilde{\chi}_1^2 (\varepsilon )$ from the estimated separability probability curve. 
\begin{figure}
    \centering
    \includegraphics{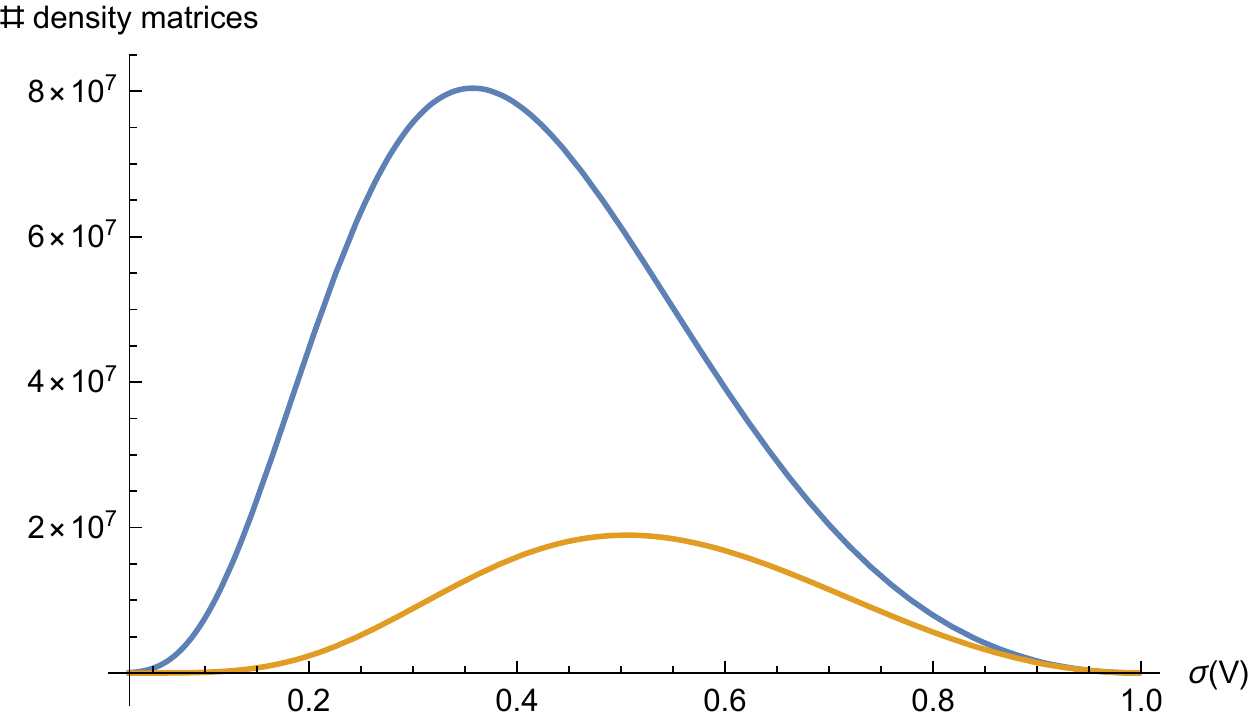}
    \caption{Recorded counts by binned values of the singular value ratio $\sigma(V)$  of all 6,680 million two-qubit density matrices randomly generated (with respect to Hilbert-Schmidt measure), along with the accompanying (lesser) counts of separable density matrices}
    \label{fig:qubitpair}
\end{figure}
\begin{figure}
    \centering
    \includegraphics{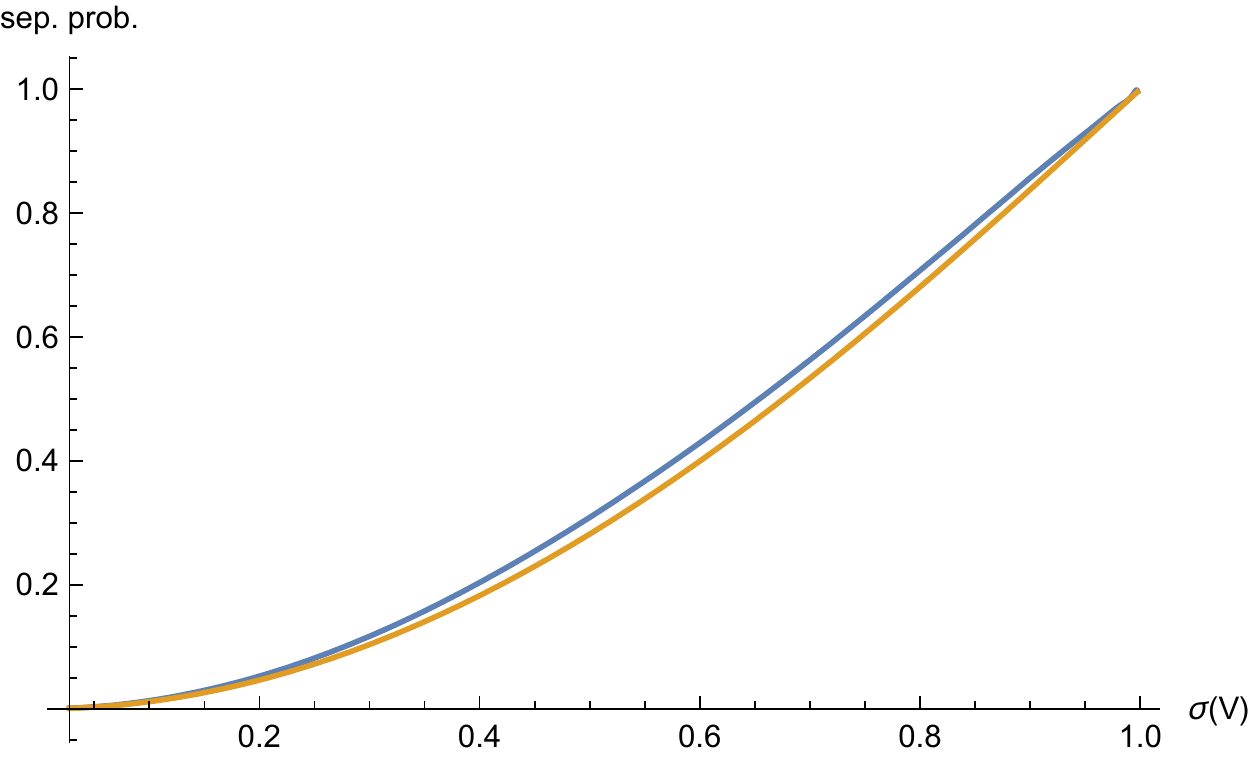}
    \caption{Estimated two-qubit separability probabilities together with  the slightly subordinate curve $\tilde{\chi}_1^2 (\varepsilon )$}
    \label{fig:qubitsep}
\end{figure}
\begin{figure}
    \centering
    \includegraphics{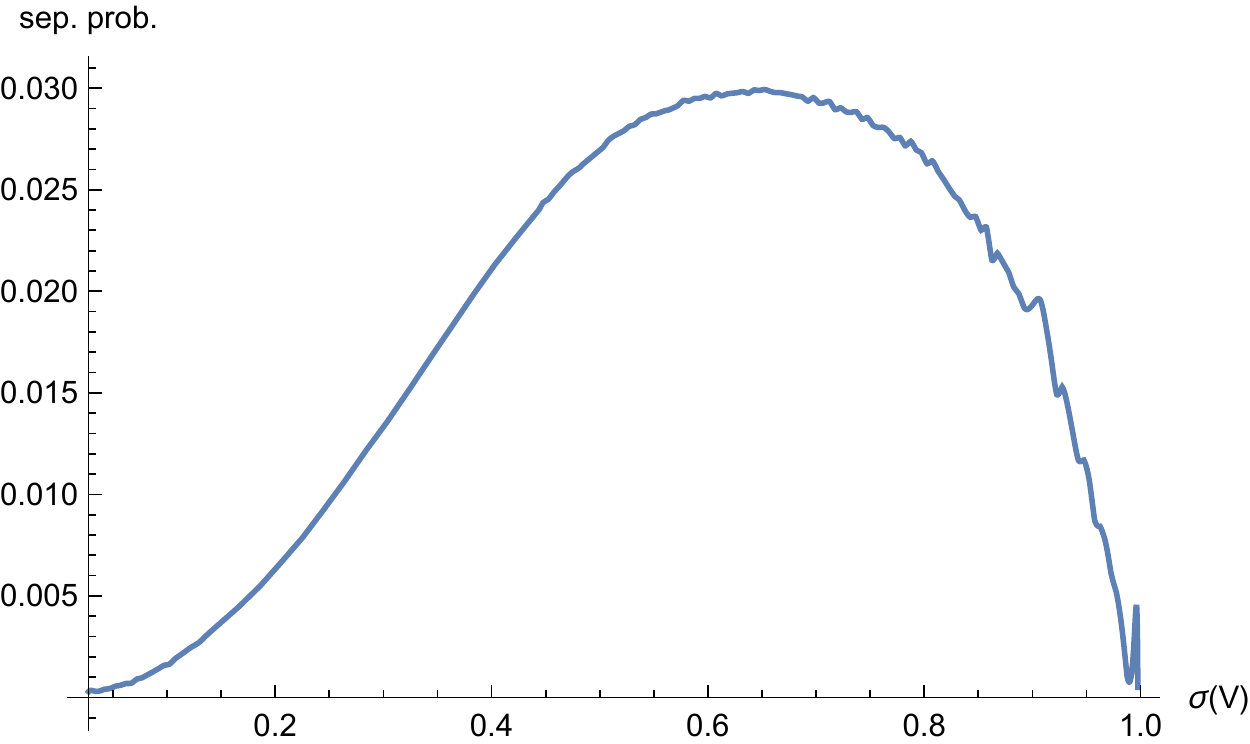}
    \caption{Result of subtracting the (slightly subordinate) $\tilde{\chi}_1^2 (\varepsilon )$ curve from the estimated two-qubit separability probability curve in Fig.~\ref{fig:qubitsep}}
    \label{fig:qubitdiff}
\end{figure}
So, it would seem that the square of the  explicitly-constructed Lovas-Andai two-rebit separability function $\tilde{\chi}_1 (\varepsilon )$ provides, at least, an interesting approximation to the sought  two-qubit separability function $\tilde{\chi}_2 (\varepsilon )$.

The Dyson-index ansatz--the focus earlier in the paper--appears to hold in some trivial/degenerate sense if we employ rather than the Lovas-Andai or Slater separability functions discussed above, the ``Milz-Strunz'' ones
\cite{milzstrunz}. Then, rather than the singular-value ratio $\varepsilon$ or the ratio of diagonal entries $\mu$, one would use as the dependent/predictor variable, the Casimir invariants of the reduced systems \cite{slater2016invariance}. In these cases, the separability functions become simply constant in nature. In the two-rebit and two-qubit cases, this invariant is the Bloch radius ($r$) of one of the two reduced systems. From the arguments of Lovas and Andai 
\cite[Cor. 2, Thm. 2]{lovasandai}, it appears that one can  assert that the Milz-Strunz form of two-rebit separability function assumes the constant value $\frac{29}{64}$ for $r \in [0,1]$. Then, it would seem that the two-qubit counterpart would be the constant value $\frac{8}{33}$ for $r \in [0,1]$, with the corresponding (Dyson-index ansatz) constant of proportionality being $\frac{\frac{8}{33}} {(\frac{29}{64})^2} =\frac{32768}{27753} \approx 1.1807$.
\section{Relations between $\epsilon =\sigma(V)$ and Bloore/Slater variable $\mu$}
Let us now note a quite interesting phenomenon, apparently relating the Lovas-Andai analyses to  previous ones of Slater \cite{slater833}. If we perform the indicated integration in the denominator of (\ref{sepR}), following the integration-by-parts scheme adopted by  Lovas and Andai 
\cite[p. 12]{lovasandai}, at an intermediate stage we arrive at the univariate integrand,
\begin{equation} \label{jacLA}
\frac{128 t^3 \left(5 \left(5 t^8+32 t^6-32 t^2-5\right)-12 \left(\left(t^2+2\right)
   \left(t^4+14 t^2+8\right) t^2+1\right) \log (t)\right)}{3 \left(t^2-1\right)^8}.
\end{equation}
(Its integral over $t \in [0,1]$ equals the noted value of $\frac{16}{35}$, where $\frac{\frac{16}{35}-\frac{1}{4}}{\frac{16}{35}}=\frac{29}{64}$.)
This, interestingly,  bears a very close (almost identical) structural resemblance to the  jacobian/volume-element
\begin{equation} \label{jacS}
\mathcal{H}_{real}(\mu)=-\frac{\mu ^4 \left(5 \left(5 \mu ^8+32 \mu ^6-32 \mu ^2-5\right)-12 \left(\left(\mu
   ^2+2\right) \left(\mu ^4+14 \mu ^2+8\right) \mu ^2+1\right) \log (\mu )\right)}{1890
   \left(\mu ^2-1\right)^9}    
\end{equation}
(integrating to $\frac{\pi ^2}{2293760}$ over $\mu \in [0,1]$) reported by Slater in \cite[eq. (15)]{slater833} and \cite[eq. (10)]{slaterPRA2},
also in the context of two-rebit separability functions. (We change the notation in those references
from $\mathcal{J}_{real}(\nu)$ to $\mathcal{H}_{real}(\mu)$ here, since we have made the 
transformation $\nu \rightarrow \mu^2$, to facilitate this comparison, and the analogous one below in the two-qubit 
context--with the approach of Lovas and Andai. However, we will still
note some results below in the original [$\nu$] framework.)
To faciltate the comparison between these two functions, we set $t = \mu= \tilde{t}$, and then divide (\ref{jacLA}) by (\ref{jacS}), obtaining the simple ratio
\begin{equation} \label{firstratio}
\frac{80640 \left(1-\tilde{t} ^2\right)}{\tilde{t}}.    
\end{equation}
But we note that in in \cite{slater833} and \cite{slaterPRA2}--motivated by work in a $3 \times  3$ density matrix context  of 
Bloore \cite{bloore1976geometrical}--the variable 
$\mu$ was taken to be the ratio $\sqrt{\frac{\rho_ {11} \rho_ {44}}{\rho_ {22} \rho_ {33}}}$ of the square root of the product of the (1,1) and (4,4) diagonal entries of the density matrix \cite[eq. (1)]{slater833}
\begin{equation} \label{Matrix}
D=
\left(
\begin{array}{cccc}
 \rho_ {11} & z_{12} \sqrt{\rho_ {11} \rho_ {22}} & z_{13} \sqrt{\rho_ {11} \rho_ {33}} &
   z_{14} \sqrt{\rho_ {11} \rho_ {44}} \\
 z_{12} \sqrt{\rho_ {11} \rho_ {22}} & \rho_ {22} & z_{23} \sqrt{\rho_ {22} \rho_ {33}} &
   z_{24} \sqrt{\rho_ {22} \rho_ {44}} \\
 z_{13} \sqrt{\rho_ {11} \rho_ {33}} & z_{23} \sqrt{\rho_ {22} \rho_ {33}} & \rho_ {33} &
   z_{34} \sqrt{\rho_ {33} \rho_ {44}} \\
 z_{14} \sqrt{\rho_ {11} \rho_ {44}} & z_{24} \sqrt{\rho_ {22} \rho_ {44}} & z_{34}
   \sqrt{\rho_ {33} \rho_ {44}} & \rho_ {44} \\
\end{array}
\right)    
\end{equation}
to the product of the (2,2) and (3,3) ones, while in \cite{lovasandai}, it would be  the ratio $\sigma(V)$ of the 
singular values of the noted $2 \times 2$ matrix $D_2^{1/2}  D_1^{-1/2}$. From \cite[eq. (91)]{slater833}, we can deduce that one must multiply
$\mathcal{H}_{real}(\mu)$ by $\frac{1048576}{\pi ^2}$, so that its integral from 0 to 1 equals the Lovas-Andai counterpart result 
of $\frac{16}{35}$. (The jacobian of the transformation to the two-rebit density matrix parameterization
(\ref{Matrix}) is $(\rho_{11} \rho_{22} \rho_{33} \rho_{44})^{3/2}$, and for the two-qubit counterpart, 
$(\rho_{11} \rho_{22} \rho_{33} \rho_{44})^3$ \cite[p. 4]{slaterPRA2}. These jacobians are also reported in \cite{andai2006volume}.)

In \cite[eq. (93)]{slater833}, the two-rebit separability function $S_{real}(\nu)$ was taken to be proportional to the incomplete beta 
function $B_{\nu}(\nu,\frac{1}{2},2)=\frac{2}{3}(3 -\nu) \sqrt{\nu}$--an apparently much simpler function
than the Lovas-Andai counterpart (\ref{poly}) above. Given the just indicated scaling  by $\frac{1048576}{\pi ^2}$, to achieve the $\frac{29}{140}$ separability probability numerator result of Lovas and Andai, we must take the hypothesized separability
function to be, then,  $\frac{3915 \pi ^2 (3-\nu ) \sqrt{\nu }}{131072}$.

A parallel phenomenon is observed in the two-qubit case, where \cite[eq. (11)]{slaterPRA2},
\begin{equation} \label{Qcomplex}
\mathcal{H}_{complex}(\mu)=-\frac{\mu ^7 (h_1+h_2)}{1801800 (\mu^2 -1)^{15}},
\end{equation}
with 
\begin{displaymath}
h_1=(\mu -1) (\mu +1) \left(363 \mu ^{12}+10310 \mu ^{10}+58673 \mu ^8+101548 \mu ^6+58673
   \mu ^4+10310 \mu ^2+363\right)
\end{displaymath}
and
\begin{displaymath}
h_2=-140 \left(\mu ^2+1\right) \left(\mu ^{12}+48 \mu ^{10}+393 \mu ^8+832 \mu ^6+393 \mu
   ^4+48 \mu ^2+1\right) \log (\mu ).
\end{displaymath}
Setting $\alpha=1$ in the denominator formula (\ref{General}), and again following the integration-by-parts
scheme of Lovas and Andai, while setting  $t = \mu= \tilde{t}$,
the simple ratio (proportional to the square of (\ref{firstratio})) is now
\begin{equation} \label{final}
\frac{210862080 \left(1- \tilde{t} ^2\right)^2}{\tilde{t} ^2}.
\end{equation}
To achieve the $\frac{256}{1575}$ Lovas-Andai two-qubit denominator result, we must multiply 
$\mathcal{H}_{complex}(\nu)$ (\ref{Qcomplex}) by 328007680.

The two-qubit separability function $S_{complex}(\nu)$ advanced in \cite{slater833} was proportional to the square
of that--$B_{\nu}(\nu,\frac{1}{2},2)=\frac{2}{3}(3 -\nu) \sqrt{\nu}$--employed in the two-rebit context.
Now, to obtain the two-qubit numerator result of $\frac{2048}{51975}$ necessary for the
$\frac{8}{33}$ separability probability outcome, we took the associated separability function to simply be
$\frac{6}{71} (3-\nu )^2 \nu$. We refer the reader to Figure 2 in \cite{slater833} (and Figs.~\ref{fig:QubitNUDiff} and \ref{fig:QubitMUDiff} below) to see the extraordinarily
good fit of this function. (However, the two-rebit fit displayed there does not appear quite as good.) 

Let us now supplement the earlier plots in \cite{slater833}, with some newly generated ones. 
(Those 2007 plots were based on quasi-Monte Carlo [``low-discrepancy'' point \cite{bratley1992implementation}] sampling, while the ones presented here
are based on more ``state-of-the-art'' sampling methods \cite{generating}, with many more density matrices [but, of ``higher-discrepancy"] generated.) In Figs.~\ref{fig:RebitNU} and \ref{fig:RebitMU} we show
the two-rebit separability probabilities as a function, firstly, of $\nu=\frac{\rho_ {11} \rho_ {44}}{\rho_ {22} \rho_ {33}}$ and, secondly, as a function of $\mu=\sqrt{\nu}= \sqrt{\frac{\rho_ {11} \rho_ {44}}{\rho_ {22} \rho_ {33}}}$,
together with the curves $\frac{3915 \pi ^2 (3-\nu ) \sqrt{\nu }}{131072}$ and $\frac{3915 \pi ^2 (3-\mu^2 ) \mu}{131072}$, respectively.
\begin{figure}
\includegraphics{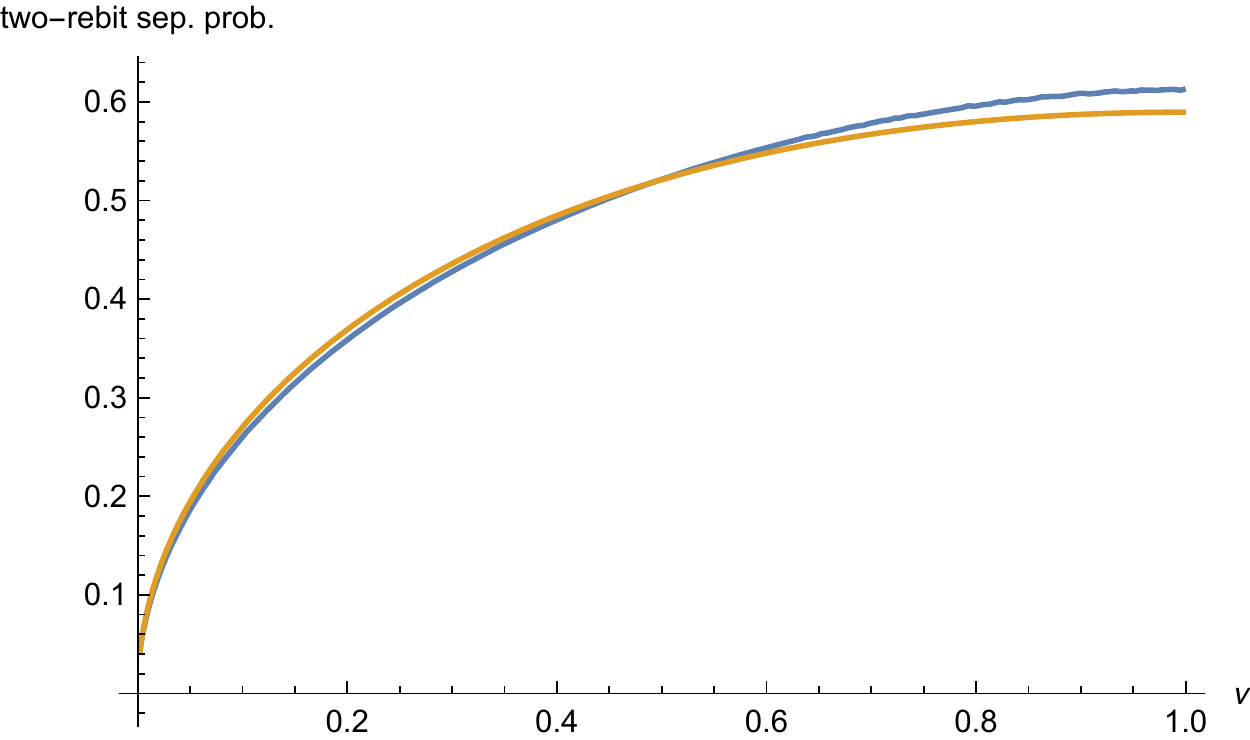}
\caption{Estimated two-rebit Hilbert-Schmidt separability probabilities, based on 687 million randomly-generated density matrices, together with
the hypothesized  (slightly subordinate) separability function $\frac{3915 \pi ^2 (3-\nu ) \sqrt{\nu }}{131072}$}
\label{fig:RebitNU}
\end{figure}
\begin{figure}
\includegraphics{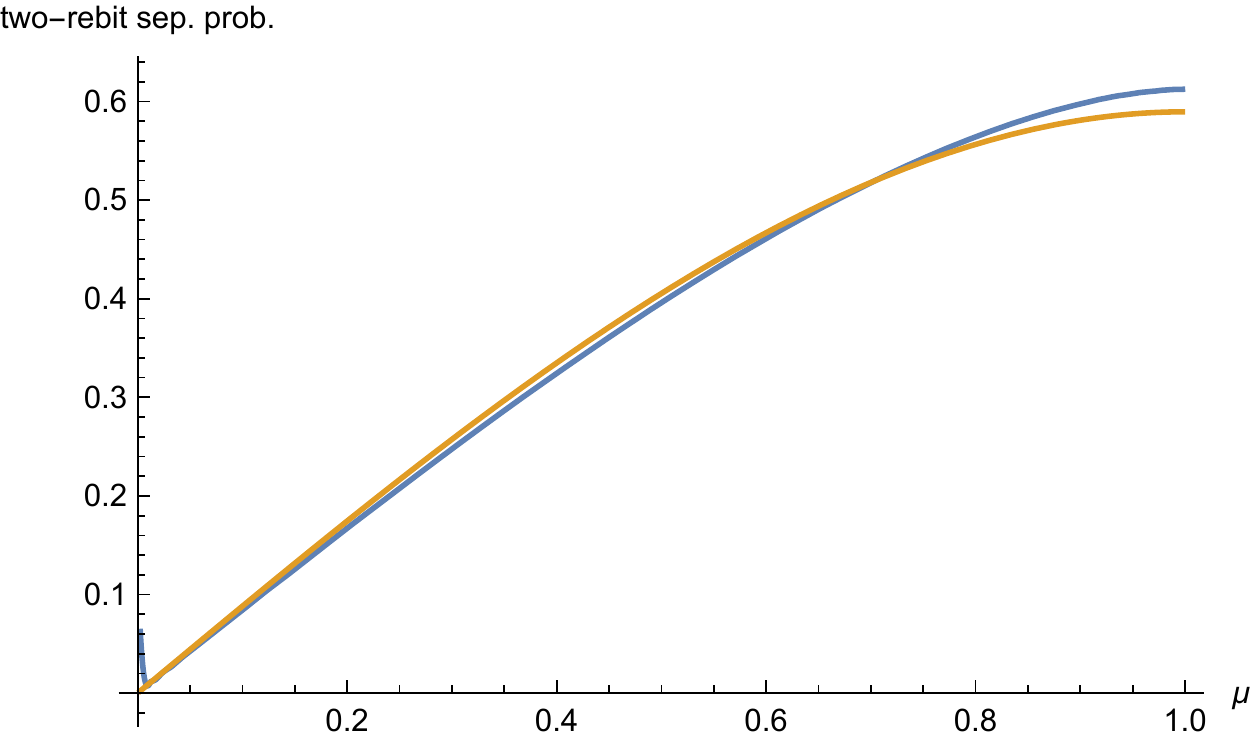}
\caption{Estimated two-rebit Hilbert-Schmidt separability probabilities, based on 5,077  million randomly-generated density matrices, together with
the hypothesized  (slightly subordinate) separability function $\frac{3915 \pi ^2 (3-\mu^2 ) \mu}{131072}$}
\label{fig:RebitMU}
\end{figure}
In Figs.~\ref{fig:QubitNU} and \ref{fig:QubitMU} we show
the two-qubit separability probabilities as a function, firstly, of $\nu$ and, secondly, as a function of $\mu$,
together with the curves $\frac{6}{71} (3-\nu)^2 \nu$ and $\frac{6}{71} (3-\mu^2)^2 \mu^2$, respectively.
\begin{figure}
\includegraphics{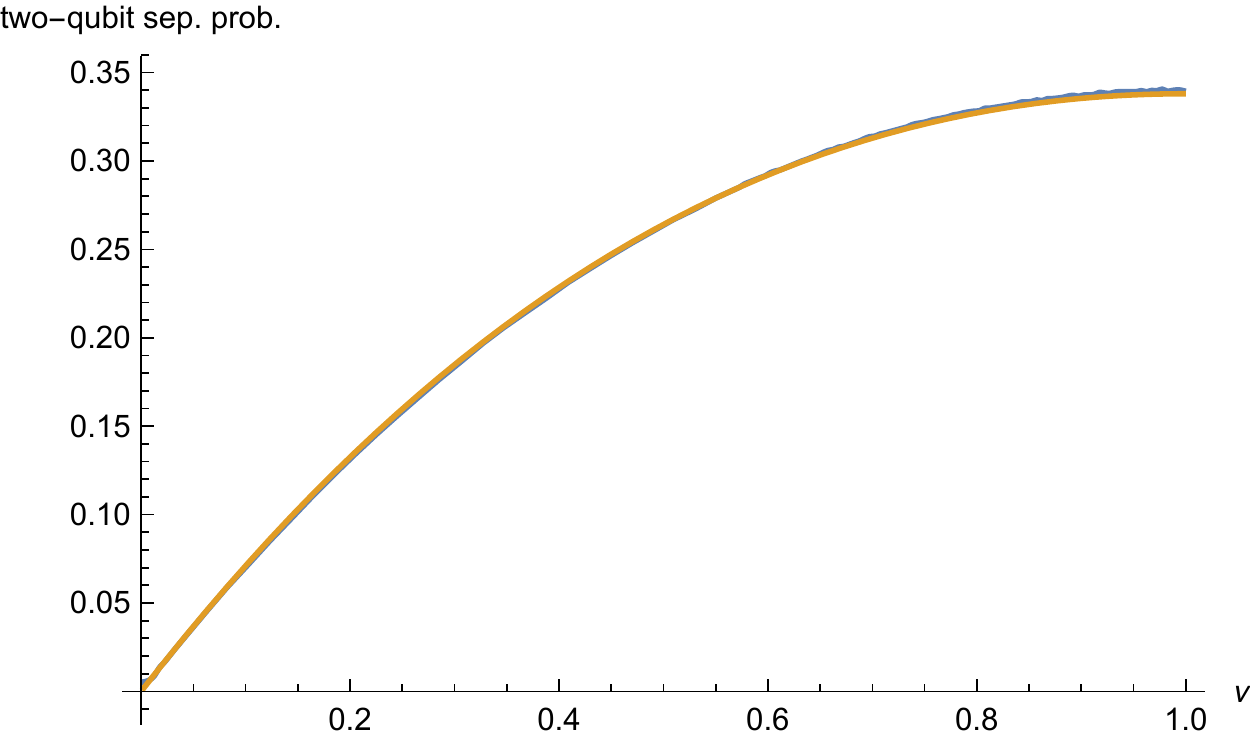}
\caption{Estimated two-qubit Hilbert-Schmidt separability probabilities, based on 507 million randomly-generated density matrices, together with
the (indiscernibly different) separability function $\frac{6}{71} (3-\nu)^2 \nu$ (cf. Fig.~\ref{fig:QubitNUDiff} and \cite[Fig. 2]{slater833} for the residuals from the fit)}
\label{fig:QubitNU}
\end{figure}
\begin{figure}
\includegraphics{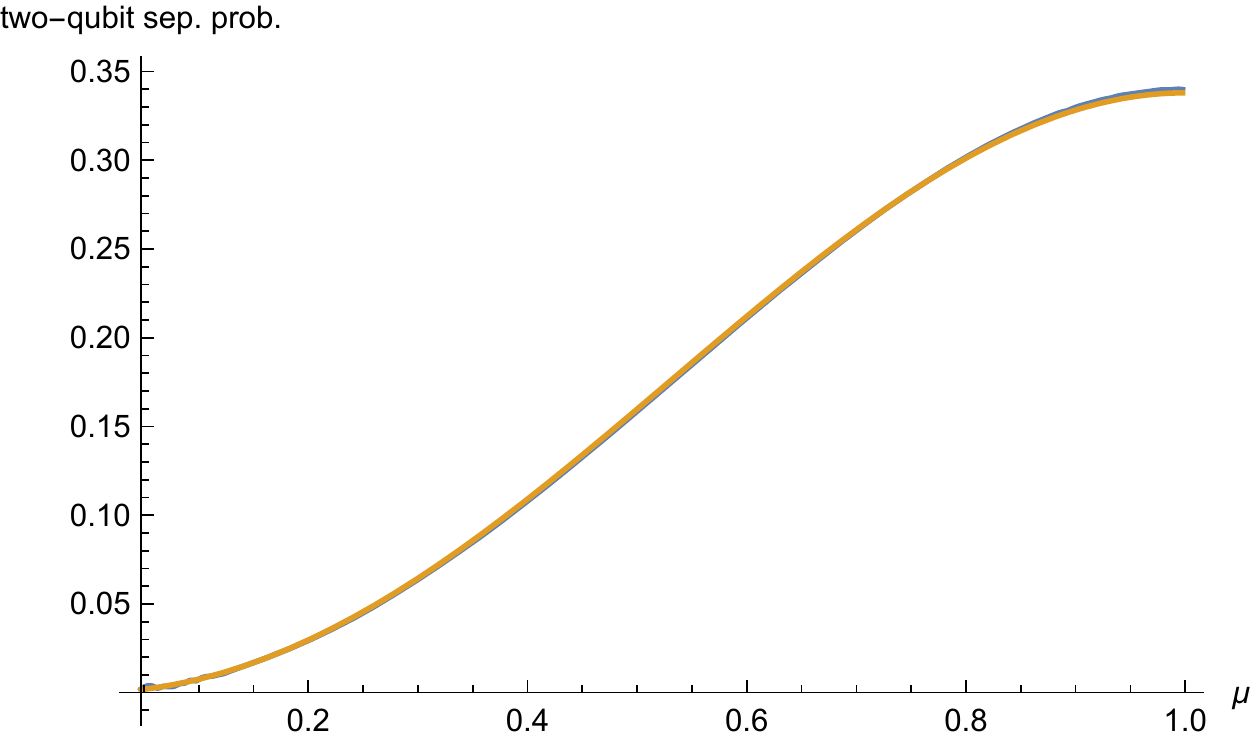}
\caption{Estimated two-qubit Hilbert-Schmidt separability probabilities, based on 3,715 million randomly-generated density matrices, together with
the (indiscernibly different) separability function $\frac{6}{71} (3-\mu^2)^2 \mu^2$}
\label{fig:QubitMU}
\end{figure}

In Figs.~\ref{fig:RebitNUDiff}, \ref{fig:RebitMUDiff}, \ref{fig:QubitNUDiff} and \ref{fig:QubitMUDiff}, rather than showing
the estimated separability probabilities together with the separability functions as in the previous four figures, we show
the estimated separability probabilities {\it minus} the separability functions, that is, the {\it residuals} from this fits.
\begin{figure}
\includegraphics{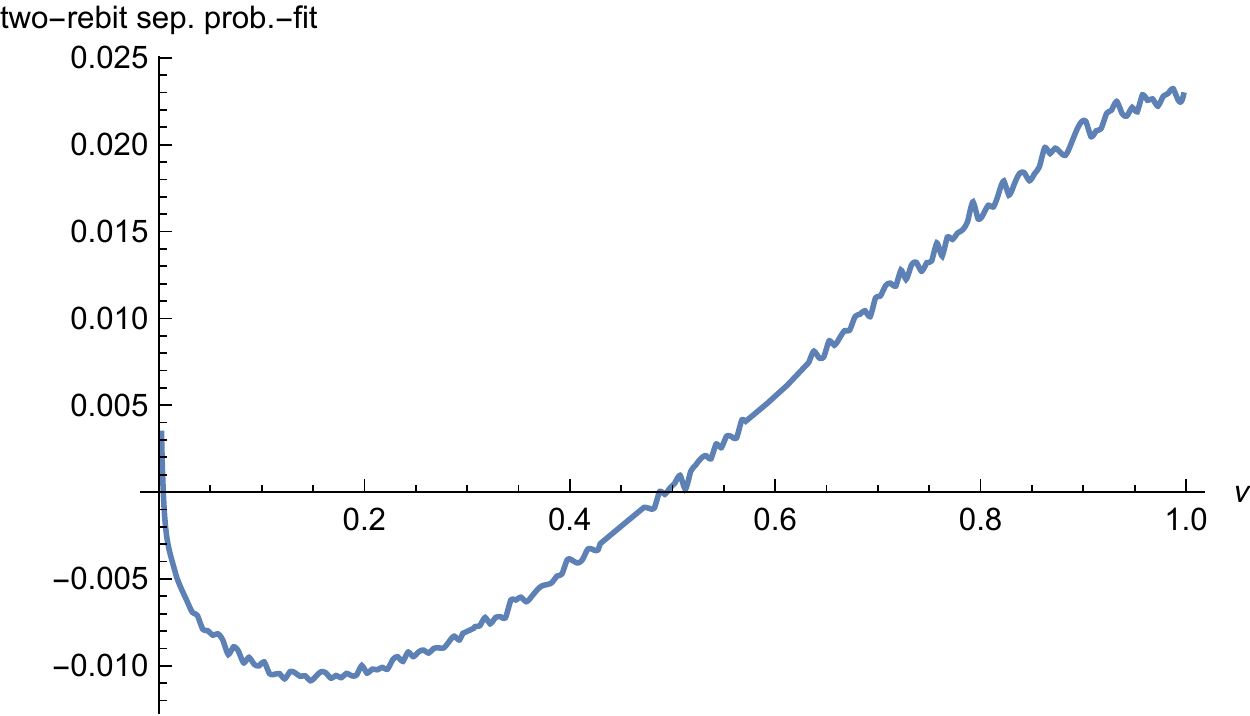}
\caption{Estimated two-rebit Hilbert-Schmidt separability probabilities, based on 687 million randomly-generated density matrices, {\it minus}
the separability function $\frac{3915 \pi ^2 (3-\nu ) \sqrt{\nu }}{131072}$}
\label{fig:RebitNUDiff}
\end{figure}
\begin{figure}
\includegraphics{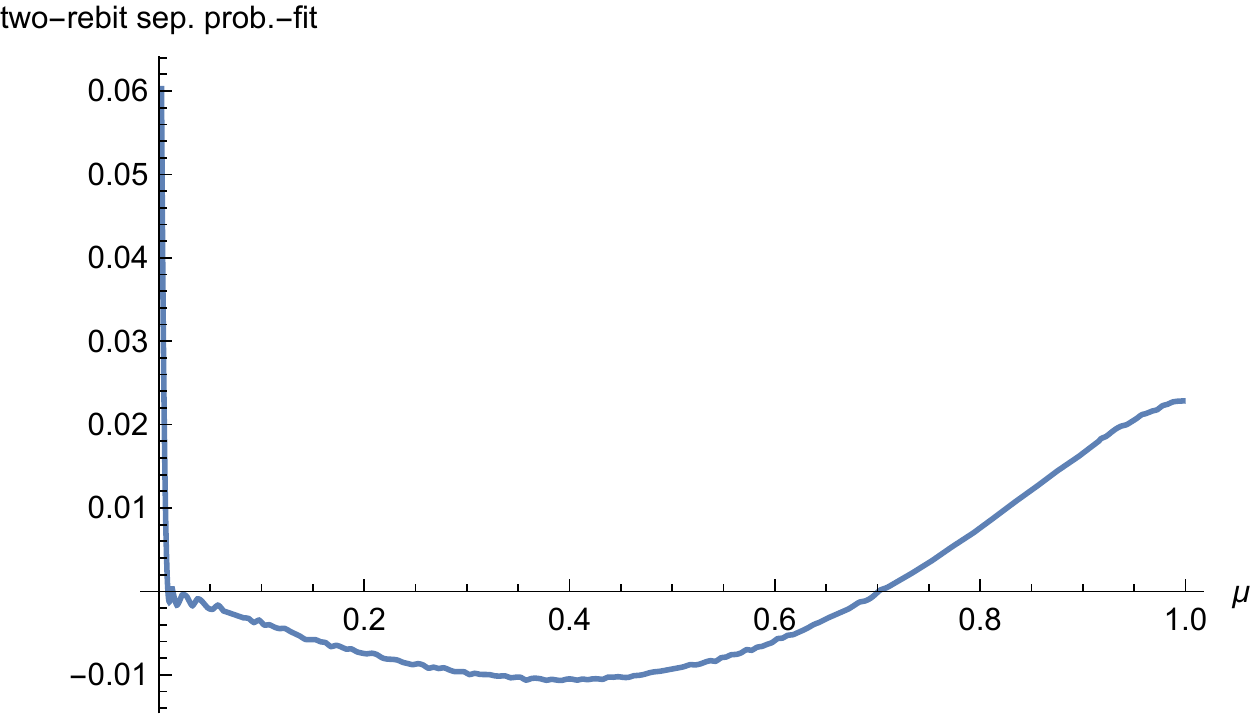}
\caption{Estimated two-rebit Hilbert-Schmidt separability probabilities, based on 5,077  million randomly-generated density matrices, {\it minus}
the separability function $\frac{3915 \pi ^2 (3-\mu^2 ) \mu}{131072}$}
\label{fig:RebitMUDiff}
\end{figure}
\begin{figure}
\includegraphics{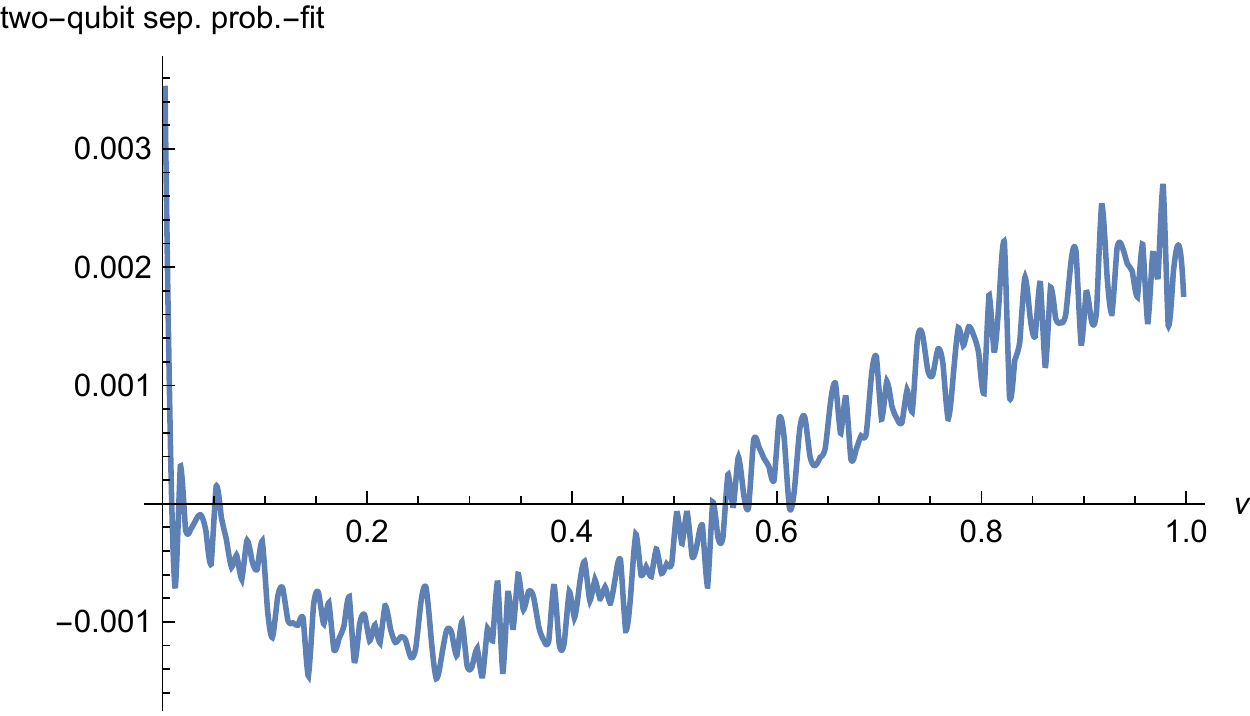}
\caption{Estimated two-qubit Hilbert-Schmidt separability probabilities, based on 3,715 million randomly-generated density matrices, {\it minus}
the separability function $\frac{6}{71} (3-\nu)^2 \nu$}
\label{fig:QubitNUDiff}
\end{figure}
\begin{figure}
\includegraphics{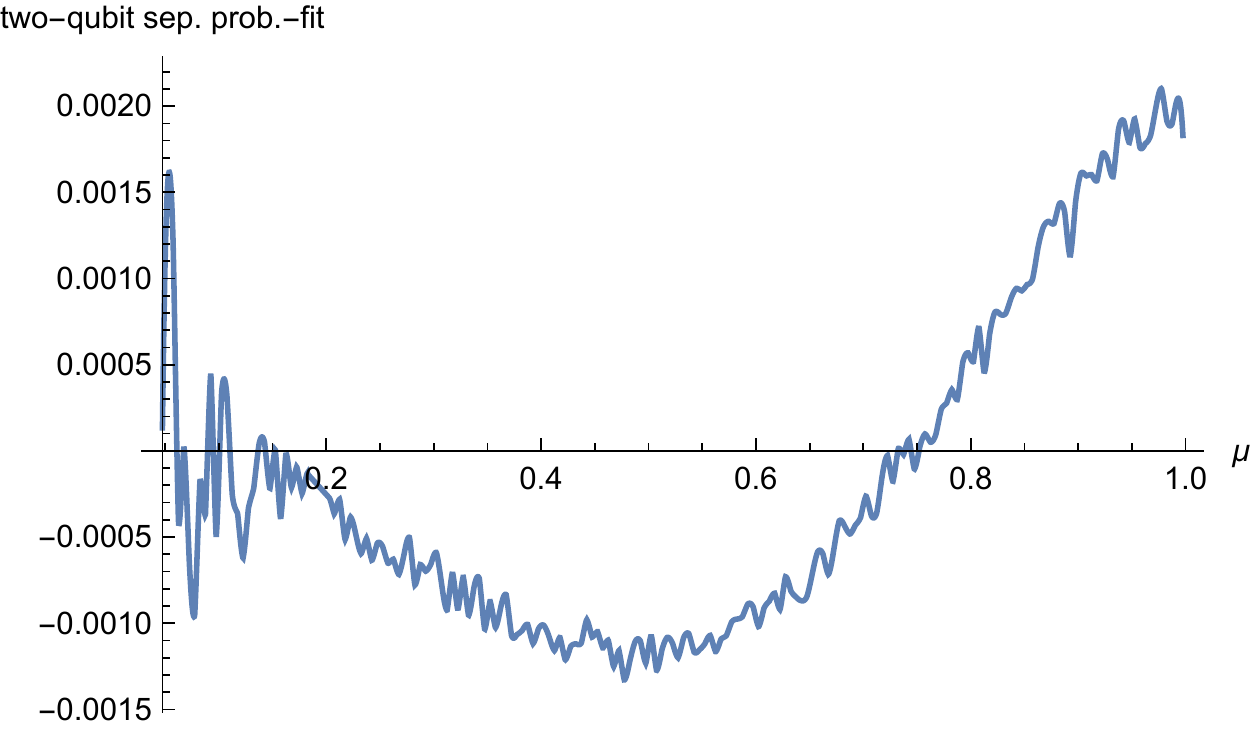}
\caption{Estimated two-qubit Hilbert-Schmidt separability probabilities, based on 3,715 randomly-generated density matrices, {\it minus}
the separability function $\frac{6}{71} (3-\mu^2)^2 \mu^2$}
\label{fig:QubitMUDiff}
\end{figure}

So, at this stage,
the evidence is certainly strong that the Dyson-index ansatz is at least of some value in  approximately fitting
the relationships between two-rebit and two-qubit Hilbert-Schmidt separability functions.
\subsection{Formulas linking the Lovas-Andai variable $\varepsilon$ and the Slater/Bloore variable $\mu$}
Using the two-rebit density matrix parameterization (\ref{Matrix}), then,
taking the previously indicated relationship (\ref{relationship}), 
which has the explicit form in this case of 
\begin{equation} \label{Linkage}
\varepsilon= \exp \left(-\cosh ^{-1}\left(\frac{-\mu ^2+2 \mu  z_{12} z_{34}-1}{2 \mu 
   \sqrt{z_{12}^2-1} \sqrt{z_{34}^2-1}}\right)\right),
\end{equation}
and inverting it, we find
\begin{equation} \label{relationship1} 
\mu = \frac{1}{2} \left(\lambda -\sqrt{\lambda ^2-4}\right),
\end{equation}
where
\begin{displaymath}
\lambda= 2 z_{12} z_{34}-\sqrt{z_{12}^2-1} \sqrt{z_{34}^2-1} \left(\frac{1}{\varepsilon ^2}+1\right)
   \varepsilon.
\end{displaymath}
For the two-qubit counterpart, 
we have
\begin{equation}
\varepsilon=  \exp \left(-\cosh ^{-1}\left(\frac{-\mu ^2+2 \mu  \left(y_{12} y_{34}+z_{12}
   z_{34}\right)-1}{2 \mu  \sqrt{y_{12}^2+z_{12}^2-1}
   \sqrt{y_{34}^2+z_{34}^2-1}}\right)\right). 
\end{equation}
The $z_{ij}$'s are as in the two-rebit case (\ref{7DimFunct}), and the $y_{ij}$'s are now the 
corresponding imaginary parts in the natural extension of the two-rebit density matrix
parameterization (\ref{Matrix}). 
A similar inversion yields 
\begin{equation} \label{relationship2}
\mu = \frac{1}{2} \left(\tilde{\lambda}-\sqrt{\tilde{\lambda} ^2-4} \right),
\end{equation}
where
\begin{displaymath}
\tilde{\lambda}=-\left(\frac{1}{\varepsilon ^2}+1\right) \varepsilon  \sqrt{y_{12}^2+z_{12}^2-1}
   \sqrt{y_{34}^2+z_{34}^2-1}+2 y_{12} y_{34}+2 z_{12} z_{34}.
\end{displaymath}

It appears to be a challenging problem, using these relations 
 (\ref{relationship}),  (\ref{relationship1}) and   (\ref{relationship2}), to  transform 
the $\varepsilon$-parameterized volume forms and separability functions in the Lovas-Andai
framework to the $\mu$-parameterized ones in the Slater setting, and {\it vice versa}. (The
presence of the $z$ and $y$ variables in the formulas, undermining any immediate one-to-one 
relationship between $\varepsilon$ and $\mu$, is a complicating factor.)
 
The correlation between the $\varepsilon$ and $\mu$ variables, estimated on the basis of one million
randomly-generated (with respect to Hilbert-Schmidt measure) density matrices was 0.631937 in the two-rebit instance, and 0.496949 in the two-qubit one. 

Also, in these two sets of one million cases, $\mu$ was always larger than $\varepsilon$.
This dominance effect (awaiting formal verification) is reflected in Figs.~\ref{fig:JointRebitPlot} and \ref{fig:JointQubitPlot}, being plots
of the separability probabilities (again based on samples of size 5,077 and 3,715 million, respectively) as {\it joint} functions of $\varepsilon$ and $\mu$, with no results appearing in the regions $\varepsilon > \mu$.
\begin{figure}
\includegraphics{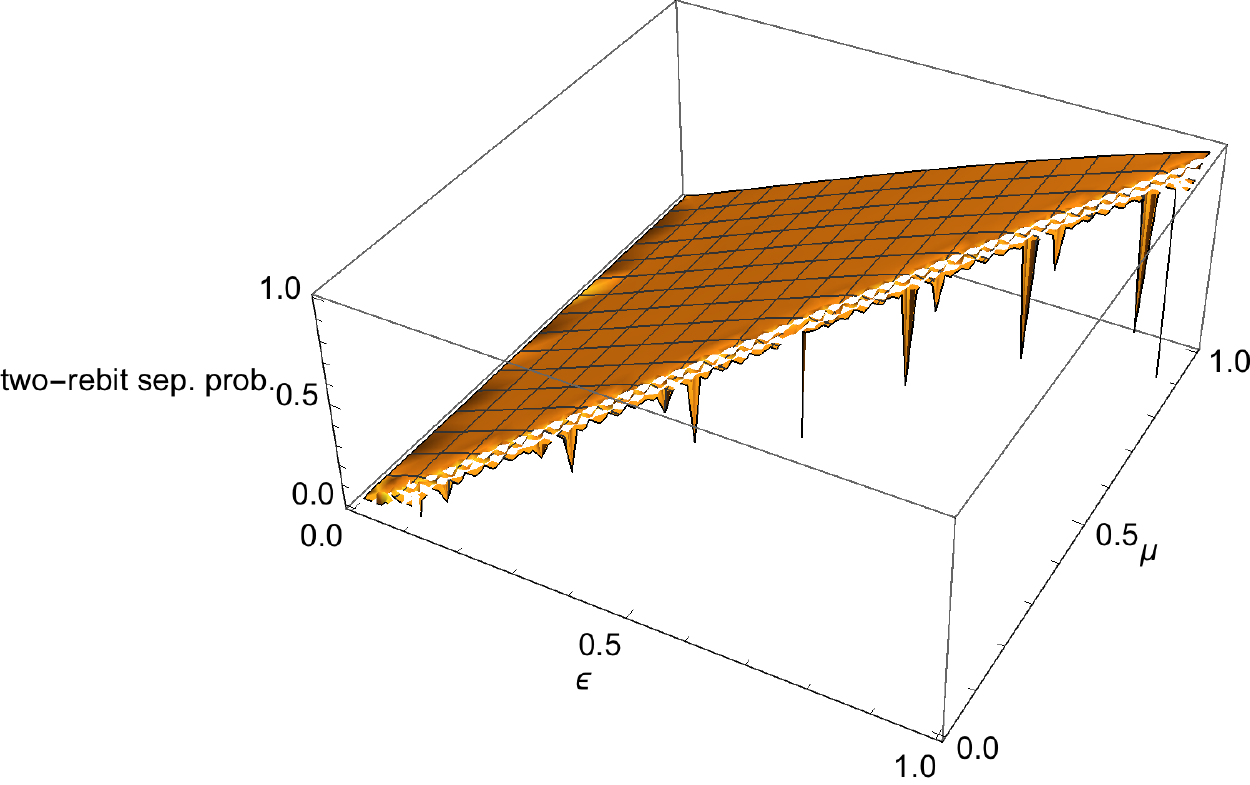}
\caption{\label{fig:JointRebitPlot}Two-rebit separability probabilities as joint
function of $\varepsilon$ and $\mu$, based on 5,077 million randomly-generated density matrices. 
Note the vacant region $\varepsilon > \mu$.}
\end{figure}
\begin{figure}
\includegraphics{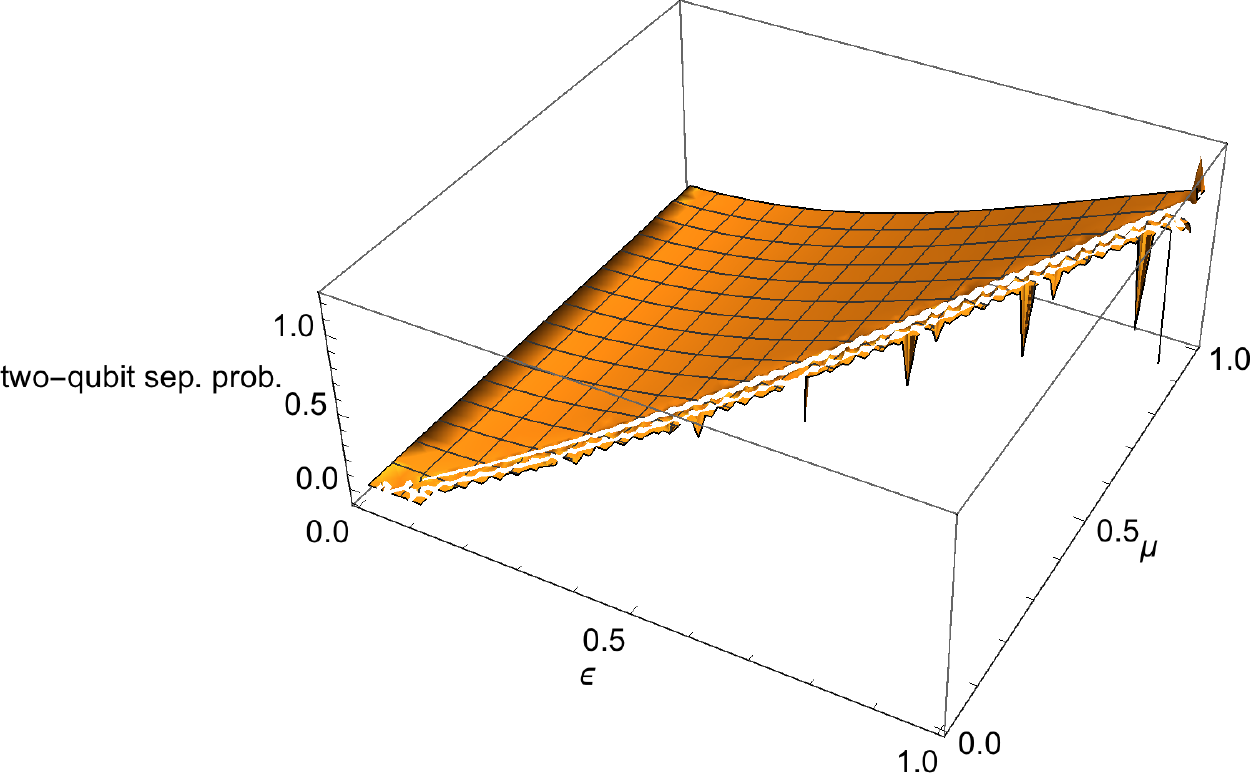}
\caption{\label{fig:JointQubitPlot}Two-qubit separability probabilities as joint
function of $\varepsilon$ and $\mu$, based on 3,715 million randomly-generated 
density matrices. Note the vacant region $\varepsilon > \mu$.}
\end{figure}
It has been noted 
(\url{https://mathoverflow.net/q/262943/47134} that for a diagonal $4 \times 4$ density matrix $D$ that $\varepsilon =\mu$ (inverting ratios, if necessary, so that both
are less than or greater than 1). This equality can also be observed by setting
$z_{12}=z_{24}=0$ (and $y_{12}=y_{24}=0$) in the equations immediate above.

Let us now display three plots that support, but only approximately, the possible
relevance of the Dyson-index ansatz for two-rebit and two-qubit separability functions.
In Fig.~\ref{fig:SlaterRatio}, we show the ratio of the square of the two-rebit separability
probabilities to the two-qubit separability probabilities, in terms of the variable
employed by Slater, $\mu=\sqrt{\frac{\rho_ {11} \rho_ {44}}{\rho_ {22} \rho_ {33}}}$.
In Fig.~\ref{fig:LovasAndaiRatio}, we show the Lovas-Andai counterpart, that is, in
terms of the ratio of singular values variable, $\varepsilon= \sigma(V)$. Further, in Fig.~\ref{fig:TwoDimensionalRatio}
we display the ratio of the {\it square} of the two-dimensional two-rebit plot (Fig.~\ref{fig:JointRebitPlot}) to the 
two-dimensional two-qubit plot (Fig.~\ref{fig:JointQubitPlot}). These three figures all manifest an upward trend
in the ratios as $\varepsilon$ and/or $\mu$ increase.
\begin{figure}
\includegraphics{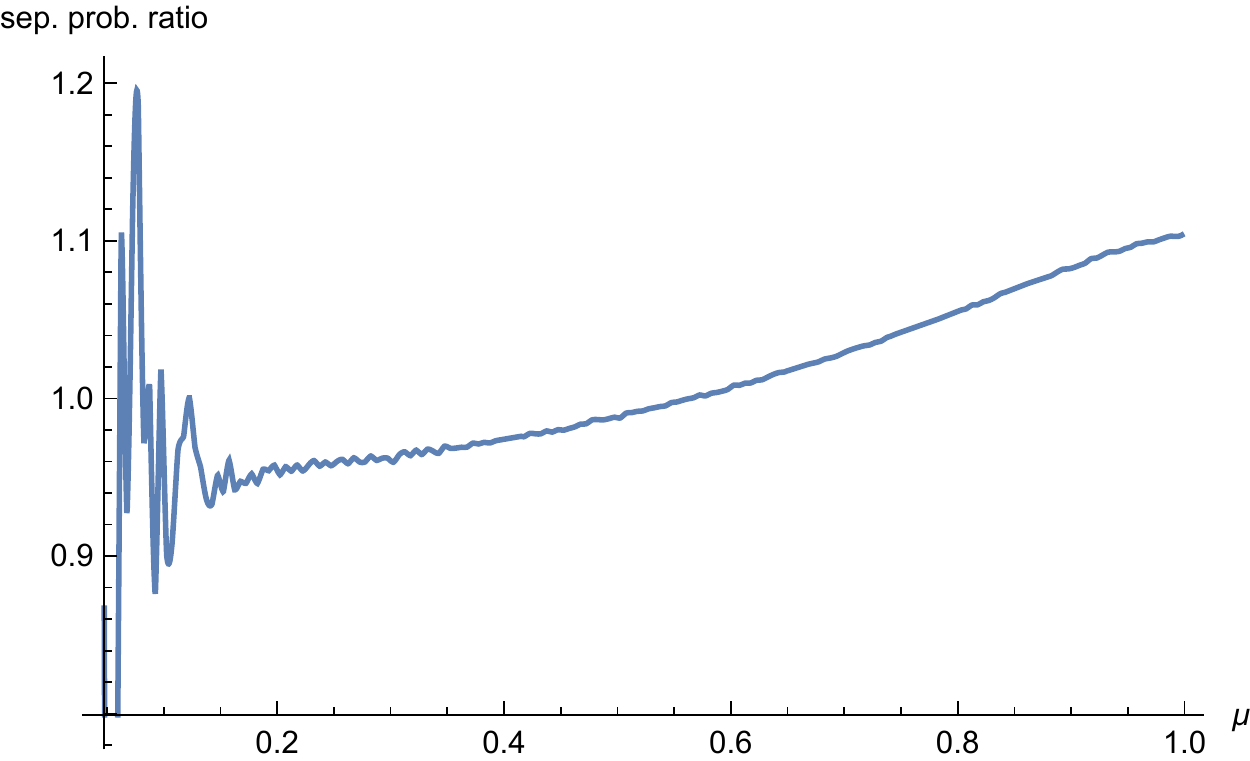}
\caption{\label{fig:SlaterRatio}Ratio of the {\it square} of the estimated two-rebit separability probabilities (Fig.~\ref{fig:RebitMU})
to the estimated two-qubit separability probabilities (Fig.~\ref{fig:QubitMU}), as a function of  $\mu=\sqrt{\frac{\rho_ {11} \rho_ {44}}{\rho_ {22} \rho_ {33}}}$}
\end{figure}
\begin{figure}
\includegraphics{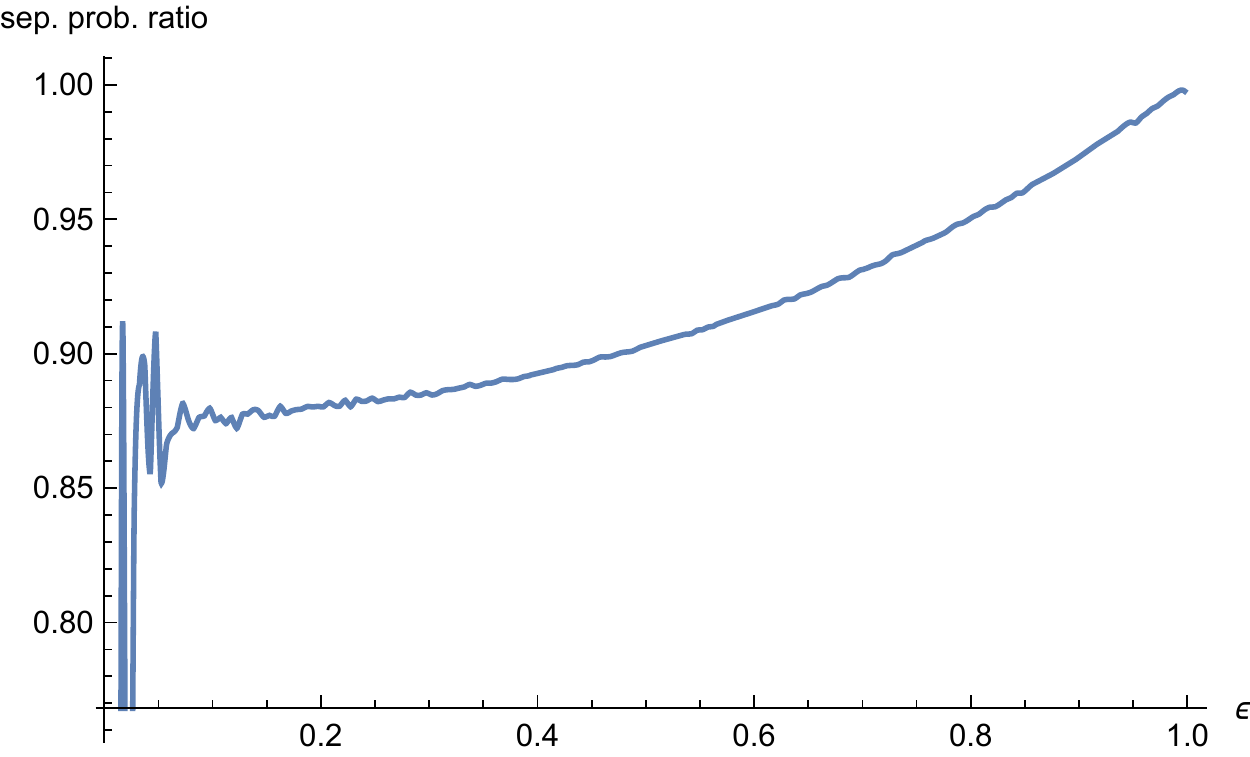}
\caption{\label{fig:LovasAndaiRatio}Ratio of the {\it square} of the estimated two-rebit separability probabilities (Fig.~\ref{fig:rebitsep})
to the estimated two-qubit separability probabilities (Fig.~\ref{fig:qubitsep}), as a function of the ratio of singular values variable, $\varepsilon= \sigma(V)$}
\end{figure}
\begin{figure}
\includegraphics{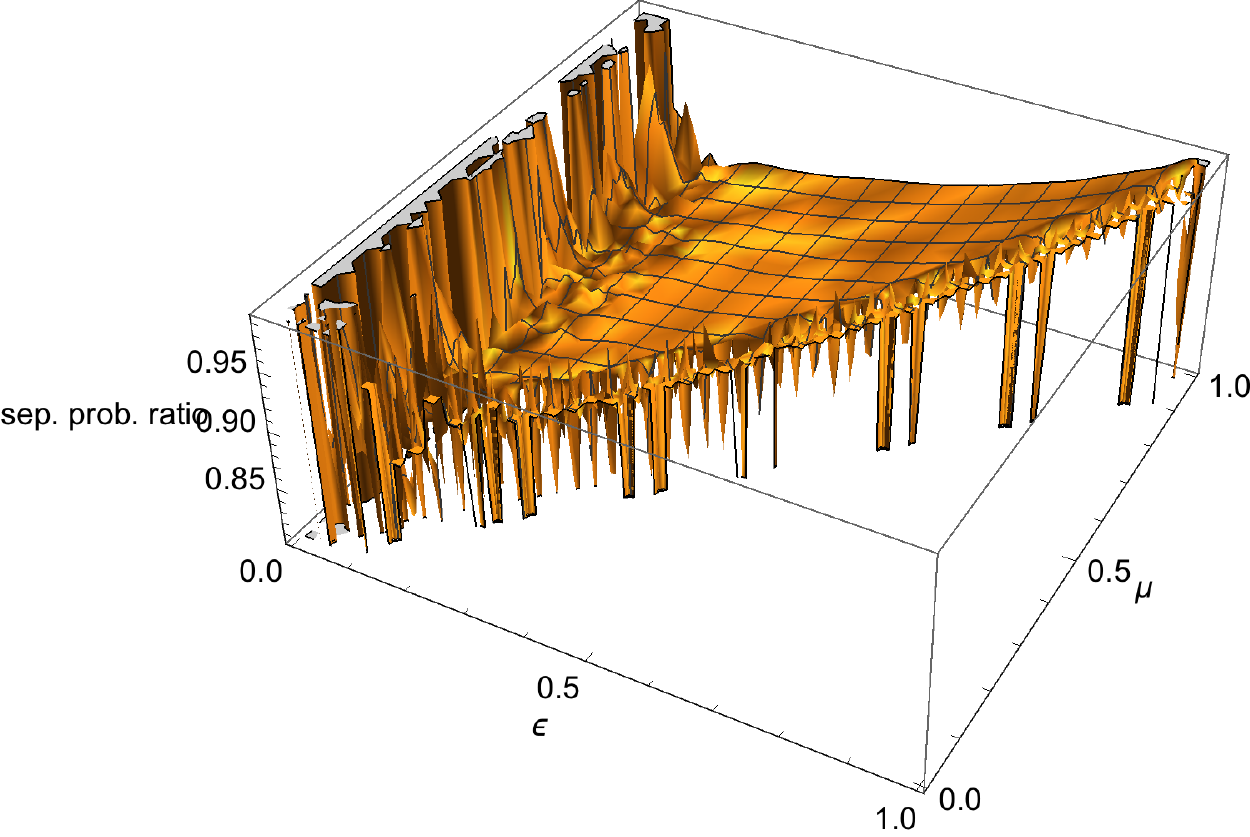}
\caption{\label{fig:TwoDimensionalRatio}The ratio of the {\it square} of the two-dimensional two-rebit plot (Fig.~\ref{fig:JointRebitPlot}) to the 
two-dimensional two-qubit plot (Fig.~\ref{fig:JointQubitPlot})}
\end{figure}
\section{Scenarios for which  $\varepsilon =\mu$ or $\varepsilon =\frac{1}{\mu}$}
\subsection{Seven-dimensional convex set of two-rebit states} \label{7Dsection}
If we set $z_{12}=z_{34}=0$ in the relation (\ref{Linkage}), we obtain $\varepsilon =\mu$ or $\frac{1}{\mu}$. So, let us try to obtain the separability function when these null conditions are fulfilled. First, we found that the Hilbert-Schmidt volume of the seven-dimensional convex set is equal to 
$\frac{1}{5040} \cdot \frac{2 \pi^2}{3} =\frac{\pi ^2}{7560} \approx 0.0013055$, with a jacobian for the transformation to $\mu$ equal to
\begin{equation} \label{sevendimJacobian}
\frac{\mu ^3 \left(-11 \mu ^6-27 \mu ^4+27 \mu ^2+6 \left(\mu ^6+9 \mu ^4+9 \mu
   ^2+1\right) \log (\mu )+11\right)}{210 \left(\mu ^2-1\right)^7}.
\end{equation}
 ($\frac{2 \pi^2}{3}$ is the normalization constant corresponding to $\chi_1(\varepsilon)$ \cite[Table 2]{lovasandai}, appearing in the 
 ``defect function'' (\ref{Defect}), as well as the volume of the standard unit ball in the normed vector space of $2 \times 2$ matrices with real entries, denoted by $\mathcal{B}_1(\mathbb{R}^{2 \times 2})$.)

We were, further, able to impose the condition that two of the principal $3 \times 3$ minors of the partial transpose are positive. 
The resultant separability function (Fig.~\ref{fig:SevenDimensionalSepProb})
was
\begin{equation} \label{7DimFunct}
\begin{cases}
 \frac{2 \left(\sqrt{\mu ^2-1}+\mu ^2 \csc ^{-1}(\mu )\right)}{\pi  \mu ^2} & \mu >1 \\
 \frac{2 \sqrt{1-\mu ^2} \mu +2 i \log \left(\mu +i \sqrt{1-\mu ^2}\right)+\pi }{\pi } &
   0<\mu <1
\end{cases},
\end{equation}
with an associated separability probability of $\frac{71}{105} \approx 0.67619$.
\begin{figure}
\centering
    \includegraphics{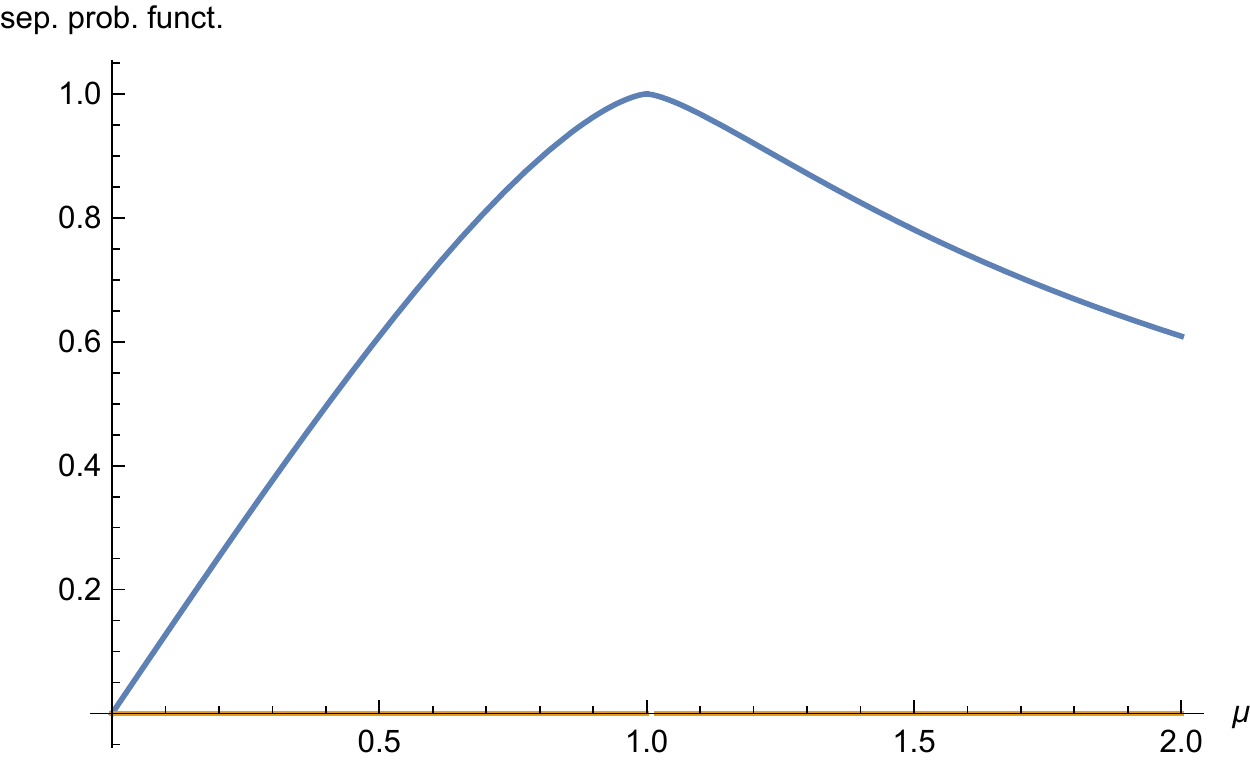}
    \caption{Two-rebit separability probability function (\ref{7DimFunct}) for the seven-dimensional convex set for which  $\varepsilon =\mu$ or $\frac{1}{\mu}$, based on the positivity of two
    principal $3 \times 3$ minors of the partial transpose}
    \label{fig:SevenDimensionalSepProb}    
\end{figure}

We, then, sought to impose--as both necessary and sufficient for separability \cite{asher,michal}--the positivity of the partial transpose of the density matrix. First, we found that the associated separability function assumes the value 1 at $\mu=1$. For $\mu=2, 3$, we formulated four-dimensional constrained integration problems. Mathematica reduced them to two-dimensional integration problems, for which we were able to perform  high precision 
calculations.
Remarkably, the values obtained agreed (using (\ref{poly})) with those for $\tilde{\chi}_1(\frac{1}{2})=\tilde{\chi}_1(2)$  and $\tilde{\chi}_1(\frac{1}{3})=\tilde{\chi}_1(3)$ to more than twenty decimal places. 
The two-dimensional integrands Mathematica yielded for $\mu=2$ were of the form
\begin{equation} \label{2Dintegrations1}
 \frac{3 \left(\pi  \sqrt{-4 z_{13}^2-z_{14}^2+4}-8 \sqrt{-z_{13}^2-z_{14}^2+1} \sin
   ^{-1}\left(\frac{z_{14}}{2 \sqrt{1-z_{13}^2}}\right)+2 \sqrt{-4 z_{13}^2-z_{14}^2+4}
   \sin ^{-1}\left(\frac{z_{14}}{\sqrt{1-z_{13}^2}}\right)\right)}{8 \pi ^2}  
\end{equation}
for 
\begin{displaymath}
-1<z_{13}<1\land z_{14}+\sqrt{1-z_{13}^2}>0\land z_{14}<0
\end{displaymath}
and 
\begin{equation} \label{2Dintegrations2}
 \frac{3 \left(\pi  \sqrt{-4 z_{13}^2-z_{14}^2+4}+8 \sqrt{-z_{13}^2-z_{14}^2+1} \sin
   ^{-1}\left(\frac{z_{14}}{2 \sqrt{1-z_{13}^2}}\right)-2 \sqrt{-4 z_{13}^2-z_{14}^2+4}
   \sin ^{-1}\left(\frac{z_{14}}{\sqrt{1-z_{13}^2}}\right)\right)}{8 \pi ^2}   
\end{equation}
for
\begin{displaymath}
-1<z_{13}<1\land z_{14}>0\land \sqrt{1-z_{13}^2}-z_{14}>0.
\end{displaymath}
So, in light of this evidence, we are confident in concluding that the Lovas-Andai two-rebit separability function $\tilde{\chi}_1(\varepsilon)$ serves as {\it both} the Lovas-Andai and Slater separability functions in this seven-dimensional setting. 

To still more formally proceed, we were able to generalize the pair of two-dimensional integrands for the specific case $\mu=2$ given in (\ref{2Dintegrations1}) and (\ref{2Dintegrations2}) 
to $\mu =1,2,3,\ldots$, obtaining
\begin{equation} \label{2Dintegrations3}
\frac{3 \left(2 \mu  \sqrt{-\mu ^2 z_{14}^2-z_{13}^2+1} \sin
   ^{-1}\left(\frac{z_{14}}{\sqrt{1-z_{13}^2}}\right)-2 \sqrt{-z_{13}^2-z_{14}^2+1} \sin
   ^{-1}\left(\frac{\mu  z_{14}}{\sqrt{1-z_{13}^2}}\right)+\pi 
   \sqrt{-z_{13}^2-z_{14}^2+1}\right)}{2 \pi ^2}   
\end{equation}
for
\begin{displaymath}
1<z_{13}<1\land z_{14}>0\land \sqrt{1-z_{13}^2}-\mu  z_{14}>0
\end{displaymath}
and
\begin{equation} \label{2Dintegrations4}
\frac{3 \left(-2 \mu  \sqrt{-\mu ^2 z_{14}^2-z_{13}^2+1} \sin
   ^{-1}\left(\frac{z_{14}}{\sqrt{1-z_{13}^2}}\right)+2 \sqrt{-z_{13}^2-z_{14}^2+1} \sin
   ^{-1}\left(\frac{\mu  z_{14}}{\sqrt{1-z_{13}^2}}\right)+\pi 
   \sqrt{-z_{13}^2-z_{14}^2+1}\right)}{2 \pi ^2}    
\end{equation}
for
\begin{displaymath}
-1<z_{13}<1\land \mu  z_{14}+\sqrt{1-z_{13}^2}>0\land z_{14}<0.
\end{displaymath}

\subsubsection{Reproduction of Lovas-Andai two-rebit separability function
$\tilde{\chi}_1 (\varepsilon )$} \label{Reproduction}
Making use of these last set of relations, we were able to reproduce the Lovas-Andai two-rebit separability function
$\tilde{\chi}_1 (\varepsilon )$, given in (\ref{poly}). We accomplished this by, first, reducing the 
(general for integer $\mu> 1$) two-dimensional 
integrands (\ref{2Dintegrations3}) and  (\ref{2Dintegrations4}) to two piecewise one-dimensional ones of the form
\begin{equation} \label{Verify1}
\frac{4 \left(\mu ^2 \sqrt{1-s^2} \sin ^{-1}\left(\frac{s}{\mu }\right)+\sqrt{\mu
   ^2-s^2} \cos ^{-1}(s)\right)}{\pi ^2 \mu ^2}    
\end{equation}
over $s \in [0,1]$
and 
\begin{equation} \label{Verify2}
\frac{2 \pi  \sqrt{\mu ^2-s^2}-4 \mu ^2 \sqrt{1-s^2} \sin ^{-1}\left(\frac{s}{\mu
   }\right)+4 \sqrt{\mu ^2-s^2} \sin ^{-1}(s)}{\pi ^2 \mu ^2}    
\end{equation}
over $s \in [-1,0]$.

To obtain these one-dimensional  integrands, which we then were able to explicitly evaluate, we made the substitution $z_{14}\to \frac{s \sqrt{1-z_{13}^2}}{\mu }$, then integrated over $z_{13} \in [-1,1]$, with $\mu \geq 1$, so that $\varepsilon =\frac{1}{\mu}$. Let us note that in this approach, the dependent variable ($\mu$) appears in the integrands, while in the Lovas-Andai derivation, the dependent 
variable ($\varepsilon$) appears as a limit of integration.
The counterpart set of two piecewise integrands to (\ref{Verify1}) and (\ref{Verify2}) for the reciprocal case of $0<\mu \leq 1$ are
\begin{equation}
\frac{4 \left(\mu ^2 \sqrt{1-s^2} \cos ^{-1}\left(\frac{s}{\mu }\right)+\sqrt{(\mu -s)
   (\mu +s)} \sin ^{-1}(s)\right)}{\pi ^2 \mu ^2}
\end{equation}
over $s \in [0,1]$ with $\mu > s$ and
\begin{equation}
\frac{2 \mu ^2 \sqrt{1-s^2} \left(2 \sin ^{-1}\left(\frac{s}{\mu }\right)+\pi \right)-4
   \sqrt{(\mu -s) (\mu +s)} \sin ^{-1}(s)}{\pi ^2 \mu ^2}
\end{equation}
over $s \in [-1,0]$ with $\mu >-s$. The corresponding univariate integrations then directly yield the Lovas-Andai two-rebit separability function $\tilde{\chi}_1 (\varepsilon )$, given in (\ref{poly}), now with
$\varepsilon =\mu$, rather than $\varepsilon =\frac{1}{\mu}$.

As an interesting aside, let us note that we can obtain $\varepsilon =\mu$ in (\ref{Linkage}), in a nontrivial fashion (that is, not just by taking 
$z_{12}=z_{34}=0$), by setting
\begin{equation}
z_{34}=    \frac{z_{12} \left(-2 \left(\mu ^3+\mu \right)+\mu ^4
   \left(-\sqrt{z_{12}^2-1}\right)+\sqrt{z_{12}^2-1}\right)}{\left(\mu ^2-1\right)^2
   z_{12}^2-\left(\mu ^2+1\right)^2},
\end{equation}
leading to an eight-dimensional framework. However, this result did not seem readily amenable to further study/analysis.
\subsection{Eleven-dimensional convex set of two-qubit states} \label{11Dsection}
Let us repeat  for the 15-dimensional convex set of two-qubit states, the successful form of analysis in the preceding section, 
again nullifying  the (1,2), (2,1), (3,4), (4,3) entries of $D$, so that the two diagonal $2 \times 2$ blocks $D_1, D_2$ are themselves diagonal. This leaves  us in an 11-dimensional setting. The associated volume we computed as  $\frac{1}{9979200} \cdot \frac{\pi^4}{6} = \frac{\pi ^4}{59875200} 
\approx 1.62687 \cdot 10^{-6}$.
(Here $\frac{\pi^4}{6}$ is the normalization constant corresponding to $\chi_2(1)$ \cite[Table 2]{lovasandai}, as well as the volume of the standard unit ball in the normed vector space of $2 \times 2$ matrices with complex entries, denoted by $\mathcal{B}_1(\mathbb{C}^{2 \times 2})$.) The associated jacobian for the transformation to the $\mu$ variable is
\begin{equation}
\frac{\mu ^5 \left(A (\mu -1) (\mu +1)-60 \left(6 \mu ^{10}+75 \mu ^8+200 \mu ^6+150 \mu ^4+30
   \mu ^2+1\right) \log (\mu )\right)}{83160 \left(\mu ^2-1\right)^{12}}
\end{equation}
with
\begin{displaymath}
A=5 \mu ^{10}+647 \mu ^8+4397 \mu ^6+6397 \mu ^4+2272 \mu ^2+142.
\end{displaymath}
The imposition of positivity for one of the $3 \times 3$ principal minors of the partial transpose yielded a separability function of 
$\frac{2 \mu ^2-1}{\mu ^4}$ for $\mu>1$, with an associated bound on the true separability probability of this set of eleven-dimensional two-qubit density matrices of $\frac{126}{181} \approx 0.696133$. (This function bears an interesting resemblance to the later reported important one (\ref{firsttake}). The insertion into (\ref{sepX}) of it, in its $0<\mu<1$ form, $\mu^2 (2 - \mu^2)$, leads to a separability probability 
prediction of $\frac{1}{3}$.)

Again--as in the immediately preceding seven-dimensional two-rebit setting--imposing, as both necessary and sufficient for separability \cite{asher,michal}, the positivity of the partial transpose of the density matrix, we find that the associated separability function assumes the value 1 at $\mu=1$. For $\mu=2$, our best estimate was 0.36848, which in line with the seven-dimensional analysis, would appear to be an approximation to the previously  unknown value of $\tilde{\chi}_2(\frac{1}{2})=\tilde{\chi}_2(2)$.
\subsubsection{Proof of the $\frac{8}{33}$-Two-Qubit Hilbert Schmidt Separability Probability Conjecture} \label{Verification}
We applied the Mathematica command GenericCylindricalDecomposition to an eight-variable set (plus $\mu$) of positivity conditions, enforcing the positive-definite nature of two-qubit ($4 \times 4$) density matrices (with their (1,2), (2,1), (3,4) and (4,3) entries nullified) and of their partial transposes, for $\mu>1$. (``GenericCylindricalDecomposition[ineqs,{$x_1,x_2,...$}]
finds the full-dimensional part of the decomposition of the region represented by the inequalities ineqs into cylindrical parts whose directions correspond to the successive $x_i$, together with any hypersurfaces containing the rest of the region.'')

These density matrices had their two $2 \times 2$ diagonal blocks, themselves set diagonal in nature. The parallel two-rebit analysis (sec.~\ref{7Dsection}) succeeded in reconstructing the Lovas-Andai function $\tilde{\chi_1} (\varepsilon )$, giving us confidence in this strategy. This pair of reduction strategies rendered the corresponding sets of density matrices as 11-dimensional and 7-dimensional in nature, rather than the standard full 15- and 9-dimensions, respectively. 

The cylindrical algebraic decomposition (CAD)--applied to the two-qubit positivity constraints (expressible in terms of $\mu$ and  four  [real part] $z_{ij}$ and four [imaginary part] $y_{ij}$ variables)--yielded three  complementary solutions. One of these consisted of three further complementary solutions. We analyzed each of the five irreducible solutions separately, employing them to perform integrations over the same set of four ($z_{23}, y_{23}, y_{24}$ and $z_{24}$) of the eight variables. Then, we summed the five results, remarkably simplifying to the four-dimensional integrand
\begin{equation} \label{FirstIntegrand}
\frac{12 \pi ^2 \left(-\left(\mu^2-1\right) y_{14}^2-\left(\mu^2-1\right)
   z_{14}^2+y_{13}^2+z_{13}^2-1\right) \left(\mu^2
   \left(y_{14}^2+z_{14}^2\right)+y_{13}^2+z_{13}^2-1\right)}{2 \pi^4
   \left(1-y_{13}^2+z_{13}^2\right)},
\end{equation}
subject to the constraints 
\begin{equation} \label{Constraints}
\mu>1\land -\frac{1}{\mu}<z_{14}<\frac{1}{\mu}\land -\frac{\sqrt{1-\mu^2
   z_{14}^2}}{\mu}<y_{14}<\frac{\sqrt{1-\mu^2 z_{14}^2}}{\mu}
\end{equation}   
\begin{displaymath}
   \land -\sqrt{1-\mu^2
   \left(y_{14}^2+z_{14}^2\right)}<y_{13}<\sqrt{1-\mu^2
   \left(y_{14}^2+z_{14}^2\right)}
\end{displaymath}
\begin{displaymath}
   \land -\sqrt{\mu^2
   \left(-\left(y_{14}^2+z_{14}^2\right)\right)-y_{13}^2+1}<z_{13}<\sqrt{\mu^2
   \left(-\left(y_{14}^2+z_{14}^2\right)\right)-y_{13}^2+1}.
\end{displaymath}
The transformation to a pair of polar coordinates 
\begin{equation}
\left\{z_{13}\to r_{13} \cos \left(\phi_{13}\right),z_{14}\to r_{14} \cos
   \left(\phi_{14}\right),y_{13}\to r_{13} \sin \left(\phi_{13}\right),y_{14}\to r_{14} \sin
   \left(\phi_{14}\right)\right\}.
\end{equation}
gave us a somewhat simpler integrand 
\begin{equation} \label{SecondIntegrand}
\frac{12 \pi ^2 r_{13} r_{14} \left(r_{14}^2 \mu^2+r_{13}^2-1\right) \left(-r_{14}^2
   \left(\mu^2-1\right)+r_{13}^2-1\right)}{2 \pi^4 \left(1-r_{13}^2\right)}.
\end{equation}
(Note four ``active'' variables in the first integrand (\ref{FirstIntegrand}), and only two radial and no angular ones in the second 
(\ref{SecondIntegrand}).)
The integration constraints (\ref{Constraints}) now  simply reduced to $r_{13}^2 +r_{14}^2 \mu^2 <1$, with $\mu>1$.  The integration result  
\begin{equation} \label{firsttake}
f(u)= \frac{4 \mu^2-1}{3 \mu^4},
\end{equation}
immediately followed.

Now, the function
that Lovai and Andai expressed hope in employing to verify  the  conjecture that the Hilbert-Schmidt two-qubit separability probability  is $\frac{8}{33} \approx 0.242424$, is
\begin{equation}  \label{VerifiedFormula}
\tilde{\chi_2}(\varepsilon)=f(\frac{1}{\varepsilon})=\frac{1}{3} \varepsilon ^2 \left(4-\varepsilon ^2\right).
\end{equation}
This can be seen since the denominator of the equation (\ref{sepC}) for $\mathcal{P}_{sep}(\mathbb{C})$  evaluates, as noted earlier, to 
\begin{equation} \label{Denominator}
\int\limits_{-1}^1\int\limits_{-1}^x  (1-x^2)^{2}(1-y^2)^{2} (x-y)^{2} \mbox{d} y\mbox{d} x=\frac{256}{1575},
\end{equation}
while the use of the newly-constructed $\tilde{\chi_2}(\varepsilon)$ yields a numerator value of 
\begin{equation}
\int\limits_{-1}^1\int\limits_{-1}^x  \tilde{\chi}_{2} \left(
\left.\sqrt{\frac{1-x}{1+x}}\right/ \sqrt{\frac{1-y}{1+y}}	
	\right)(1-x^2)^{2}(1-y^2)^{2} (x-y)^{2} \mbox{d} y\mbox{d} x=\frac{2048}{51975}
\end{equation}
with the ratio giving the  $\frac{8}{33}$ result. (For the reduced 11-dimensional two-qubit setting, making use of $\tilde{\chi_2}(\varepsilon)$, 
we were able to compute the associated 
separability probability as $\frac{746149}{21}-3600 \pi ^2 \approx 0.328918$. Proceeding similarly in the reduced 7-dimensional two-rebit setting,
we obtained an associated separability probability of 0.4197023.)

It is somewhat startling to compare the quite simple (polynomial) nature of $\tilde{\chi_2}(\varepsilon)$ with its two-rebit (polylogarithmic/inverse hyperbolic tangent) counterpart ((\ref{BasicFormula}), (\ref{poly})).
Let us now present (Fig.~\ref{fig:qubitdiffZ}) the two-qubit version of Fig.~\ref{fig:rebitsep}, showing again a random distribution of residuals, serving as further validation/support for the newly-constructed $\tilde{\chi_2}(\varepsilon)$.
\begin{figure}
    \centering
    \includegraphics{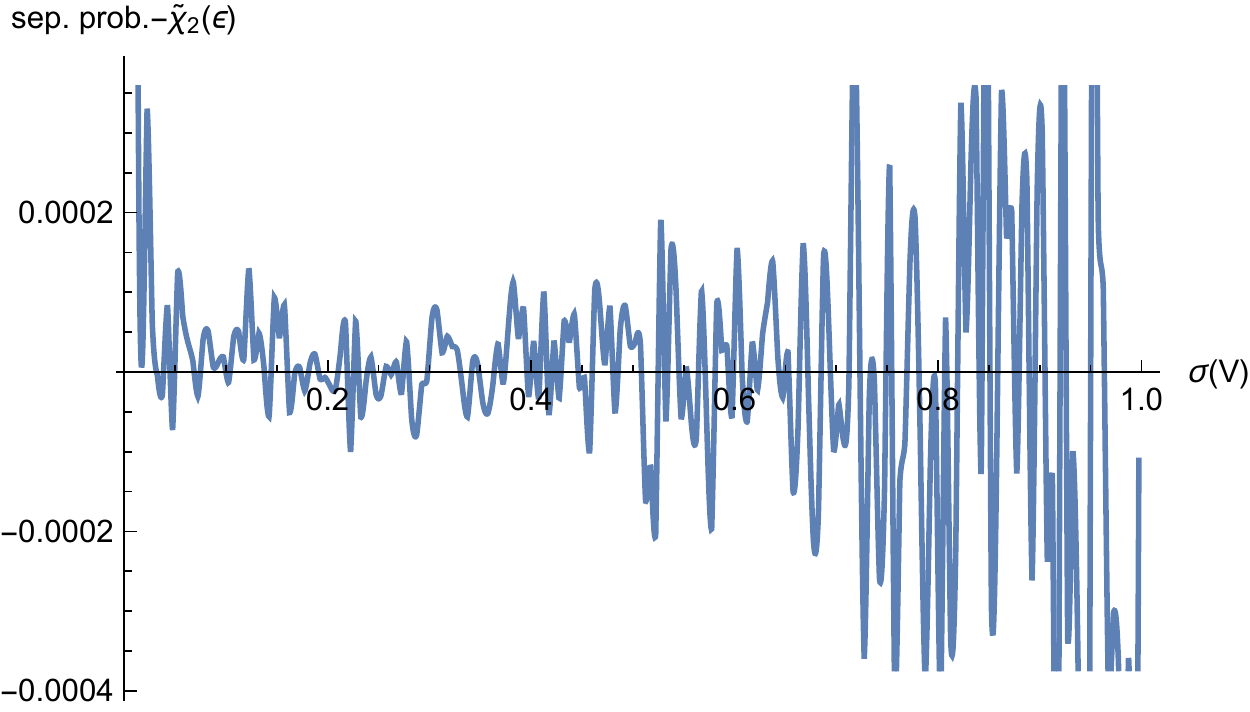}
    \caption{Result of subtracting $\tilde{\chi}_2 (\varepsilon )$ from the estimated two-qubit separability probability curve (Fig.~\ref{fig:qubitsep}). Fig.~\ref{fig:rebitdiff} is the two-rebit analogue.}
    \label{fig:qubitdiffZ}
\end{figure}

Let us interestingly note that it was conjectured in 2007 \cite[eqs. (93), (95); sec. 9.2]{slater833} that the ``two-qubit separability function'' (in the Slater framework) had the form
\begin{equation}
\frac{6}{71} (3 -\mu^2) \mu^2,
\end{equation}
somewhat similar in nature to (\ref{VerifiedFormula}) (cf. Fig.~\ref{fig:QubitMU}, ~\ref{fig:QubitMUDiff}).

A formidable challenge, to continue this line of research, is now to establish that the "two-quater[nionic]bit" Hilbert-Schmidt separability probability is $\frac{26}{323}$. This would move us, first, from the original 9-dimensional two-rebit and  15-dimensional two-qubit settings to a 27-dimensional one.  But these dimensions can be reduced to 7-, 11- and 19-, using the apparently acceptable strategy--that has given us  $\tilde{\chi_1} (\varepsilon )$ and $\tilde{\chi_2} (\varepsilon )$--of setting the two $2 \times 2$ diagonal blocks themselves to diagonal form. In turn, this leads to cylindrical algebraic decompositions with 4, 8 and 16 variables--with the last, quaternionic one, still seemingly computationally unfeasible.  (Theorem 3 of \cite{andai2006volume} yields $\frac{\pi ^{12}}{315071454005160652800000}=\frac{\pi ^{12}}{2^{15} \cdot 3^{10} \cdot 5^5 \cdot 7^3 \cdot 11^2 \cdot 13^2 \cdot 17 \cdot 19 \cdot 23}$ for the volume of the state space of quaternionic $4 \times 4$ density matrices.) 

It should be pointed out that the manner of derivation of $\tilde{\chi_2} (\varepsilon )$ here is distinctly different from that employed by Lovas and Andai  \cite[App. A]{lovasandai} in obtaining the form of $\tilde{\chi_1} (\varepsilon )$, though it has also been able to find this result here using the cylindrical algebraic decomposition approach (sec.~\ref{Reproduction}). 
\subsection{Common features of two-rebit, two-qubit and two-quaterbit constraint sets}
To begin, let us make the temporary change of notation,
\begin{equation}
 z_{13}=r_{13}, z_{14}=r_{14}, z_{23}=r_{23}, z_{24}=r_{24}.   
\end{equation}
To construct the two-rebit function $\tilde{\chi_1} (\varepsilon )$, in our mode of analysis above,
we need, first,  to ensure the positive-definiteness of the associated $4 \times 4$ real-entry density matrix $D$ (with its (1,2), (2,1), (3,4), (4,3)-entries nullified). To accomplish this, 
we must enforce the pair of constraints
\begin{equation} \label{Constraint0}
 -r_{13}^2-r_{23}^2+1>0   
\end{equation}
and
\begin{equation} \label{Constraint1}
-r_{13}^2-2 r_{13} r_{14} r_{23} r_{24}-r_{14}^2+r_{14}^2
   r_{23}^2-r_{23}^2+\left(r_{13}^2-1\right) r_{24}^2+1>0.
\end{equation}
Additionally, to ensure the positive-definiteness  of its partial transpose (and, thus, the separability of $D$ \cite{augusiak2008universal}),
we must enforce the constraint
\begin{equation} \label{Constraint2}
 \mu ^4 \left(-r_{14}^2\right)+\mu ^2 \left(-r_{13}^2-2 r_{13} r_{14} r_{23}
   r_{24}+r_{14}^2 r_{23}^2+\left(r_{13}^2-1\right) r_{24}^2+1\right)-r_{23}^2 > 0.
\end{equation}

Now, moving on to the two-qubit case, let us employ polar coordinates of the form, 
\begin{equation}
z_{13} =r_{13} \cos{\phi_{13}}, y_{13} =r_{13} \sin{\phi_{13}},\ldots.
\end{equation}
Then, the constraint (\ref{Constraint0}) remains as is,
while both (\ref{Constraint1}) and  (\ref{Constraint2}) are modified by replacing the term
$2 r_{13} r_{14} r_{23} r_{24}$ by
\begin{equation} \label{lastterm}
2 r_{13} r_{14} r_{23} r_{24} \cos \left(\phi_{13}-\phi_{14}-\phi_{23}+\phi_{24}\right).
\end{equation}
The range of $C_{\mathbb{C}}=\cos \left(\phi_{13}-\phi_{14}-\phi_{23}+\phi_{24}\right)$ is $[-1,1]$.
(So, implicitly, $C_{\mathbb{R}}=1$.)

The same set of three constraints, as in this two-qubit case, holds as well (generalizing from polar to 
hyperspherical coordinates \cite{10.2307/2308932}) in the two-quaterbit case \cite{hildebrand2008semidefinite}, but for the 
replacement in two of the three constraints of the factor $C_{\mathbb{C}}$
by an (apparently) much more cumbersome term $C_{\mathbb{Q}}$. (We employ the ``Moore determinant" for our
quaternionic calculations \cite{moore1922determinant}.) This term is composed of twelve angular, and again none of the four 
radial variables. (However, simulations strongly indicate that the range of this expression is also $[-1,1]$.)

Further, we observe that only the two even powers of $\mu$, that is $\mu^2$ and $\mu^4$, appear in the constraints.
In the two-rebit case, the corresponding integrand in the multidimensional integration is simply 1, while in the two-qubit case it would be
$r_{13} r_{14} r_{23} r_{24}$, and in the two-quaterbit instance, it would be $(r_{13} r_{14} r_{23} r_{24})^3$.
\subsubsection{Attempted construction of  $\tilde{\chi_4} (\varepsilon )$} \label{Attempted1}
We were able to again obtain the two-qubit formula (\ref{VerifiedFormula}) for  $\tilde{\chi_2} (\varepsilon )$, using the just indicated
set of three integration constraints. So, we naturally attempted to extend the scheme of analysis to the two-quaterbit
case. 

Our parallel calculation, then, yielded (here, as a beginning exercise, we take $C_{\mathbb{C}}=C_{\mathbb{Q}}$)
\begin{equation} \label{QuaterbitStart}
   \tilde{\chi_4} (\varepsilon)=  -\frac{1}{385} \varepsilon ^4 \left(75 \varepsilon ^4+128 \varepsilon ^2-588\right)
\end{equation}
(again a function of $\varepsilon^2$). But when we substituted this function into our ansatz (\ref{sepX}), with $\alpha =2$, we obtain a separability probability result of 
$\frac{58}{969} \approx 0.0598555$, rather than the strongly-supported value of $\frac{26}{323} \approx 0.0804954$  \cite[eq. (4)]{slaterJModPhys}. 
(An effort to similarly study the [non-division-algebra] $\alpha =\frac{3}{2}$ case, with an integrand of  $(r_{13} r_{14} r_{23} r_{24})^2$, having a presumed separability probability of $\frac{36061}{262144}$, led to intractable integrals involving elliptic functions.)
So, we must conclude that either our ansatz (\ref{sepX}), for this particular case, is not proper or perhaps more likely
that our treatment of the twelve-angular-variable factor (ranging from -1 to 1) as the cosine function of a single
variable, in the same manner as the two-qubit term  $C_{\mathbb{C}}$ is, has led us astray. (We found specific support for the ansatz in this case by computing the volume form for the $2 \times 2$ self-adjoint {\it quaternionic} matrices, denoted by $\mathcal{M}^{sa}_{2 \mathbb{Q}}$ in the notation of \cite[Table 1]{lovasandai}. It was of the form $\frac{1}{4} \sin (2 \eta ) \sin ^2(\theta ) (x-y)^4 \sin ^3(\phi ) \cos (\phi )$.)  Certainly, it would be desirable to obtain a concise re-expression of this twelve-angular-variable term.

In our efforts in this latter regard, we have obtained an interesting concise expression of certain nine-dimensional sections of the twelve-dimensional body $C_{\mathbb{Q}}$.
The (1,3), (1,4), (2,3) and (2,4) quaternionic entries (of absolute value no greater than 1) of the $4 \times 4$ density matrix $D$ were each parameterized, in the standard hyperspherical manner,
by a single radial and three angular variables. There are six possible pairs of such entries. If we equated the three angles of the (1,3)-(1,4) or (1,3)-(2,3) or (1,4)-(2,4) or (2,3)-(2,4) pair (thus, reducing the dimensionality from twelve to nine), then the twelve-dimensional (radii-free) term collapses greatly to the {\it six}-variable (conditional) expression,
\begin{equation} \label{vanishing}
\cos{\theta_1}  \cos{\theta_2}  +\sin{\theta_1}  \sin{\theta_2}(\cos{\eta_1} \cos{\eta_2}+\cos{(\phi_1-\phi_2)} \sin{\eta_1} \sin{\eta_2}).
\end{equation}
Here, the $\theta$'s and $\eta$'s are the corresponding two latitudinal angles (varying from 0 to $\pi$), and the $\phi$'s the corresponding  longitudinal angles (varying from 0 to $2 \pi$) of the {\it unmatched} pair, the three equated angles of the matched pair {\it vanishing} from the expression (\ref{vanishing}). If the three angles of {\it both} pairs are equated, the expression (\ref{vanishing}) further reduces to -1.

The {\it marginal} distribution (integrating out the eight latitudinal angles) of $C_{\mathbb{Q}}$ over the longitudinal angles equalled $\frac{16 \pi^4}{81}  C_{\mathbb{C}}$. (The relations between conditional, marginal and full/joint distributions comprise a research topic of considerable interest 
\cite{arnold1989compatible,gelman1993characterizing}.)

The term, $-2 r_{13} r_{14} r_{23} r_{24} C_{\mathbb{Q}}$ arising in the constraints stems from summing two identical terms in the expansion of the Moore determinant of
the two-quaterbit density matrix with, in our now usual analytical manner, its (1,2), (2,1), (3,4), (4,3) entries nullified. One of the $24 =4!$ terms in the expansion corresponds to minus the {\it ordered} product of the (1,4), (4,2), (2,3) and (3,1) entries, while the second identical term corresponds to minus the ordered product of the (1,3), (3,2), (2,4) and (4,1) entries. The two-rebit and two-qubit constraint terms  $-2 r_{13} r_{14} r_{23} r_{24} C_{\mathbb{R}}= -2 r_{13} r_{14} r_{23} r_{24}$ and  $-2 r_{13} r_{14} r_{23} r_{24} C_{\mathbb{C}}$ arise similarly, but the ordering of the entries in the product becomes irrelevant.

If we assume that the Lovas-Andai two-quaterbit function has the polynomial structure exhibited by (\ref{QuaterbitStart}) then we could obtain the presumed separability probability $\frac{26}{323}$ using the remarkably simple function
\begin{equation} \label{SimpleQuaterbit}
 \tilde{\chi_4} (\varepsilon)= \varepsilon ^4 \left(2-\varepsilon ^4\right). 
\end{equation}
Slightly less simple, but also similarly successful would be
\begin{equation}
  \tilde{\chi_4} (\varepsilon)= \epsilon ^4 \left(-26 \epsilon ^4+32 \epsilon ^2-5\right).   
\end{equation}
We wanted to test the fit of these last two functions by the generation of random two-quaterbit matrices--but had not yet found an effective manner of doing so (cf. \cite{osipov,mittelbach2012sampling,wang2013structure}). (The possible use of Ginibre ensembles, in the manner of \cite{osipov}, and the associated issues in doing so, has been addressed by C. F. Dunkl in App.~\ref{CholGini}.) The computational resources employed by Fei and Joynt in this regard greatly exceeded those available to us \cite{FeiJoynt}.
\subsubsection{Attempted construction of  $\tilde{\chi_8} (\varepsilon )$ for the presumptive octonionic case ($\mathbb{K}=\mathbb{O}$)} \label{AttemptOcto}
We formally repeated, now with $d=8$, the form of analysis used in our attempted construction (sec.~\ref{Attempted1}) of  $\tilde{\chi_4} (\varepsilon )$.
So, we (again, naively) assumed that we had a radial-variable-free factor of the form $C_{\mathbb{O}}=C_{\mathbb{C}}=\cos \left(\phi_{13}-\phi_{14}-\phi_{23}+\phi_{24}\right)$, and further used as our jacobian $(\rho_{11} \rho_{22} \rho_{33} \rho_{44})^7$, rather than $(\rho_{11} \rho_{22} \rho_{33} \rho_{44})^3$ (cf. \cite{slateroctonionic}). Then, the pseudo-Lovas-Andai function we obtained was
\begin{equation}
   \tilde{\chi_8} (\varepsilon)= \frac{\varepsilon ^8 \left(2458624 \varepsilon ^2-45 \left(847 \varepsilon ^4+17408 \varepsilon
   ^2+24192\right) \varepsilon ^4+1159340\right)}{1707849}.
\end{equation}
This led to a separability probability estimate of $\frac{11410114}{5429220123} \approx 0.00210161$, while the formal prediction 
\cite{slaterJModPhys} is
$\frac{44482}{4091349} \approx 0.0108722$, with a consequent ratio of these two values of $\frac{5705057}{29513807} \approx 0.193301$. 

A function that does produce $\frac{44482}{4091349}$ is
\begin{equation}
\tilde{\chi_8} (\varepsilon)=\varepsilon ^8 \left(-305 \varepsilon ^8+131 \varepsilon ^6+5 \varepsilon ^4+5
   \varepsilon ^2+3\right).
\end{equation}
\section{X-states functions conform to Dyson-index ansatz}
Since the (1,2)-, (2,1)-, (3,4)-, (4,3)-entry nullified two-quaterbit 
problem was proving challenging, it seemed of interest to investigate what arises
when we additionally nullify the (1,3)-, (3,1)-, (2,4)-, (4,2)-entries, reducing
then to the X-states framework \cite{Xstates2,khvedelidze2016spectrum}.
I, first, examined the two-rebit and two-qubit problems, and then went on
to the two-quaterbit question.
I was able to reproduce the corresponding $X$-statess PPT-probabilities 
$\frac{16}{3 \pi^2}, \frac{2}{5}$ and $\frac{2}{7}$
reported in  \cite{dunkl2015separability}.

Now, the interesting finding, certainly conforming 
to the Dyson-index ansatz, is that the 
functions of the variable ($\varepsilon$) that in the Lovas-Andai
framework is  the singular-value ratio,  are simply
\begin{equation}
\varepsilon, \varepsilon^2, \varepsilon^4
\end{equation}
in these three scenarios. (Let us note that $\varepsilon^2$ is a factor of 
 $\tilde{\chi_2} (\varepsilon ) = \varepsilon^2 (4-\varepsilon^2)/3$ and  $\varepsilon^4$ is a factor of our 
 then conjectured  $\tilde{\chi_4} (\varepsilon ) = \varepsilon^4 (2 -\varepsilon^4)$. ``It is somewhat interesting
 that the identity function approximates well  $\tilde{\chi_1} (\varepsilon ))$'' \cite[p. 6]{lovasandai}. In the 
 X-states case, we see that the relationship is an exact one.)

If we insert these functions into (\ref{sepX}),
rather than the true 
X-states PPT-probabilities, we obtain
$\frac{16}{9}- \frac{35 \pi^2}{256} \approx 0.428418, \frac{13}{66} \approx 0.19697$ and $\frac{124}{2907} \approx
0.0426557$, all smaller than the true values.

Also, if we formally proceed with the (non-division-agebra-based) case $\alpha=\frac{3}{2}$, by setting one 
of the three non-real components of the quaternionic off-diagonal entries to zero, we get
(still conforming to the ansatz), $\varepsilon^{\frac{3}{2}}$. The corresponding
X-states PPT-probability is $\frac{1024}{315 \pi^2} \approx 0.329374$, while the 
use of $\varepsilon^{\frac{3}{2}}$ in (\ref{sepX}), with   $\alpha=\frac{3}{2}$, gives $\frac{2816}{147}-\frac{129633075 \pi ^2}{67108864} 
\approx 0.0915121$.
\section{Formula for the Lovas-Andai two-quaterbit function  $\tilde{\chi_4} (\varepsilon)$}
The simple nature of the results of the X-states analyses led us to consider the slightly expanded/intermediate 
scenarios in which either the two members of the (1,3)-, (3,1)- or of the (2,4)-, (4,2)- pairs are not constrained to 
zero. We were quite surprised to find that in both the two-rebit and two-qubit cases, the associated 
separability functions were precisely  the previously found $\tilde{\chi_1} (\varepsilon)$ and  $\tilde{\chi_2} (\varepsilon)$.
(So, it appeared that we could obtain these functions by nullifying as many as three pairs of off-diagonal
entries, rather than just two, as had been our strategy up until this point in time.)

This encouraged us to similarly examine the two-quaterbit case. We employed hyperspherical coordinates for the four
components of each quaternion. Then, the positivity constraints to be enforced only involved the three radial, 
and none of the angular variables.
We rather readily arrived at the result
\begin{equation}
\tilde{\chi_4} (\varepsilon) = \frac{1}{35} \varepsilon ^4 \left(15 \varepsilon ^4-64 \varepsilon ^2+84\right).
\end{equation}
Substitution of this function into the ansatz (\ref{sepX}), with $\alpha =2$, gave us a numerator of
\begin{equation}
\int\limits_{-1}^1\int\limits_{-1}^x  \tilde{\chi}_{4} \left(
\left.\sqrt{\frac{1-x}{1+x}}\right/ \sqrt{\frac{1-y}{1+y}}	
	\right)(1-x^2)^{4}(1-y^2)^{4} (x-y)^{4} \mbox{d} y\mbox{d} x = \frac{1048576}{430890075}
\end{equation}
and a denominator of 
\begin{equation}
\int\limits_{-1}^1\int\limits_{-1}^x  (1-x^2)^{4}(1-y^2)^{4}(x-y)^{4}  \mbox{d} y \mbox{d} x =\frac{524288}{17342325}
\end{equation}
yielding the relatively long-standing conjecture \cite{FeiJoynt2,slaterJModPhys,MomentBased} of
\begin{equation}
\mathcal{P}_{sep/PPT}(\mathbb{Q})= \frac{26}{323}.
\end{equation}
(Charles Dunkl has pointed out that $524288 =2^{19}$, $1048576 =2^{20}$, $17342325= 3^2 \cdot 5^2 \cdot 7^2 \cdot 11^2 \cdot 13$, and $430890075 = 3^2 \cdot 5^2 \cdot 7^2 \cdot 11^2 \cdot 17 \cdot 19$.)

Now we will be able to further pursue this line of approach (of nullifying {\it three } off-diagonal pairs) 
to find the Lovas-Andai functions 
$\tilde{\chi_d} (\varepsilon)$, {\it in full generality}.
\section{General Construction of the Lovas-Andai Formulas} \label{GeneralConstruction}
In their interesting, important 
Conclusion section, Lovas and Andai write: ``The structure of the unit ball in operator norm of $2\times 2$ matrices plays a 
  critical role in separability probability of qubit-qubit and rebit-rebit 
  quantum systems.
It is quite surprising that the space of $2\times 2$ real or complex matrices 
  seems simple, but to compute the volume of the set
\begin{equation*}
\Big\{\begin{pmatrix}a & b\\ c& e\end{pmatrix} \Big\vert\ a, b, c, e \in \mathbb{K},
 \norm{\begin{pmatrix} a & b\\ c& e\end{pmatrix}} <1,\ \
 \norm{\begin{pmatrix} a & \varepsilon b\\ \frac{c}{\varepsilon}& e
\end{pmatrix}}  <1 \Big\}
\end{equation*}
  for a given parameter $\varepsilon\in [0,1]$, which is the value of 
  the function $\chi_{d}(\varepsilon)$, is a very challenging problem.
The gist of our considerations is that the behavior of the function 
  $\chi_{d}(\varepsilon)$ determines the separability probabilities with respect
  to the Hilbert-Schmidt measure.'' (The operator norm $\norm{\cdot}$ is the largest singular value or Schatten-$\infty$ norm. Let us note 
  that Gl{\"o}ckner studied functions on the quaternionic unit ball \cite{glockner2001functions}.)
  
It appears that the cylindrical-algebraic-decomposition approach we have applied to $4 \times 4$ density matrices $D$ with diagonal $2 \times 2$ diagonal blocks $D_1, D_2$ (and even one additional pair of nullified entries) to obtain the trio $\tilde{\chi_1} (\varepsilon )$,  $\tilde{\chi_2} (\varepsilon )$ and $\tilde{\chi_4} (\varepsilon )$  specifically answers this ``very challenging'' question, but in a manner quite different than Lovas and Andai applied in deriving $\tilde{\chi_1} (\varepsilon )$ \cite[App. A]{lovasandai}.

Our analyses, however, now are able to reveal that this problem has an extremely succinct formulation.
We employ the set of constraints (imposing--in quantum-information-theoretic terms--the positivity of the density matrix and its partial transpose),
\begin{equation}
r_{23}^2<1\land \left(r_{14}^2-1\right) \left(r_{23}^2-1\right)>r_{24}^2\land r_{23}^2
   \left(\varepsilon ^2 r_{14}^2-1\right)>\varepsilon ^2 \left(\varepsilon ^2
   r_{14}^2+r_{24}^2-1\right).   
\end{equation}
 Then, subject to these constraints, we have to integrate the jacobian (corresponding to the hyperspherical parameterization
of the three off-diagonal non-nullified entries of the density matrix) $\left(r_{14} r_{23} r_{24}\right){}^{d-1}$
over the unit cube $[0,1]^3$. Dividing the result of the integration by
\begin{equation} \label{denominator}
\frac{\pi  4^{-d} \Gamma \left(\frac{d}{2}+1\right)^2}{d^3 \Gamma
   \left(\frac{d+1}{2}\right)^2},
\end{equation}
yields the desired $\tilde{\chi_d} (\varepsilon )$. (If we take $r_{24}=0$, and a jacobian of 
$\left(r_{14} r_{23}\right){}^{d-1}$, we revert to the X-states setting, and obtain simply 
$\varepsilon^d$ as the corresponding  function.)

This last result (\ref{denominator}) is obtained by integrating the same jacobian 
$\left(r_{14} r_{23} r_{24}\right){}^{d-1}$ over the unit cube,
subject to the constraints (imposing the positivity of the density matrix),
\begin{equation}
r_{23}^2<1\land \left(r_{14}^2-1\right) \left(r_{23}^2-1\right)>r_{24}^2.
\end{equation}

In fact, we have found that for $d=8$, presumptively corresponding
to the fourth division algebra, the octonions ($\mathcal{O}$) (cf. sec.~\ref{AttemptOcto}), the function
\begin{equation}
\tilde{\chi_8} (\varepsilon )=\frac{1}{1287}  \varepsilon ^8 \left(1155 \varepsilon ^8-7680 \varepsilon ^6+20160 \varepsilon
   ^4-25088 \varepsilon ^2+12740\right),   
\end{equation}
yielded through the indicated pair of three-dimensional constrained integrations, does give the value $\frac{44482}{4091349}$, generated 
in previous studies \cite{slater833,slaterJGP2,JMP2008,ratios2}.

Further, we have confirmed through symbolic means for $d=6,10,12,14,16,18,20$ (in addition, to
the cases $d=2,4,8$ already studied) and numerical means for $d=3,5 $ (in adddition, to the two-rebit 
case $d=1$)
that this methodology, reproduces the corresponding exact (rational) values (parameterized by
$\alpha =\frac{d}{2}$) reported in \cite[eqs. (1)-(3)]{slaterJModPhys} (also eqs.~(\ref{Hou1})-\ref{Hou3})) above), based on a 
certain ``concise'' formulation of a hypergeometric-based expression. 

For even $d$, the constant term in $\tilde{\chi_d} (\varepsilon )$ is 
given by
\begin{equation}
\frac{8^d \Gamma \left(\frac{d+1}{2}\right)^3}{\pi ^{3/2} \Gamma \left(\frac{3
   d}{2}+1\right)},
\end{equation}
while the coefficient of $\varepsilon^2$ is 
\begin{equation}
-\frac{2^{3 d-1} 3^{-\frac{3}{2} (d+1)} d^3 \Gamma
   \left(\frac{d+1}{2}\right)^3}{\sqrt{\pi } \Gamma \left(\frac{d}{2}+\frac{2}{3}\right)
   \Gamma \left(\frac{d}{2}+\frac{4}{3}\right) \Gamma \left(\frac{d}{2}+2\right)}.
\end{equation}
We have obtained further formulas in this series up to the coefficient of $\varepsilon^{22}$.

Further, we have (for both even and odd values of $d$),
\begin{equation} \label{Almost}
\tilde{\chi_d} (\varepsilon )=
\end{equation}
\begin{displaymath}
\frac{\varepsilon ^d \Gamma (d+1) \left(\Gamma (d+1)^2 \,
   _3\tilde{F}_2\left(-\frac{d}{2},\frac{d}{2},d;\frac{d}{2}+1,\frac{3
   d}{2}+1;\varepsilon ^2\right)+2 \,
   _2F_1\left(-\frac{d}{2},\frac{d}{2};\frac{d+2}{2};\varepsilon ^2\right)\right)}{2
   \Gamma \left(\frac{d}{2}+1\right)^2}-
\end{displaymath}
\begin{displaymath}
\frac{4^d d \varepsilon ^d \Gamma \left(\frac{d+1}{2}\right)^2}{\pi  \Gamma
   \left(\frac{d}{2}+1\right)^2} \int_0^1 r_{14}^{2 d-1} \, _2F_1\left(-\frac{d}{2},\frac{d}{2};\frac{d+2}{2};r_{14}^2\right)
   \left(1-\varepsilon ^2 r_{14}^2\right){}^{d/2} d r_{14}
\end{displaymath}
(the tilde indicating regularization).
\subsection{Master formula for $\tilde{\chi_d} (\varepsilon )$}
An integration-by-parts 
(\url{https://mathoverflow.net/q/279065/47134}), and subsequent simplification of (\ref{Almost}), then yields (Fig.~\ref{fig:MasterFormula})
\begin{equation} \label{Ourformula}
\tilde{\chi_d} (\varepsilon )=
\end{equation}
\begin{displaymath}
\frac{\varepsilon ^d \Gamma (d+1)^3 \,
   _3\tilde{F}_2\left(-\frac{d}{2},\frac{d}{2},d;\frac{d}{2}+1,\frac{3
   d}{2}+1;\varepsilon ^2\right)}{\Gamma \left(\frac{d}{2}+1\right)^2}.
\end{displaymath}
\begin{figure}
\includegraphics{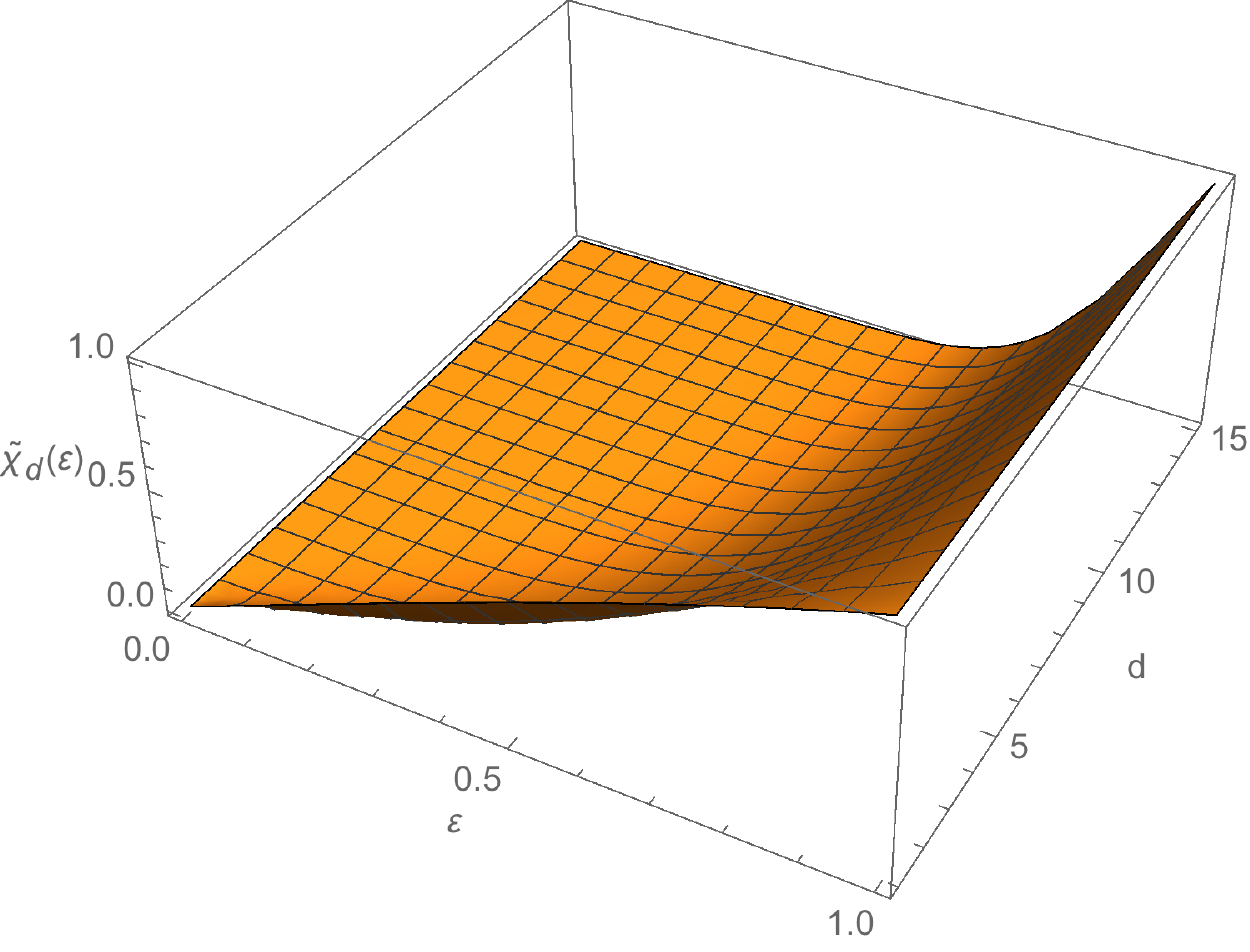}
\caption{\label{fig:MasterFormula}Lovas-Andai master formula (\ref{Ourformula}) for $\tilde{\chi_d} (\varepsilon )$}
\end{figure}

What one would now aspire to accomplish is to replicate (now for general $d$, rather than specifically $d=1$) the course pursued 
by Lovas and Andai in the proof of their Theorem 2, establishing the two-rebit separability probability as $\frac{29}{64}$. 
They employed two changes-of-variables and an integration-by-parts in their argument, recasting the problem as involving an integration
over $s \in [0,\infty]$ and $t \in [0,1]$. 
The generalization for our purposes of their ($d=1$) integrand factor
\begin{equation}
\frac{256 s^4 t^3 \left(1-t^2\right)}{(s+t)^5 (s t+1)^5}
\end{equation}
is
\begin{equation}
f(d,s,t)= \frac{2^{5 d+3} \left(\frac{s t}{(s+t)^2}\right)^{d+1} \left(\frac{s t}{(s
   t+1)^2}\right)^d \left(\frac{1}{s t+1}-\frac{t}{s+t}\right)^d}{(s t+1)^2} =
\end{equation}
\begin{equation}
(-1)^d 2^{5 d+3} s^{3 d+1} t^{2 d+1} \left(t^2-1\right)^d ((s+t) (s t+1))^{-3 d-2}.
\end{equation}
This, in fact, can be integrated over $s \in [0,\infty]$, yielding
\begin{equation} \label{PreFactor}
(-1)^d 2^{5 d+3} t^{-4 d-3} \left(t^2-1\right)^d \Gamma (3 d+2)^2 \, _2\tilde{F}_1\left(3
   d+2,3 d+2;6 d+4;1-\frac{1}{t^2}\right).    
\end{equation}
So, we are faced with the task of integrating the product of this term and $\tilde{\chi_d} (t )$, given by (\ref{Ourformula}), over $t \in [0,1]$.
Division of this ``numerator'' result by the denominator (obtained by substituting $\alpha =\frac{d}{2}$ in (\ref{General})),
\begin{equation} \label{denominatorfactor}
\frac{\pi  3^{-3 d/2} 8^d d \Gamma \left(\frac{3 d}{2}\right) \Gamma (d+1)^2}{\Gamma
   \left(\frac{d}{2}+\frac{5}{6}\right) \Gamma \left(\frac{d}{2}+\frac{7}{6}\right)
   \Gamma \left(\frac{5 d}{2}+2\right)},  
\end{equation}
would then give us the Hilbert-Schmidt separability/PPT-probability for the corresponding $d$-setting.
\subsubsection{MeijerG-based formulas for separability/PPT-probabilities for even $d$} \label{MeijerGsection}
We have, in fact, not yet to this point in time, been able to explicitly perform the indicated integration
$t \in [0,1]$, while allowing $d$ to be free, though the integration can readily be carried out for any specific even value of $d$.
In our quest for such a general formula, 
we have conducted indefinite integrations over $T=\sqrt{t}$, for even values of $d$ and have analyzed 
the results to try to uncover a general rule. We have found that for any specific even value of $d$, the corresponding indefinite 
integration yields a weighted sum of $1 +\frac{3 d}{2}$ MeijerG functions
\begin{equation}
 T^{-\frac{3 d}{2}-1} G_{3,3}^{2,3}\left(\frac{1}{T}|
\begin{array}{c}
 -3 d-1,-3 d-1,-\frac{3 d}{2} \\
 0,0,-\frac{3 d}{2}-1 \\
\end{array}
\right) +   
\end{equation}
\begin{displaymath}
\sum _{i=1}^{\frac{3 d}{2}-1} T^{-\frac{3 d}{2}+i-1} f(d,i)
   G_{3,3}^{2,3}\left(\frac{1}{T}|
\begin{array}{c}
 -3 d-1,-3 d-1,i-\frac{3 d}{2} \\
 0,0,-\frac{3 d}{2}+i-1 \\
\end{array}
\right)+
\end{displaymath}
\begin{displaymath}
f\left(d,\frac{3 d}{2}\right) G_{3,3}^{2,3}\left(\frac{1}{T}|
\begin{array}{c}
 1,-3 d,-3 d \\
 1,1,0 \\
\end{array}
\right)
\end{displaymath}
times a factor of 
\begin{equation} \label{NumFactor}
\frac{\pi ^3 \left(-\frac{1}{2}\right)^{d-1} 3^{-\frac{9 d}{2}-3} \Gamma (3 d+2)}{\Gamma
   \left(\frac{d}{2}+\frac{5}{6}\right)^3 \Gamma \left(\frac{d}{2}+\frac{7}{6}\right)^3
   \Gamma \left(\frac{3 d}{2}+1\right)^4}.    
\end{equation}
The ratio of this ``numerator factor'' (\ref{NumFactor}) to the ``denominator factor'' (\ref{denominatorfactor})
substantially simplifies (correct for both odd and even $d$) to
\begin{equation}
\frac{(-1)^{d+1} \Gamma \left(\frac{5 d}{2}+2\right)}{\Gamma \left(\frac{d}{2}+1\right)^2
   \Gamma \left(\frac{3 d}{2}+1\right)^3 \Gamma (3 d+2)}.    
\end{equation}

Let us note that (with $\mbox{i}$ being the imaginary unit)
\begin{equation}
f\left(d,\frac{3 d}{2}\right) = \frac{3 \mbox{i}^d \Gamma \left(\frac{3 d}{2}\right)^2}{8 \Gamma (d) \Gamma (2 d)},
\end{equation}
and
\begin{equation}
\frac{f\left(d,\frac{3 d}{2}-1\right)}{f(d,1)}=\frac{(27 \mbox{i})^d 4^{-2 d-3} (d+2) (5 (d-1) d+2) \Gamma \left(\frac{d}{2}-\frac{1}{3}\right)
   \Gamma \left(\frac{d}{2}+\frac{1}{3}\right) \Gamma
   \left(\frac{d}{2}+\frac{2}{3}\right) \Gamma
   \left(\frac{d}{2}+\frac{4}{3}\right)}{\sqrt{\pi } (d-1) \Gamma
   \left(d+\frac{1}{2}\right) \Gamma \left(\frac{d+3}{2}\right)^2} .   
\end{equation}

Also,
\begin{equation}
\sum _{i=1}^{\frac{3 d}{2}} f(d,i)=-1.    
\end{equation}

In general, giving us the weights $f(d,i)$ to be employed, we have the linear difference equation (constructed based on multiple applications of the Mathematica FindSequenceFunction command) shown in Fig.~\ref{fig:DifferenceRootExpression},
\begin{figure} 
\includegraphics[scale=0.75]{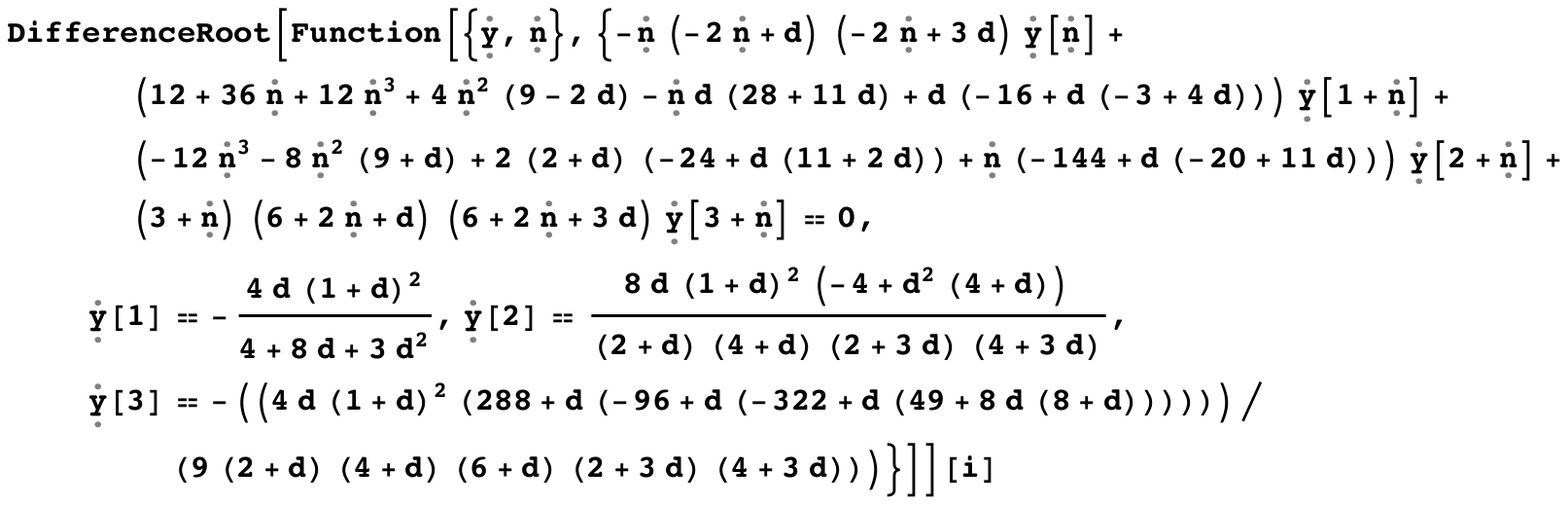}
\caption{\label{fig:DifferenceRootExpression}Linear difference equation for MeijerG summation weights $f(d,i)$}
\end{figure}

Now, one must evaluate the weighted sum of the $1 +\frac{3 d}{2}$ MeijerG functions at the specific end points $T=1$ and $T=0$, taking their
difference to arrive at the desired (definite integration) result. For the two-qubit case $d=2$, we found these two values to be
$148 +\frac{8}{33}$ and 148, respectively, giving us the correct result for the difference of $\frac{8}{33}$. For the two-quaterbit 
$d=4$ instance, we have $-333631 -\frac{297}{323}$ and -333632, giving us the $\frac{26}{323}$ oucome. Now for $d=6$, we have 
$1008871862+\frac{2999}{103385}$ and $10088871862$, consistent with \cite[eq.  (4)]{slaterJModPhys}, and similarly for $d=8$, with
the upper value being $-3543784402375-\frac{4046867}{4091349}$ and the lower value being
-3543784402376, the difference being $\frac{44482}{4091349}$, as expected. So, for odd values of $\frac{d}{2}$, we appear to have pairs of positive limits, and for even values, negative limits. 

In fact, if we replace $\alpha$ in Fig. 3 of \cite{slaterJModPhys} by $\frac{d}{2}$, we arrive at a large hypergeometric-based expression that serves as an alternative--succeeding for both odd and even values of $d$--to these MeijerG-related results.

See also App.~\ref{CFDNotes} for an alternative (arguably, superior/finite in character) approach developed by C. Dunkl to the MeijerG one just outlined.

\subsection{Equivalence argument of C. Koutschan}
It remains now to formally demonstrate the equivalence in predicted Hilbert-Schmidt separabilty/PPT-probabilities yielded by 
the Lovas-Andai-based procedure developed
in this  paper  and the earlier-presented ``concise formula''.

Let us note that with our new formula (\ref{Ourformula}) for $\tilde{\chi_d} (\varepsilon )$, we are able--at least for 
even $d$-to compute the exact rational values of the corresponding separability/PPT-probabilities. With the earlier concise formula ((\ref{Hou1})-(\ref{Hou3})) \cite[eqs. (1)-(3)]{slaterJModPhys}
we are only able--but to apparently arbitrarily high-accuracy--to approximate these values. The case of odd values of $d$ appears to be somewhat more problematical/challenging in obtaining the corresponding exact values. 

To be most specific, to compute the $d$-th separability-PPT probability ($d=1,2,4$ corresponding to $\mathbb{R}, \mathbb{C}, \mathbb{O},\ldots$), we must integrate over $t \in [0,1]$ the product of 
\begin{equation} \label{FinalLovasAndai}
\frac{(-1)^d 3^{\frac{3 d}{2}+1} 8^{2 d+1} t^{-3 (d+1)} \left(t^2-1\right)^d \Gamma
   \left(\frac{d}{2}+\frac{5}{6}\right) \Gamma \left(\frac{d}{2}+\frac{7}{6}\right)
   \Gamma \left(\frac{d+1}{2}\right) \Gamma \left(\frac{3 (d+1)}{2}\right) \Gamma
   \left(\frac{5 d}{2}+2\right) \Gamma (3 d+2)}{\pi ^2 \Gamma \left(\frac{d}{2}+1\right)}
\end{equation}
and
\begin{equation} \label{FinalLovasAndai2}
\, _2\tilde{F}_1\left(3 d+2,3 d+2;6 d+4;1-\frac{1}{t^2}\right) \,
   _3\tilde{F}_2\left(-\frac{d}{2},\frac{d}{2},d;\frac{d}{2}+1,\frac{3 d}{2}+1;t^2\right),    
\end{equation}
where  $\tilde{}$ indicates regularization. (We follow the use of Lovas and Andai in employing
either $t$ or $\varepsilon$, in different settings. When $d$ is even, the $3\tilde{F2}$ function terminates, and is a polynomial in $t$, Dunkl
indicated.)

In fact, C. Koutschan has been able, applying creative telescoping with the use of  his HolonimicFunctions package \cite{Koutschan2010}, to derive a recurrence of order 4 for this integral $I$, 
involving the terms $I(d), I(d+2), I(d+4)$. He has also  
constructed an order 6 recurrence for the large hypergeometric-based expression ($G$), involving $ G(d), G(d+2), G(d+4), G(d+6)$, given in Fig. 3 of \cite{slaterJModPhys} that yields the separability/PPT-probabilities. Equivalently to $G$, we can apparently employ \cite[p. 26]{SlaterFormulas}

He further checked that the order-6 recurrence for G is a left-multiple
(in the [Ore] operator sense) of the order-4 recurrence. Indeed, G also
satisfies the order-4 recurrence (Fig.~\ref{fig:Order4Recurrence}).
\begin{figure}
    \centering
    \includegraphics[scale=0.7]{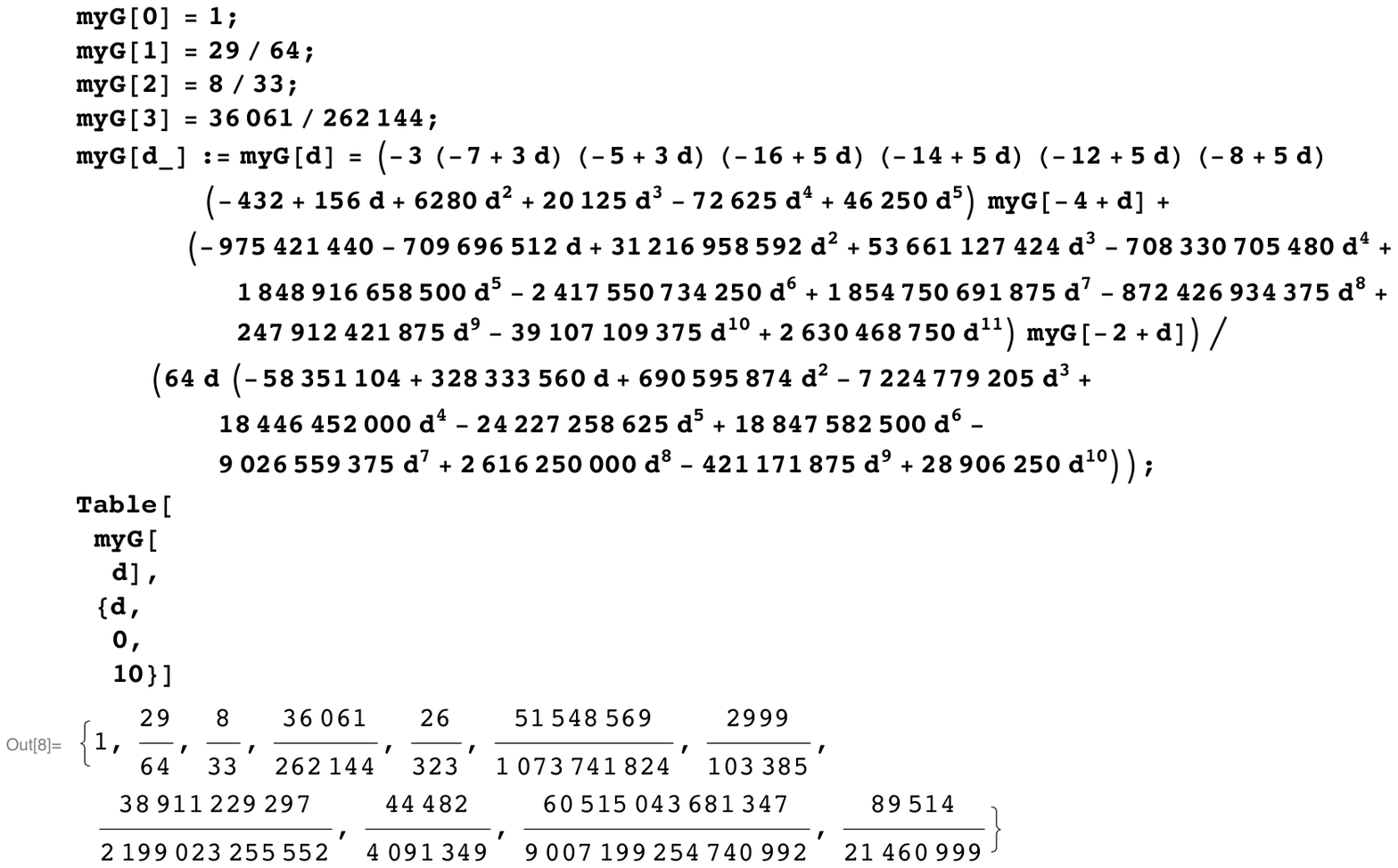}
    \caption{Order-4 recurrence satisfied by: (1) the large $_{7}F_{6}$-hypergeometric-based expression ($G$) [and its ``concise'' reformulation ((\ref{Hou1})-(\ref{Hou3}))]; (2) the Lovas-Andai-based integral $I(d)$ of the product of 
 (\ref{FinalLovasAndai}) and (\ref{FinalLovasAndai2});  and (3) the specialized random induced measure formula (\ref{InducedMeasureCase}).}
    \label{fig:Order4Recurrence}
\end{figure}

He, then, confirmed that certain initial conditions, namely, that $I(0)=G(0), ..., I(3)=G(3)$
are satisfied.
By numerically
evaluating G, and then rationalizing, he found that the first 4 values of
the sequence G(d) were the rational numbers $\left\{\frac{29}{64},\frac{8}{33},\frac{36061}{262144},\frac{26}{323}\right\}$. By applying the
recurrence, he saw that G(d) is rational for each natural number d. As a
consistency check, he compared the values for G(4), ..., G(10) that he got
from: (1) numerical evaluation; and  (2) by applying the recurrence. Indeed, they agreed.

Equivalently to $G$, in this argument, we can apparently employ \cite[p. 26]{SlaterFormulas}
\begin{equation} \label{InducedMeasureCase}
\mathcal{P}_{sep/PPT}(0,d)= 2 Q(0,d)= 1-    
\frac{\sqrt{\pi } 2^{-\frac{9 d}{2}-\frac{5}{2}} \Gamma \left(\frac{3 (d+1)}{2}\right)
   \Gamma \left(\frac{5 d}{4}+\frac{19}{8}\right) \Gamma (2 d+2) \Gamma \left(\frac{5
   d}{2}+2\right)}{\Gamma (d)} \times
\end{equation}
\begin{displaymath}
\, _6\tilde{F}_5\left(1,d+\frac{3}{2},\frac{5 d}{4}+1,\frac{1}{4} (5 d+6),\frac{5
   d}{4}+\frac{19}{8},\frac{3 (d+1)}{2};\frac{d+4}{2},\frac{5
   d}{4}+\frac{11}{8},\frac{1}{4} (5 d+7),\frac{1}{4} (5 d+9),2 (d+1);1\right).
\end{displaymath}
(More generally, $Q(k,d)$ gives that portion, for random {\it induced} measure, parameterized by $k$, of the total separability/PPT-probability for which the determinantal inequality 
$|\rho^{PT}| >|\rho|$ holds. The sum of the six upper parameters of the $_{6}F_{5}$ function here minus the sum of the 
five lower ones is ``d-free'', equalling $-\frac{1}{2}$. This indicates
that terminating the infinite sum associated with the $_{6}F_{5}$ function after $n$ terms, leads to a truncation error
of $O(n^{-\frac{1}{2}})$. In general, $_{p+1}F_p$ converges at t=1 provided the sum of the $p+1$ upper parameters minus the sum 
of the $p$ lower parameters is less than zero  \url{http://dlmf.nist.gov/16.2#i} 16.2(iii).)

\subsection{Monotone Measure Application}
In section 4 of their recent study \cite{lovasandai}, Lovas and Andai extend their analyses from one involving the (non-monotone \cite{ozawa2000entanglement}) Hilbert-Schmidt measure to
one based on the operator monotone function $\sqrt{x}$. They are able to conclude (for the case $d=1$) that the applicable ``separability function" in this case, 
$\tilde{\eta}_d(\varepsilon)$, 
is precisely the same as the Hilbert-Schmidt counterpart  $\tilde{\chi}_d(\varepsilon)$. However, rather than the complementary ``normalization factor'' ((\ref{sepX}), with $\alpha = \frac{d}{2}$),
\begin{equation} \label{MonotoneNumerator}
 {\int\limits_{-1}^1\int\limits_{-1}^x  (1-x^2)^{d}(1-y^2)^{d}(x-y)^{d}  \mbox{d} y \mbox{d} x} ,  
\end{equation}
it is necessary to employ
\begin{equation} \label{monotoneDenominator}
 {\int\limits_{-1}^1\int\limits_{-1}^x  (1-x^2)^{-\frac{d}{4}}(1-y^2)^{-\frac{d}{4}}(x-y)^{d}  \mbox{d} y \mbox{d} x}. 
\end{equation}
(We have not, to this point in time, been able to perform an integration parallel to that yielding (\ref{PreFactor}), expressing the normalization term as a bivariate function of $t$ and 
$d$.) Proceeding, as before, for specific values of $d$, we are able to verify their numerical (two-rebit [$d=1$]) separability probability result $\mathcal{P}_{sep.\sqrt{x}}(\mathbb{R})$ of 0.26223.
(We, further, observe that the normalization term (\ref{monotoneDenominator}), although not amenable
apparently to exact integration, clearly evaluates to $\frac{2 \pi}{3}$.)

Now, quite strikingly, we obtain for the two-qubit ($d=2$) analysis, the ratio of
$\frac{\pi ^2}{2}-\frac{128}{27}$ to $\frac{\pi^2}{2}$, that is,
\begin{equation}
\mathcal{P}_{sep.\sqrt{x}}(\mathbb{C})  = 1-\frac{256}{27 \pi ^2} \approx 0.0393251.
\end{equation}
(We observe that such results--as with the Hilbert-Schmidt 
$\frac{8}{33}$--appear to reach their most simple/elegant in the [standard, 15-dimensional] two-qubit setting.)

For the two-quaterbit ($d=4$) instance, we obtain the ratio of 
$\frac{4 \pi ^2}{3}-\frac{5513}{420}$ to $1.478504859 \times 10^{13}$, yielding (the ``infinitesimal'') result
\begin{equation}
\mathcal{P}_{PPT.\sqrt{x}}(\mathbb{Q})  = 2.2510618339 \times 10^{-15}.
\end{equation}

In light of these three results, it seems of clear interest to pursue parallel analyses for the 
interesting variety of monotone metrics, with the minimal monotone (Bures) \cite{dittmann1999explicit,vsafranek2017discontinuities} one seemingly of 
particular interest. One question of note is whether the original (Hilbert-Schmidt-based) Lovas-Andai function $\tilde{\chi}_d(\varepsilon)$, given by (\ref{Ourformula}), will
continue to be appropriate (as it has in the $\sqrt{x}$ case), and only the complementary normalization factor will change. 
The analyses of X-states, in these regards, might be informative (cf. \cite{slater2015bloch}).

\section{Concluding Remarks}
We have found (sec.~\ref{7Dsection}) that the Lovas-Andai two-rebit separability function $\tilde{\chi}_1(\varepsilon)$ also serves as the Slater separability function in a reduced (from nine to seven-dimensional) setting where the $2 \times 2$ diagonal block matrices $D_1, D_2$ are themselves diagonal. Additionally, we know that the Lovas-Andai two-qubit separability function $\tilde{\chi}_2(\varepsilon)$  serves as the Slater separability functions in a reduced (from fifteen to eleven-dimensional) setting where the $2 \times 2$ diagonal block matrices $D_1, D_2$ are themselves diagonal. It remains a question of some interest as to what the Slater two-rebit and two-qubit separability functions themselves are in the full nine- and fifteen-dimensional settings, in particular, the possibility that the two-qubit separability function might be $\frac{6}{71} (3-\mu^2) \mu^2$ 
(cf. Figs.~\ref{fig:QubitMU}, ~\ref{fig:QubitMUDiff}). (Can the solutions in the Lovas-Andai setting be ``lifted'' to those in the Slater one 
[cf. eqs.(\ref{jacLA})-(\ref{final})]?) Also, we note that Lovas and Andai did not specifically consider $D_1$ and $D_2$ to be diagonal. So, if would be interesting to ascertain whether their same conclusions (such as the formula for  $\tilde{\chi}_1 (\varepsilon )$) could have been reached under such  assumptions.

The counterpart rebit-retrit and qubit-qutrit $6 \times 6$ problems (sec.~\ref{sixbysix}) might also be productively studied when the $3 \times 3$ diagonal blocks are themselves diagonal. The problems under consideration would then be 14 and 23-dimensional in nature, as opposed to 20 and 35-dimensional, with lower-dimensional CAD's still.

In brief summary, let us emphasize that, at this point in time, we have basically {\it four} quite distinct formulas for the generalized two-qubit Hilbert-Schmidt separability/PPT-probabilities. In order of chronological development, we have the large expression containing six $_{7}F_{6}$ hypergeometric functions (all with argument $\frac{27}{64}$) given in \cite[Fig. 3]{slaterJModPhys}, developed on the basis of extensive moment (density approximation) calculations \cite{Provost,MomentBased}. Then, we have the ``concise'' reexpression of this formula obtained by Qing-Hu Hou through the application of Zeilberger's algorithm (creative telescoping) ((\ref{Hou1})-(\ref{Hou3})) \cite[eqs. (1)-(3)]{slaterJModPhys}. Next, we have a formula (\ref{InducedMeasureCase}) containing a single $_{6}F_{5}$ hypergeometric function, that is the specialization of an ``induced measure'' formula \cite[p. 26]{SlaterFormulas} to the Hilbert-Schmidt case ($k=0$). (This ``specialization'' relies upon the observation that, in the Hilbert-Schmidt case, the separability probability is {\it equally} divided between the cases where the nonnegative determinant of a partial transpose is greater or less than the determinant of the density matrix itself.) Finally, we have the formula developed here, the product of (\ref{FinalLovasAndai}) and (\ref{FinalLovasAndai2}), requiring an integration over $t \in [0,1]$, within the Lovas-Andai framework. (Strategies for carrying out the integration are presented in sec.~\ref{MeijerGsection} and App.~\ref{CFDNotes}). Given the constructions by C. Koutschan, using his HolonomicFunctions program \cite{Koutschan2013}, that the first, third and last of these
four formulas satisfy the same order-4 recurrence, we essentially possess a demonstration of the equivalence of all four formulas (as stringent numerics further support). (Of the four formulas,
the only one Mathematica seems able to  exactly evaluate for d = 1, 2, 3, 4,…
is the $_{6}F_{5}$-based one.)

A problem still to be addressed is to extend the set of equivalent formulas studied above, applicable to the ($k=0$) Hilbert-Schmidt case, to the more general random induced measure setting \cite{induced,LatestCollaboration} \cite[sec. XIII, App. E]{SlaterFormulas} (where many $d$-specific formulas for $\mathcal{P}_{sep/PPT}(k,d)$ are given). (Here $k=K-4$, where the measure is induced in the space of $4 \times 4$ mixed states by the natural, 
rotationally invariant measure on the set of all pure states of a $4 \times K$ system.) 
\appendix
\section{{\it Absolute} separability probabilities}
Those separable states that can not be entangled through unitary operations have
been designated as {\it absolutely} separable \cite[p. 392]{ingemarkarol}.

In \cite{slater2009eigenvalues}, we reported exact (but now decidedly not rational-valued, and much
smaller-valued) 
formulas for the Hilbert-Schmidt absolute separability probabilities
for the two-rebit, two-qubit and two-quaterbit states. For the convenience and interest of the reader,
we present them here, while simplifying the forms of the last two.

The two-rebit absolute separability probability is expressible as 
\cite[eq. (32)]{slater2009eigenvalues} 
\begin{equation}
 \frac{6928-2205 \pi }{16 \sqrt{2}}   \approx 0.0348338,
\end{equation}
the two-qubit as
\cite[eq. (34)]{slater2009eigenvalues}
\begin{equation}
1-\frac{3217542976-5120883075 \pi +16386825840 \tan ^{-1}\left(\sqrt{2}\right)}{32768
   \sqrt{2}}-\frac{29901918259}{497664}    \approx 0.00365826
\end{equation}
and the two-quaterbit as
\cite[eq. (36)]{slater2009eigenvalues}
\begin{equation}
 \frac{13}{3043362286338048} \times   
\end{equation}
\begin{displaymath}
(806338156306739134839776-658857590468226345222144 \sqrt{2}+
\end{displaymath}
\begin{displaymath}
1048604423167357891775325
   \sqrt{2} \pi -3355534154135545253681040 \sqrt{2} \tan ^{-1}\left(\sqrt{2}\right)) \approx 0.0000401326.
\end{displaymath}
Let us note that the integer components of the denominators appearing above are either simply powers of 2, or involve high powers of 2.
Also, $ \tan ^{-1}\left(\sqrt{2}\right)) \approx 0.955317$--the angle between the space diagonal of a cube and any of its three connecting angles--has been termed the ``magic angle'' (see the eponymous wikipedia article). We have also been able to obtain absolute separability probabilities in the two-rebit case, again featuring this particular angle prominently, when the Hilbert-Schmidt ($k=0$) measure is replaced by random induced measures \cite{induced} for $k=1,2,3$.

It appears to be a substantial challenge--using eq. (4) of  \cite{roland} to find the $6 \times 6$ counterparts to these
three formulas for $4 \times 4$ systems.
\section{Rebit-retrit and qubit-qutrit analyses} \label{sixbysix}

Let us now attempt to extend the two-rebit and two-qubit line of analysis above to rebit-retrit and qubit-qutrit settings--now, of course, passing from consideration of $4 \times 4$ density matrices to $6 \times 6$ ones. Lovas and Andai, in their quite recent study, had not yet addressed such issues. In \cite{slater833}, candidate (Slater-type) separability functions had been proposed. Two dependent variables (cf. 
the use of  $\mu=\sqrt{\frac{\rho_ {11} \rho_ {44}}{\rho_ {22} \rho_ {33}}}$ in the lower-dimensional setting above) had been employed \cite[eq. (44)]{slater833}. Let us now refer to these two variables as $\tau_1=\sqrt{\frac{\rho_ {11} \rho_{55}}{\rho_ {22} \rho_ {44}}}$  and $\tau_2 =\sqrt{\frac{\rho_ {22} \rho_{66}}{\rho_ {33} d _{55}}}$. But, interestingly, it was argued that only a single dependent variable $\tau= \tau_1 \tau_2 = 
\sqrt{\frac{\rho_ {11} \rho_{66}}{\rho_ {33} \rho_ {44}}}$ sufficed for modeling the corresponding separability functions.
The separability function in the rebit-retrit case was proposed to be simply proportional to $\tau$ \cite[eq.(98)]{slater833}.

In our effort to extend the Lovas-Andai analyses \cite{lovasandai} to this setting, we now took $D_1$ and $D_2$ to equal the upper and lower
diagonal $3 \times 3$ blocks of the $6 \times 6$ density matrix in question. Then, we computed the three singular values ($s_1 \geq s_2 \geq s_3$) of $D_2^{1/2} D_1^{-1/2}$, and took the ratio variables $\varepsilon_1 = \frac{s_2}{s_1}$ and $\varepsilon_2 = \frac{s_3}{s_2}$ as the dependent ones in question. (An issue of possible concern is that, unlike the $4 \times 4$ case \cite{augusiak2008universal}, positivity of the determinant of the partial transpose of a $6 \times 6$ density matrix is only a necessary, but not sufficient condition for separability.) Also, in the case of diagonal $D$, again the two variables in the $\mu$ framework are equal to those in the $\varepsilon$ setting, or to their reciprocals.

Then, we generated 3,436 million rebit-retrit and 2,379 million qubit-qutrit density matrices, randomly with respect to Hilbert-Schmidt measure. 
(These sizes are much larger than those employed in 2007--for similar purposes--in \cite{slater833}.) We appraised the separability of the density matrices $D$ by testing whether the partial transpose, using the four $3 \times 3$ blocks, had all its six eigenvalues positive. The separability probability estimates were $0.13180011 \pm 0.0000113109$ and $0.02785302 \pm 6.6124281 \cdot 10^{-6}$, respectively. (We can reject the qubit-qutrit conjecture of $\frac{32}{1199} \approx  0.0266889$ advanced in \cite[sec. 10.2]{slater833}. A possible alternative candidate is $\frac{72}{2585} =\frac{5 \cdot 11 \cdot 47}{2^3 \cdot 3^2} \approx 0.027853$, while in the rebit-retrit case, we have $\frac{298}{2261} =\frac{7 \cdot 17 \cdot 19}{2 \cdot 149} \approx 0.1318001$.) Further, our estimates of the probabilities that $D$ had two [the most possible \cite{PhysRevA.87.064302}] negative eigenvalues, and hence a positive determinant, although being entangled, were $0.0334197 \pm 0.0000409506$ in the rebit-retrit case, and $0.0103211 \pm 0.000031321$ in the qubit-qutrit instance.) 

In the two-variable settings, we partition the square $[0,1]^2$ of possible separability probability results into an 
$80 \times 80$ grid, and in the one-variable setting, use a partitioning (as in the two-rebit and two-qubit analyses above) into 200 subintervals of [0,1]. In Fig.~\ref{fig:TritRatioTau} we show the ratio of the square of the rebit-retrit separability probability to the
qubit-qutrit separability probability as a function of $\tau$, while in Fig.~\ref{fig:TritRatioTau1Tau2}, we show a two-dimensional version. Fig.~\ref{fig:TritRatioEpsilon} is the analog of this last plot using the singular-value ratios  $\varepsilon_1$ and $\varepsilon_2$.
\begin{figure}
\includegraphics{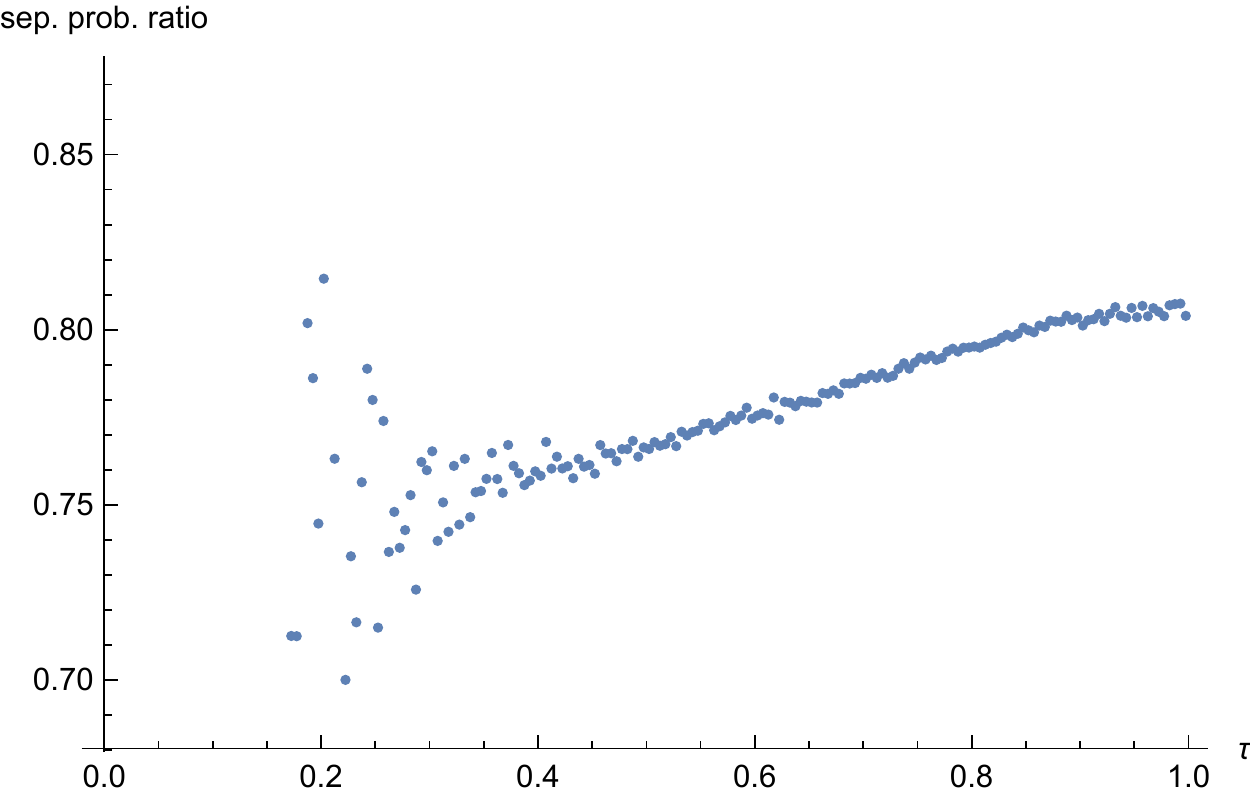}
\caption{\label{fig:TritRatioTau}The ratio of the square of the rebit-retrit separability probability to the
qubit-qutrit separability probability as a function of $\tau= \tau_1 \tau_2 = 
\sqrt{\frac{\rho_ {11} \rho_{66}}{\rho_ {33} \rho_ {44}}}$}
\end{figure}
As in Figs.~\ref{fig:SlaterRatio}, \ref{fig:LovasAndaiRatio} and \ref{fig:TwoDimensionalRatio}, we observe a gradual increase in these Dyson-index-oriented analyses. In Figs.~\ref{fig:DiagonalRebitRetrit} and \ref{fig:DiagonalQubitQutrit}, we show the highly linear (``diagonal'') rebit-retrit and qubit-qutrit separability
probabilities, holding  $\tau_1=\tau_2$.
\begin{figure}
\includegraphics{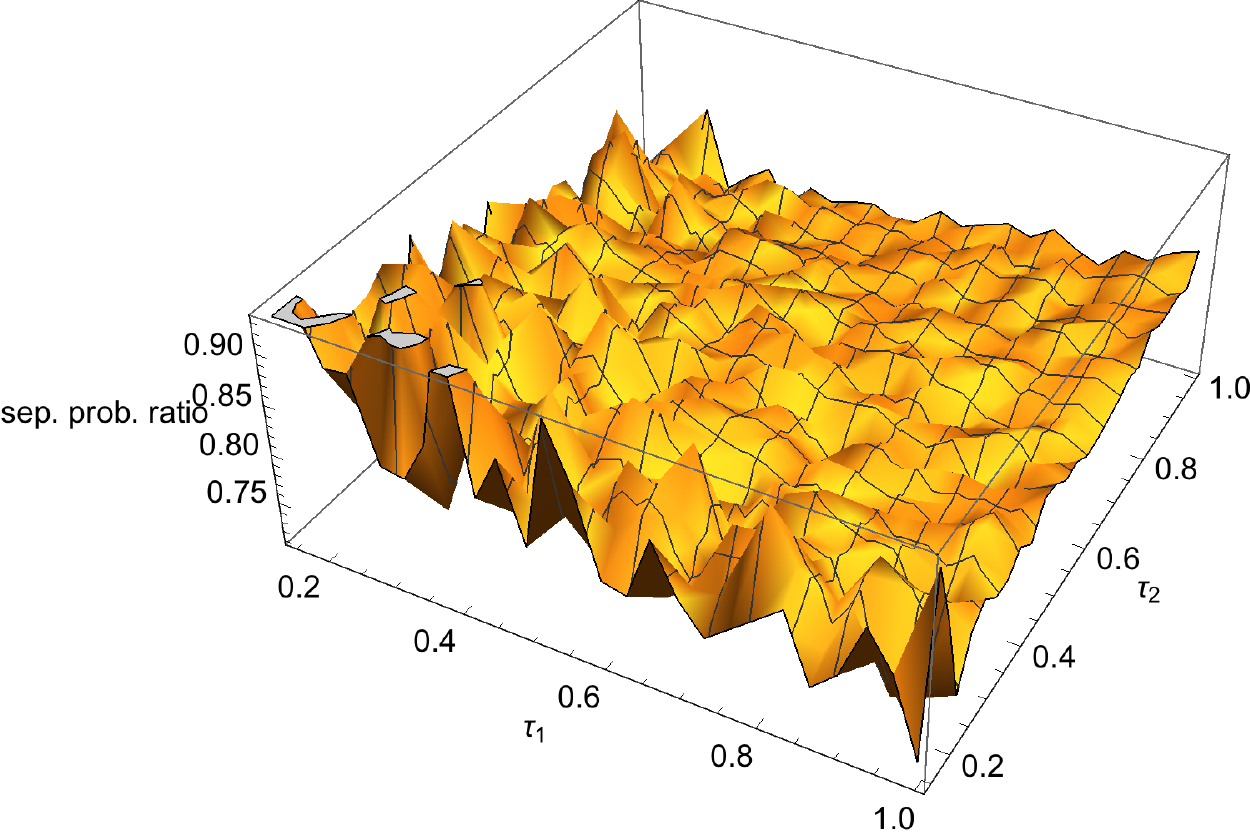}
\caption{\label{fig:TritRatioTau1Tau2}The ratio of the square of the rebit-retrit separability probability to the
qubit-qutrit separability probability as a function of $\tau_1=\sqrt{\frac{\rho_ {11} \rho_{55}}{\rho_ {22} \rho_ {44}}}$  and $\tau_2 =\sqrt{\frac{\rho_ {22} d _{66}}{\rho_ {33} d _{55}}}$}
\end{figure}
\begin{figure}
\includegraphics{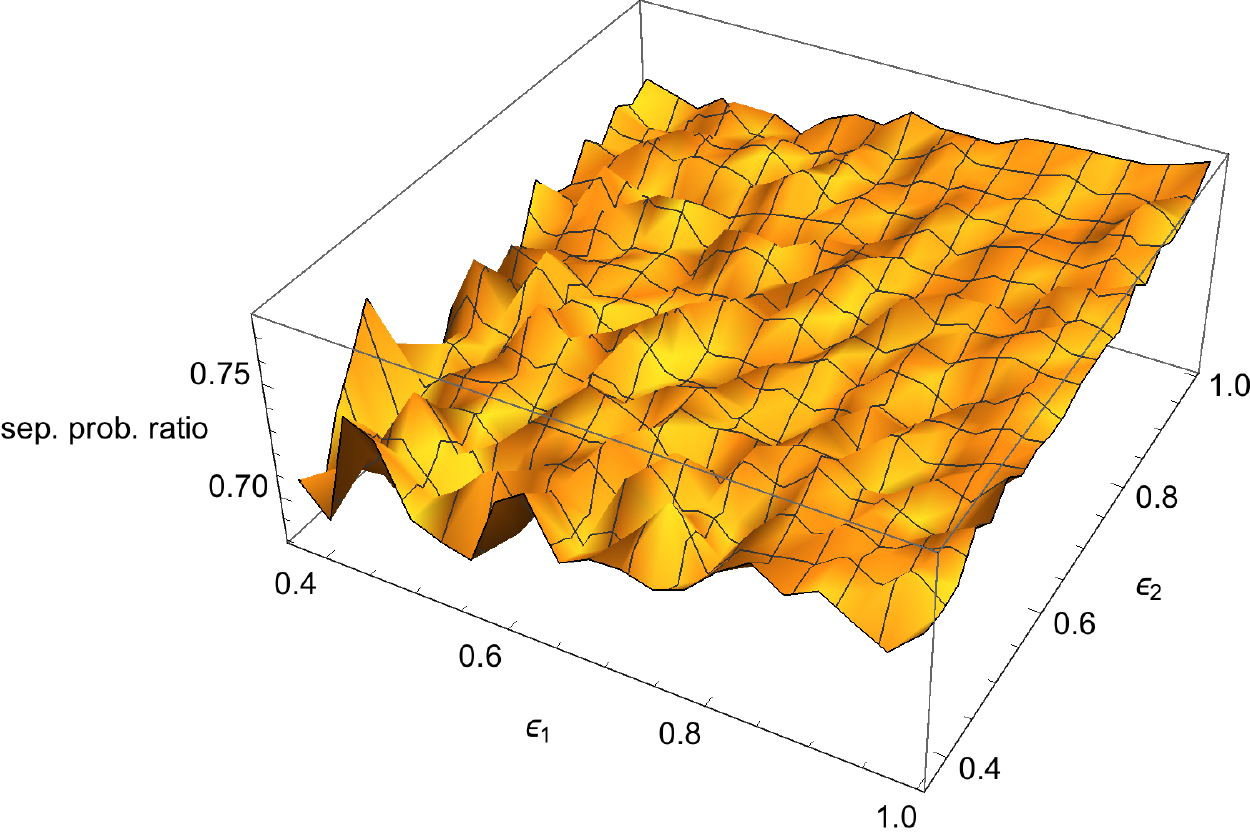}
\caption{\label{fig:TritRatioEpsilon}The ratio of the square of the rebit-retrit separability probability to the
qubit-qutrit separability probability as a function of the singular value ratios $\varepsilon_1$ and $\varepsilon_2$}
\end{figure}
\begin{figure}
    \centering
    \includegraphics{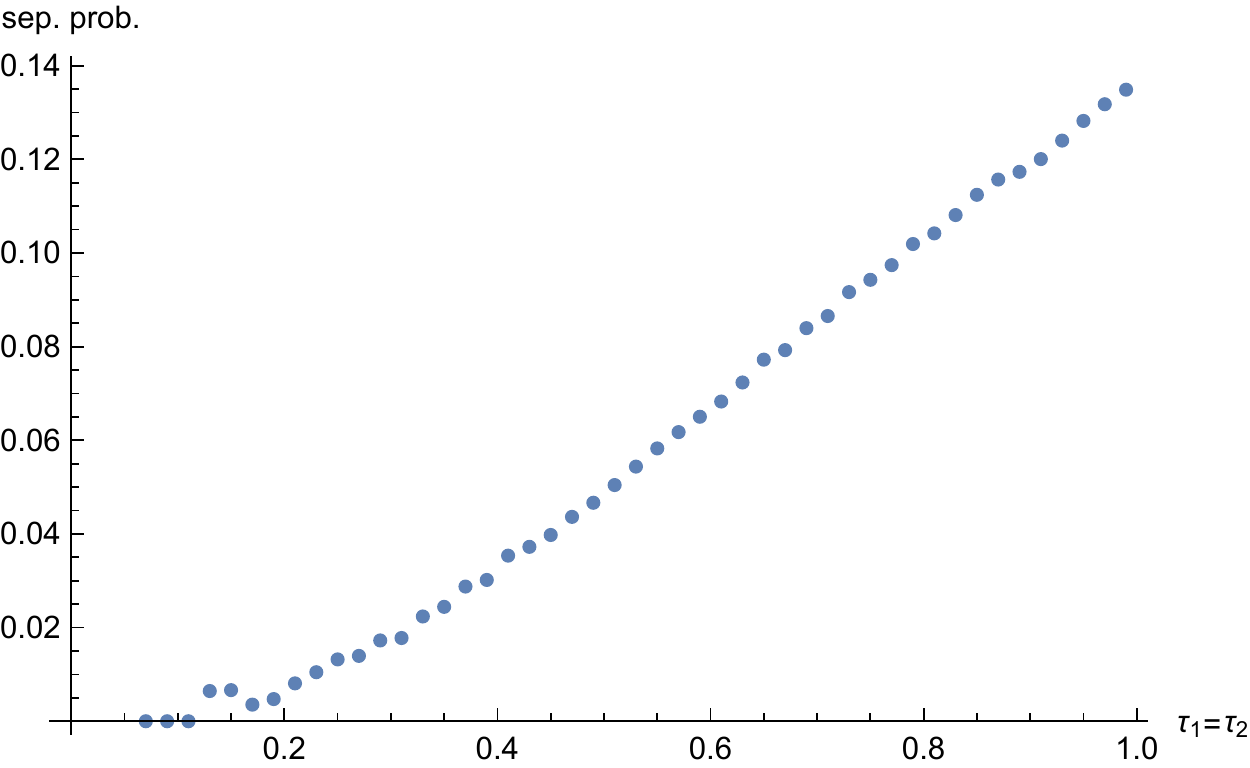}
    \caption{Rebit-retrit separability probabilities for $\tau_1=\tau_2$}
    \label{fig:DiagonalRebitRetrit}
\end{figure}
\begin{figure}
    \centering
    \includegraphics{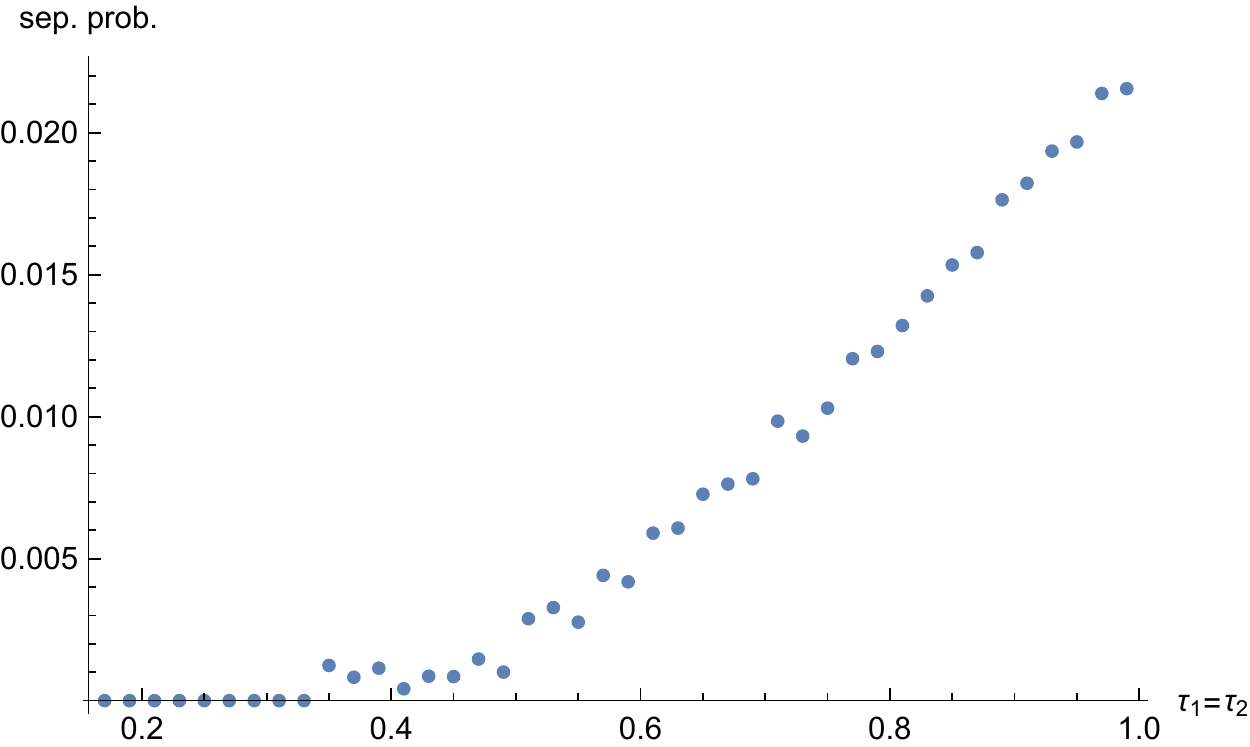}
    \caption{Qubit-qutrit separability probabilities for $\tau_1=\tau_2$}
    \label{fig:DiagonalQubitQutrit}
\end{figure}

Let us note that Mendon{\c{c}}a, and Marchiolli, and Hedemann have recently
shown \cite[App. A]{mendoncca2017maximally} that for qubit-qutrit $X$-states,
the partial transposes can--in contrast to more general such $6 \times 6$ systems--have no more than one negative eigenvalue.
Therefore, positivity of the determinant of the partial transpose 
is both necessary and sufficient for separability, in this case. Nevertheless, Dunkl
has been able to conclude that the Hilbert-Schmidt separability probabilities
reported in \cite{dunkl2015separability} for {\it two-qubit} $X$-states, continue to hold
in these higher-dimensional qubit-qutrit $X$-state systems.
\section{Comparison of Ginibre and Cholesky methods for quaternion
positive-definite matrices--by C. F. Dunkl} \label{CholGini}

The calculations depend on integrating monomials over the unit sphere in
$\mathbb{R}^{N}$. We use the Pochhammer symbol $\left(  a\right)  _{n}%
:=\prod_{i=1}^{n}\left(  a+i-1\right)  $. If $a\neq0,-1,-2,\ldots$ then
$\Gamma\left(  a+n\right)  /\Gamma\left(  a\right)  =\left(  a\right)  _{n}$.

\begin{lemma}
\label{monint}Let $S_{N-1}$ be the unit sphere in $\mathbb{R}^{N}$ with the
inherited rotation-invariant measure $d\omega$ and let $n_{1},n_{2}%
,\ldots,n_{N}\in\mathbb{N}_{0}$ ($\left\{  0,1,2,3\ldots\right\}  $) then%
\[
\int_{S_{N-1}}\prod_{i=1}^{N}\left\vert x_{i}\right\vert ^{n_{i}}%
d\omega\left(  x\right)  =2\prod_{i=1}^{N}\Gamma\left(  \frac{n_{i}+1}%
{2}\right)  /\Gamma\left(  \frac{1}{2}\left(  \sum_{i=1}^{N}n_{i}+N\right)
\right)  .
\]

\end{lemma}

\begin{proof}
Let $n=\sum_{i=1}^{N}n_{i}$ and $f\left(  x\right)  =\prod_{i=1}^{N}\left\vert
x_{i}\right\vert ^{n_{i}}$. In spherical polar coordinates%
\[
\int_{\mathbb{R}^{N}}f\left(  x\right)  \exp\left(  -\frac{\left\vert
x\right\vert ^{2}}{2}\right)  dx=\int_{0}^{\infty}r^{n}\exp\left(
-\frac{r^{2}}{2}\right)  r^{N-1}dr\int_{S_{N-1}}f\left(  x\right)
d\omega\left(  x\right)  .
\]
The left hand side equals (by the substitution $x_{i}^{2}=2t$ )%
\[
\prod_{i=1}^{N}\int_{-\infty}^{\infty}\left\vert x_{i}\right\vert ^{n_{i}%
}e^{-\left\vert x_{i}\right\vert ^{2}/2}dx_{i}=\prod_{i=1}^{N}2^{\left(
n_{i}+1\right)  /2}\Gamma\left(  \frac{n_{i}+1}{2}\right)  ,
\]
and%
\[
\int_{0}^{\infty}r^{n}\exp\left(  -\frac{r^{2}}{2}\right)  r^{N-1}%
dr=2^{\left(  N+n\right)  /2-1}\Gamma\left(  \frac{n+N}{2}\right)  .
\]
Divide the left side by this to obtain $\int_{S}f\left(  x\right)
d\omega\left(  x\right)  $.
\end{proof}

Denote the right hand side by $I\left[  n_{1},n_{2},\ldots,n_{N}\right]  $. We
use $0\$n$ to denote $0$ listed $n$ times. Thus the surface measure of
$S_{N-1}$ is $I[0\$N]$. We need another lemma for integrating powers of sums
of squares.

\begin{lemma}
\label{sqint}Suppose $1\leq A\leq B$ and $n=1,2,3,\ldots$then the normalized
integral%
\[
\int_{S_{B-1}}\left(  \sum_{i=1}^{A}y_{i}^{2}\right)  ^{n}d\omega\left(
y\right)  \left\{  \int_{S_{B-1}}d\omega\left(  y\right)  \right\}
^{-1}=\dfrac{\left(  A/2\right)  _{n}}{\left(  B/2\right)  _{n}}.
\]

\end{lemma}

\begin{proof}
The argument is similar to the proof of Lemma \ref{monint}. An alternative
approach would rely on Dirichlet integrals
\end{proof}

The Cholesky method begins with a random point on $S_{27}\subset
\mathbb{R}^{28}$ to form an upper triangular matrix $A$ such that
$A_{11},A_{22},A_{33},A_{44}\geq0$ and $A_{ij}\in\mathbb{H}$ for $1\leq
i<j\leq4$. Then $Q:=A^{\ast}A$ is positive-definite and $trQ=1$. In particular
$Q_{11}=A_{11}^{2}$ and $\det Q=\left(  A_{11}A_{22}A_{33}A_{44}\right)  ^{2}%
$. The Jacobian is $A_{11}^{13}A_{22}^{9}A_{33}^{5}A_{44}$. For the measure
$\left(  \det Q\right)  ^{k}$ the $n^{th}$ moment of $Q_{11}$ (that is
$\mathcal{E}\left(  Q_{11}^{n}\right)  $) is given by%
\begin{align*}
\mu_{n}  &  =\frac{I\left[  2n+2k+13,2k+9,2k+5,2k+1,0\$24\right]
}{I[2k+13,2k+9,2k+5,2k+1,0\$24]}\\
&  =\frac{\Gamma\left(  k+7+n\right)  \Gamma\left(  4k+28\right)  }%
{\Gamma\left(  k+7\right)  \Gamma\left(  4k+28+n\right)  }=\frac{\left(
k+7\right)  _{n}}{\left(  4k+28\right)  _{n}},\\
\mu_{1}  &  =\frac{1}{4},~\mu_{2}=\frac{k+8}{4\left(  4k+29\right)  },\\
\mu_{3}  &  =\frac{\left(  k+8\right)  \left(  k+9\right)  }{8\left(
4k+29\right)  \left(  2k+15\right)  }.
\end{align*}

The Ginibre method for $M\times4$ begins with a random point on $S_{16M-1}%
\subset\mathbb{R}^{16M}$ to form an $M\times4$ matrix $H$ with $H_{ij}%
\in\mathbb{H}$ for $1\leq i\leq M,1\leq j\leq4$ and $\sum_{i=1}^{M}\sum
_{j=1}^{4}\left\vert H_{ij}\right\vert ^{2}=1$. (For a quaternion
$q=x_{1}+x_{2}\boldsymbol{i}+x_{3}\boldsymbol{j}+x_{4}\boldsymbol{k}$ define
$\overline{q}=x_{1}-x_{2}\boldsymbol{i}-x_{3}\boldsymbol{j}-x_{4}%
\boldsymbol{k}$ then $\left\vert q\right\vert ^{2}=\overline{q}q=\sum
_{i=1}^{4}x_{i}^{2}$.) Then $Q:=H^{\ast}H$ is positive-definite and $trQ=1$.
In particular $Q_{11}=\sum_{i=1}^{M}\left(  A^{\ast}\right)  _{1i}A_{i1}%
=\sum_{i=1}^{M}\overline{A_{i1}}A_{i1}=\sum_{i=1}^{M}\left\vert A_{i1}%
\right\vert ^{2}$. In real terms each $\left\vert A_{i1}\right\vert ^{2}$ is a
sum of four squared real variables so rewrite $Q_{11}=\sum_{j=1}^{4M}x_{j}%
^{2}$ where $\left(  x_{j}\right)  _{j=1}^{16M}$ is a random point on
$S_{16M-1}$. Apply Lemma \ref{sqint} with $A=4M$ and $B=16M$ to obtain%
\[
\nu_{n}=\mathcal{E}\left(  Q_{11}^{n}\right)  =\mathcal{E}\left(  \sum
_{j=1}^{4M}x_{j}^{2}\right)  ^{n}=\frac{\left(  2M\right)  _{n}}{\left(
8M\right)  _{n}}%
\]
for $n=1,2,3,\ldots$. In particular $\nu_{1}=\frac{1}{4}$ and $\nu_{2}%
=\dfrac{2M+1}{4\left(  8M+1\right)  }$.

Thus $\nu_{n}=\dfrac{\left(  2M\right)  _{n}}{\left(  8M\right)  _{n}}$ and
$\mu_{n}=\dfrac{\left(  k+7\right)  _{n}}{\left(  4k+28\right)  _{n}}$ are
equal for all $n$ exactly when $k+7=2M$. In itself this is not a proof that
the Ginibre method produces $\left(  \det Q\right)  ^{2M-7}$ times the HS
measure. This statement is a consequence of equation (4.6) in \cite{induced}.

In particular the Ginibre method does not lead to the HS measure for any $M$
since $k$ is necessarily odd.

The above calculations can be adapted to other values of the parameter
$\alpha$ (with $\alpha=\frac{1}{2}$ for $\mathbb{R}$, $\alpha=1$ for
$\mathbb{C}$, and $\alpha=2$ for $\mathbb{H}$). The Cholesky method starts
with the sphere in $\mathbb{R}^{4+12\alpha}$ ($4$ on diagonal, $12\alpha$ off
diagonal) and the Jacobian is $A_{11}^{1+6\alpha}A_{22}^{1+4\alpha}%
A_{33}^{1+2\alpha}A_{44}$. The \textit{n}th moment of $Q_{11}$ is%
\[
\frac{I\left[  2n+1+6\alpha+2k,1+4\alpha+2k,1+2\alpha+2k,1+2k,0\$12\alpha
\right]  }{I\left[  1+6\alpha+2k,1+4\alpha+2k,1+2\alpha+2k,1+2k,0\$12\alpha
\right]  }=\frac{\left(  k+3\alpha+1\right)  _{n}}{\left(  4k+12\alpha
+4\right)  _{n}}.
\]
Similarly to above the Ginibre matrix size $M\times4$ comes from a random
point on $S_{8\alpha M-1}\subset\mathbb{R}^{8\alpha M}$ and $Q_{11}$ is the
sum of $2\alpha M$ squares $x_{i}^{2}$ and the \textit{n}th moment of $Q_{11}$
is
\[
\frac{\left(  \alpha M\right)  _{n}}{\left(  4\alpha M\right)  _{n}},
\]
which agrees with the $\left(  \det Q\right)  ^{k}\times HS$ when $k=\alpha
M-3\alpha-1=\left(  M-3\right)  \alpha-1$. To apply the formula:%
\begin{align*}
\alpha &  =\frac{1}{2},k=\frac{1}{2}(M-5),\\
\alpha &  =1,k=M-4,\\
\alpha &  =2,k=2M-7.
\end{align*}
Thus the Ginibre method does produce HS random density matrices with $M=5$ for
$\mathbb{R}$ and $M=4$ for $\mathbb{C}$.

\section{A Further Formula for the Lovas-Andai Integral--by C. F.  Dunkl} \label{CFDNotes}
We deal here with the evaluation of the integral for even~$d$.%
\begin{equation}
\int_{0}^{1}t^{-3\left(  1+d\right)  }~_{2}F_{1}\left(
\genfrac{}{}{0pt}{}{2+3d,2+3d}{4+6d}%
;1-\frac{1}{t^{2}}\right)  \left(  1-t^{2}\right)  ^{d}~_{3}F_{2}\left(
\genfrac{}{}{0pt}{}{-\frac{d}{2},\frac{d}{2},d}{1+\frac{d}{2},1+\frac{3d}{2}}%
;t^{2}\right)  \mathrm{d}t\label{probint}%
\end{equation}

The first step is to rewrite the $_{2}F_{1}$ series as a power series in
$t^{2}$. By use of the identity%
\[
_{2}F_{1}\left(
\genfrac{}{}{0pt}{}{a,b}{c}%
;x\right)  =\left(  1-x\right)  ^{-a}~_{2}F_{1}\left(
\genfrac{}{}{0pt}{}{a,c-b}{c}%
;\frac{x}{x-1}\right)
\]
we obtain%
\[
_{2}F_{1}\left(
\genfrac{}{}{0pt}{}{2+3d,2+3d}{4+6d}%
;1-\frac{1}{t^{2}}\right)  =t^{2\left(  2+3d\right)  }~_{2}F_{1}\left(
\genfrac{}{}{0pt}{}{2+3d,2+3d}{4+6d}%
;1-t^{2}\right)
\]
Thus the desired integral equals
\[
\int_{0}^{1}t^{1+3d}~_{2}F_{1}\left(
\genfrac{}{}{0pt}{}{2+3d,2+3d}{4+6d}%
;1-t^{2}\right)  \left(  1-t^{2}\right)  ^{d}~_{3}F_{2}\left(
\genfrac{}{}{0pt}{}{-\frac{d}{2},\frac{d}{2},d}{1+\frac{d}{2},1+\frac{3d}{2}}%
;t^{2}\right)  \mathrm{d}t
\]

By expanding the series we can integrate term-by-term. The typical term is%
\begin{align*}
\int_{0}^{1}t^{1+3d+2j}\left(  1-t^{2}\right)  ^{d+n}\mathrm{d}t &  =\frac
{1}{2}\int_{0}^{1}s^{j+3d/2}\left(  1-s\right)  ^{d+n}\mathrm{d}s\\
&  =\frac{1}{2}B\left(  j+\frac{3d}{2}+1,d+n+1\right)  =\frac{1}{2}%
\frac{\Gamma\left(  j+\frac{3d}{2}+1\right)  \Gamma\left(  d+n+1\right)
}{\Gamma\left(  j+\frac{5d}{2}+2+n\right)  }\\
&  =\frac{1}{2}\frac{\Gamma\left(  \frac{3d}{2}+1\right)  \Gamma\left(
d+1\right)  }{\Gamma\left(  \frac{5d}{2}+2\right)  }\frac{\left(  d+1\right)
_{n}}{\left(  j+\frac{5d}{2}+2\right)  _{n}}\frac{\left(  \frac{3d}%
{2}+1\right)  _{j}}{\left(  \frac{5d}{2}+2\right)  _{j}},
\end{align*}
by use of the change-of-variable $s=t^{2}$ and the beta function. The result
for the integral (\ref{probint}) is (with $d$ even)%
\begin{equation}
\frac{1}{2}\frac{\Gamma\left(  \frac{3d}{2}+1\right)  \Gamma\left(
d+1\right)  }{\Gamma\left(  \frac{5d}{2}+2\right)  }\sum_{j=0}^{d/2}%
\frac{\left(  -\frac{d}{2}\right)  _{j}\left(  \frac{d}{2}\right)  _{j}\left(
d\right)  _{j}}{\left(  1+\frac{d}{2}\right)  _{j}\left(  2+\frac{5d}%
{2}\right)  _{j}j!}~_{3}F_{2}\left(
\genfrac{}{}{0pt}{}{2+3d,2+3d,d+1}{4+6d,j+\frac{5d}{2}+2}%
;1\right)  .\label{big1}%
\end{equation}
By use of the Gauss sum and contiguous hypergeometric series we can produce
finite expressions for the integral. We state the Gauss sum (with $\dot
{c}>a+b$) and define utility functions.%
\begin{align*}
S\left(  a,b,c\right)   &  :=\frac{\Gamma\left(  c-a-b\right)  \Gamma\left(
c\right)  }{\Gamma\left(  c-a\right)  \Gamma\left(  c-b\right)  },\\
_{2}F_{1}\left(
\genfrac{}{}{0pt}{}{a,b}{c}%
;1\right)   &  =S\left(  a,b,c\right)  .
\end{align*}
For $k=0,1,2\ldots$ and $k<\min\left(  c-a-b,g\right)  $ define
\[
S_{32}\left(  a,b,c;g,k\right)  :=~_{3}F_{2}\left(
\genfrac{}{}{0pt}{}{a,b,g}{c,g-k}%
;1\right)  .
\]

\begin{proposition}
For $k=0,1,2\ldots$ and $k<\min\left(  c-a-b,g\right)  $%
\[
S_{32}\left(  a,b,c;g,k\right)  =\frac{\Gamma\left(  c-a-b\right)
\Gamma\left(  c\right)  }{\Gamma\left(  c-a\right)  \Gamma\left(  c-b\right)
}\sum_{j=0}^{k}\frac{\left(  -k\right)  _{j}\left(  a\right)  _{j}\left(
b\right)  _{j}}{j!\left(  1+a+b-c\right)  _{j}\left(  g-k\right)  _{j}}.
\]

\end{proposition}

\begin{proof}
For any $n$ there is the formula (easy to verify, by finite differences, or
the Chu-Vandermonde sum)%
\[
\frac{\left(  g+n-k\right)  _{k}}{\left(  g-k\right)  _{k}}=\sum_{j=0}%
^{k}\binom{k}{j}\frac{n\left(  n-1\right)  \left(  n-2\right)  \cdots\left(
n-j+1\right)  }{\left(  g-k\right)  _{j}}.
\]
Then%
\[
~_{3}F_{2}\left(
\genfrac{}{}{0pt}{}{a,b,g}{c,g-k}%
;1\right)  =\sum_{n=0}^{\infty}\frac{\left(  a\right)  _{n}\left(  b\right)
_{n}}{\left(  c\right)  _{n}n!}\frac{\left(  g\right)  _{n}}{\left(
g-k\right)  _{n}}%
\]
and
\[
\frac{\left(  g\right)  _{n}}{\left(  g-k\right)  _{n}}=\frac{\left(
g\right)  _{n-k}\left(  g+n-k\right)  _{k}}{\left(  g-k\right)  _{k}\left(
g\right)  _{n-k}}=\frac{\left(  g+n-k\right)  _{k}}{\left(  g-k\right)  _{k}}.
\]
Also%
\begin{gather*}
\sum_{n=0}^{\infty}\frac{\left(  a\right)  _{n}\left(  b\right)  _{n}}{\left(
c\right)  _{n}n!}n\left(  n-1\right)  \left(  n-2\right)  \cdots\left(
n-j+1\right)  =\sum_{n=j}^{\infty}\frac{\left(  a\right)  _{n}\left(
b\right)  _{n}}{\left(  c\right)  _{n}\left(  n-j\right)  !}\\
=\frac{\left(  a\right)  _{j}\left(  b\right)  _{j}}{\left(  c\right)  _{j}%
}\sum_{m=0}^{\infty}\frac{\left(  a+j\right)  _{m}\left(  b+j\right)  _{m}%
}{\left(  c+j\right)  _{m}m!}=\frac{\left(  a\right)  _{j}\left(  b\right)
_{j}}{\left(  c\right)  _{j}}S\left(  a+j,b+j,c+j\right)  .
\end{gather*}
making the change of index $m=n-j$ (observe that $n\cdots\left(  n-j+1\right)
=0$ for $0\leq n\leq j-1$). The $S$-term equals%
\[
\frac{\Gamma\left(  c-a-b-j\right)  \Gamma\left(  c+j\right)  }{\Gamma\left(
c-a\right)  \Gamma\left(  c-b\right)  }=\frac{\left(  c\right)  _{j}}{\left(
c-a-b-j\right)  _{j}}\frac{\Gamma\left(  c-a-b\right)  \Gamma\left(  c\right)
}{\Gamma\left(  c-a\right)  \Gamma\left(  c-b\right)  }%
\]
by use of the relation $\Gamma\left(  t\right)  \left(  t\right)  _{j}%
=\Gamma\left(  t+j\right)  $; also by reversal $\left(  c-a-b-j\right)
_{j}=\left(  -1\right)  ^{j}\left(  1+a+b-c\right)  _{j}$. The binomial
coefficient $\binom{k}{j}=\left(  -1\right)  ^{j}\frac{\left(  -k\right)
_{j}}{j!}$. Combine the ingredients and this proves the formula.
\end{proof}

Consider the typical term in (\ref{big1})%
\[
_{3}F_{2}\left(
\genfrac{}{}{0pt}{}{2+3d,2+3d,d+1}{4+6d,j+\frac{5d}{2}+2}%
;1\right)  .
\]
Set $a=2+3d,b=1+d,c=4+6d,g=2+3d$, $k=\frac{d}{2}-j$. Thus the desired integral
(\ref{probint})  equals%
\begin{align*}
&  \frac{1}{2}\frac{\Gamma\left(  \frac{3d}{2}+1\right)  \Gamma\left(
d+1\right)  }{\Gamma\left(  \frac{5d}{2}+2\right)  }\sum_{j=0}^{d/2}%
\frac{\left(  -\frac{d}{2}\right)  _{j}\left(  \frac{d}{2}\right)  _{j}\left(
d\right)  _{j}}{\left(  1+\frac{d}{2}\right)  _{j}\left(  2+\frac{5d}%
{2}\right)  _{j}j!}S_{32}\left(  2+3d,1+d,4+6d;2+3d,\frac{d}{2}-j\right)  \\
&  =\frac{1}{2}\frac{\Gamma\left(  \frac{3d}{2}+1\right)  \Gamma\left(
d+1\right)  }{\Gamma\left(  \frac{5d}{2}+2\right)  }\frac{\Gamma\left(
4+6d\right)  \Gamma\left(  1+2d\right)  }{\Gamma\left(  2+3d\right)
\Gamma\left(  3+5d\right)  }\\
&  \times\sum_{j=0}^{d/2}\frac{\left(  -\frac{d}{2}\right)  _{j}\left(
\frac{d}{2}\right)  _{j}\left(  d\right)  _{j}}{\left(  1+\frac{d}{2}\right)
_{j}\left(  2+\frac{5d}{2}\right)  _{j}j!}\sum_{i=0}^{d/2-j}\frac{\left(
j-\frac{d}{2}\right)  _{i}\left(  2+3d\right)  _{i}\left(  1+d\right)  _{i}%
}{i!\left(  -2d\right)  _{i}\left(  2+\frac{5d}{2}+j\right)  _{i}}.
\end{align*}
The double sum in the last line can be rewritten as%
\[
\sum_{i\geq0,j\geq0}^{i+j\leq d/2}\frac{\left(  -\frac{d}{2}\right)
_{i+j}\left(  \frac{d}{2}\right)  _{j}\left(  d\right)  _{j}\left(
2+3d\right)  _{i}\left(  1+d\right)  _{i}}{\left(  2+\frac{5d}{2}\right)
_{i+j}\left(  1+\frac{d}{2}\right)  _{j}i!j!\left(  -2d\right)  _{i}}.
\]
Collect the (independent of $t$) prefactors (from the product of (\ref{FinalLovasAndai}) and (\ref{FinalLovasAndai2}))%
\[
18\frac{\Gamma\left(  \frac{1}{2}+\frac{d}{2}\right)  ^{3}\Gamma\left(
\frac{7}{6}+\frac{d}{2}\right)  ^{2}\Gamma\left(  \frac{5}{6}+\frac{d}%
{2}\right)  ^{2}\Gamma\left(  1+2d\right)  }{\pi^{7/2}\Gamma\left(  1+\frac
{d}{2}\right)  \Gamma\left(  3+5d\right)  }3456^{d}.
\]
Evaluate for even $d$, $\Gamma\left(  \frac{1}{2}+\frac{d}{2}\right)
=\Gamma\left(  \frac{1}{2}\right)  \left(  \frac{1}{2}\right)  _{d/2}$ (and
$\Gamma\left(  \frac{1}{2}\right)  =\sqrt{\pi}$);
\begin{align*}
\Gamma\left(  \frac{7}{6}+\frac{d}{2}\right)  \Gamma\left(  \frac{5}{6}%
+\frac{d}{2}\right)   &  =\left(  \frac{7}{6}\right)  _{d/2}\left(  \frac
{5}{6}\right)  _{d/2}\Gamma\left(  \frac{7}{6}\right)  \Gamma\left(  \frac
{5}{6}\right)  \\
&  =\left(  \frac{7}{6}\right)  _{d/2}\left(  \frac{5}{6}\right)  _{d/2}%
\frac{1}{6}\Gamma\left(  \frac{1}{6}\right)  \Gamma\left(  \frac{5}{6}\right)
,
\end{align*}
and $\Gamma\left(  \frac{1}{6}\right)  \Gamma\left(  \frac{5}{6}\right)
=\dfrac{\pi}{\sin\frac{\pi}{6}}=2\pi$ (recall $\Gamma\left(  t\right)
\Gamma\left(  1-t\right)  =\dfrac{\pi}{\sin\pi t}$). Put it all together (even
$d$)%
\[
\mathcal{P}_{sep/PPT}\left(  d\right)  =3456^{d}\frac{\left(  \frac{1}{2}\right)  _{d/2}%
^{3}\left(  \frac{7}{6}\right)  _{d/2}^{2}\left(  \frac{5}{6}\right)
_{d/2}^{2}\left(  2d\right)  !}{\left(  \frac{d}{2}\right)  !\left(  3\right)
_{5d}}\sum_{i\geq0,j\geq0}^{i+j\leq d/2}\frac{\left(  -\frac{d}{2}\right)
_{i+j}\left(  \frac{d}{2}\right)  _{j}\left(  d\right)  _{j}\left(
2+3d\right)  _{i}\left(  1+d\right)  _{i}}{\left(  2+\frac{5d}{2}\right)
_{i+j}\left(  1+\frac{d}{2}\right)  _{j}i!j!\left(  -2d\right)  _{i}}.
\]

\section{Remark on Lovas-Andai paper}
It certainly appears that the work of Lovas and Andai \cite{lovasandai}--inspired by that of Milz and Strunz \cite{milzstrunz}--is highly innovative and successful in finding the two-rebit separability function $\tilde{\chi}_1(\varepsilon)$, and verifying the conjecture that the two-rebit Hilbert-Schmidt separability probability is $\frac{29}{64}$. However, in our study of the Lovas-Andai paper, we remain unconvinced by the chain of arguments on page 13 leading to the result $\frac{1}{4}$, and have posted a stack exchange question (\url{https://mathematica.stackexchange.com/q/144277/29989}) in this regard.

\acknowledgments
A number of people provided interesting comments in regard
to questions posted on the mathematics, Mathematica, mathoverflow and physics stack exchanges.
I discussed the two-quaterbit PPT-probability problem--and other items--extensively with
(the always helpful/insightful) Charles Dunkl. Christoph Koutschan, as noted, performed certain calculations laying the
foundation for a formal proof that the Lovas-Andai and ``concise'' formulas
yield the same set of results. Christian Krattenthaler also responded to certain  queries.

\bibliography{main}

\end{document}